\documentclass[letterpaper,twocolumn,10pt]{article}
\usepackage{zhanggroup}

\usepackage{booktabs}
\usepackage{epsfig,endnotes,xspace,url,diagbox}
\usepackage{amssymb, amsmath, amsthm, amsfonts}
\usepackage{enumerate}
\usepackage{varioref}
\usepackage{graphicx}
\usepackage[font=footnotesize]{subfig}
\usepackage{epstopdf}
\usepackage{xcolor}
\usepackage{multirow}
\usepackage{comment}
\usepackage{color}
\usepackage[utf8]{inputenc}
\usepackage{mathtools}
\usepackage{wrapfig}
\usepackage{tabularx}
\usepackage{xspace}
\usepackage[noend, ruled]{algorithm2e}

\DeclareMathOperator*{\argmin}{arg\,min}

\usepackage{cases}
\usepackage{stfloats}
\usepackage{fancyhdr}
\usepackage{cite}
\usepackage{enumitem}

\newcommand{\red}[1]{{\textcolor{red}{#1}}}
\newcommand{\blue}[1]{{\textcolor{blue}{#1}}}
\newcommand{\brown}[1]{{\textcolor{brown}{#1}}}

\newcommand{\mypara}[1]{\smallskip\noindent\textbf{#1.} \xspace}

\definecolor{revision}{RGB}{0,0,0}
\newcommand{\revision}[2]{{\color{revision} #2\xspace}\xspace}
\newcommand{\revisionstart}{\begin{color}{revision}}
\newcommand{\revisionend}{~\!\!\end{color}}

\newcommand{\etal}{\textit{et al.}\xspace}
\newcommand{\ie}{\textit{i.e.}\xspace}
\newcommand{\eg}{\textit{e.g.}\xspace}

\renewcommand{\Pr}[1]{\ensuremath{\mathsf{Pr}\left[#1\right]}\xspace}

\newcommand{\abs}[1]{\ensuremath{\left| #1 \right|}}

\newtheorem{theorem}{Theorem}
\newtheorem{lemma}[theorem]{Lemma}

\newtheorem{prop}[theorem]{Proposition}

\newtheorem{definition}{Definition}

\newcommand{\method}{\ensuremath{\mathsf{PrivSyn}}\xspace}
\newcommand{\synflow}{\ensuremath{\mathsf{MCF}}\xspace}

\newcommand{\priview}{\ensuremath{\mathsf{PriView}}\xspace}
\newcommand{\privbayes}{\ensuremath{\mathsf{PrivBayes}}\xspace}
\newcommand{\jtree}{\ensuremath{\mathsf{JTree}}\xspace}
\newcommand{\mcf}{\ensuremath{\mathsf{MCF}}\xspace}

\newcommand{\gum}{\ensuremath{\mathsf{GUM}}\xspace}
\newcommand{\gumfull}{Graduate Update Method\xspace}
\newcommand{\pgm}{\ensuremath{\mathsf{PGM}}\xspace}
\newcommand{\bsg}{\ensuremath{\mathsf{BSG}}\xspace}
\newcommand{\dq}{\ensuremath{\mathsf{DualQuery}}\xspace}
\newcommand{\maj}{\ensuremath{\mathsf{Majority}}\xspace}
\newcommand{\nonpriv}{\ensuremath{\mathsf{NonPriv}}\xspace}
\newcommand{\privbayesid}{\ensuremath{\mathsf{PrivBayes(InDif)}}\xspace}

\newcommand{\marginal}{\ensuremath{\mathsf{DenseMarg}}\xspace}

\newcommand{\equallap}{\ensuremath{\mathsf{Equal \ Laplace}}\xspace}
\newcommand{\equalgauss}{\ensuremath{\mathsf{Equal \ Gaussian}}\xspace}
\newcommand{\weightgauss}{\ensuremath{\mathsf{Weighted \ Gaussian}}\xspace}
\newcommand{\margselect}{\ensuremath{\mbox{InDif}}\xspace}
\newcommand{\margselectfull}{Independent Difference\xspace}

\renewcommand{\AA}{\mathcal{A}\xspace}
\newcommand{\calN}{\mathcal{N}\xspace}

\newcommand{\doo}{\ensuremath{D_o}\xspace}
\newcommand{\ds}{\ensuremath{D_s}\xspace}

\newcommand{\dr}[3]{\mathrm{D}_{#1}\left(#2\middle\|#3\right)}
\newcommand{\ex}[2]{\underset{#1}{\mathbb{E}}\left[#2\right]}
\newcommand{\zcdp}{zCDP\xspace}
\newcommand{\Lapp}[1]{\ensuremath{\mathcal{L}\left(#1\right)}\xspace}

\begin{document}

\date{}

\author{
Zhikun Zhang\textsuperscript{1,2}\ \ \
Tianhao Wang\textsuperscript{3}\ \ \
Ninghui Li\textsuperscript{3}\\
Jean Honorio\textsuperscript{3}\ \ \
Michael Backes\textsuperscript{2}\ \ \
Shibo He\textsuperscript{1}\ \ \
Jiming Chen\textsuperscript{1}\ \ \
Yang Zhang\textsuperscript{2}\ \ \
\\
\\
\textsuperscript{1}\textit{Zhejiang University} \\
\textsuperscript{2}\textit{CISPA Helmholtz Center for Information Security} \\
\textsuperscript{3}\textit{Purdue University}
}


\title{\method: Differentially Private Data Synthesis}

\maketitle 
\begin{abstract}
In differential privacy (DP), a challenging problem is to generate synthetic datasets that efficiently capture the useful information in the private data.  
The synthetic dataset enables any task to be done without privacy concern and modification to existing algorithms.  
In this paper, we present \method, the first automatic synthetic data generation method that can handle general tabular datasets (with 100 attributes and domain size $>2^{500}$).  
\method is composed of a new method to automatically and privately identify correlations in the data, and a novel method to generate sample data from a dense graphic model.  We extensively evaluate different methods on multiple datasets to demonstrate the performance of our method.
\end{abstract}

\section{Introduction}
\label{sec:intro}

Differential privacy (DP)~\cite{Dwo06} 
has been accepted as the \textit{de facto} notion for protecting privacy.  
Companies and government agencies use DP for privacy-preserving data analysis.  
Uber implements Flex~\cite{johnson2018towards} 
that answers data SQL queries with DP.
LinkedIn builds Pinot~\cite{rogers2020linkedin}, 
a DP platform that enables analysts to gain insights 
about its members' content engagements.  
Within the government, the US census bureau 
plans to publish the 2020 census statistics with DP~\cite{abowd2018us}.

Previous work on DP mostly focuses on designing tailored algorithms 
for specific data analysis tasks.
This paradigm is time consuming, requires a lot of expertise knowledge, and is error-prone. 
For example, many algorithms have been proposed for mining frequent itemset~\cite{lee2014top,li2012privbasis,WXYZGY17}.  
Some of them incorrectly use the Sparse Vector Technique (SVT) and results in non-private
algorithm being incorrectly proven to satisfy DP, see, e.g.,~\cite{Lyu2017vldb} for an analysis
of incorrect usage of SVT.  To answer SQL queries under the constraint of DP, 
the SQL engine needs to be patched~\cite{johnson2018towards}.
For another example, to train a differentially private deep neural network, 
the stochastic gradient descent step is
modified~\cite{abadi2016deep}. 
Moreover, this paradigm does not scale: 
more tasks lead to worse privacy guarantee 
as each task reveals more information about the private data.

One promising solution to address this problem is generating a synthetic dataset 
that is similar to the private dataset while satisfying differential privacy.
As additional data analysis tasks performed on the published dataset are post-processing, they 
can be performed without additional privacy cost.  
Furthermore, existing algorithms for performing data analysis do not need to be modified.

The most promising existing method for private generation of synthetic datasets uses probabilistic graphical models.  
\privbayes~\cite{zhang2017privbayes} uses a Bayesian network.  
It first privately determines the network structure, then obtains noisy marginals for the Conditional Probability Distribution of each node. 
More recently, \pgm, which uses Markov Random Fields, was proposed in \cite{mckenna2019graphical}. 
In 2018, NIST hosted a Differential Privacy Synthetic Data Challenge~\cite{nist-challenge}, 
\pgm achieves the best result.  
Approaches that do not use probabilistic graphical models, such as~\cite{BLR08,hardt2012simple,gaboardi2014dual,vietrinew,zhang2018differentially,beaulieu2019privacy,abay2018privacy,frigerio2019differentially,tantipongpipat2019differentially}, either are computationally inefficient or have poor empirical performance. 

\privbayes and \pgm have two limitations.  First, as a graphical model aims to provide a compact representation of joint probability distributions, it is sparse by design. 
Once a structure is fixed, it imposes conditional independence assumptions that may not exist in the dataset.  
Second, since each model is sparse, the structure is data dependent and finding the right structure is critically important for the utility.  
Bayesian Networks are typically constructed by iterative 
selection using mutual information metrics.  
However, mutual information has high sensitivity, and cannot be estimated accurately under DP.  
\privbayes introduces a low-sensitivity proxy for mutual information, but it is slow (quadratic to the number of users in the dataset) to compute.
In \cite{mckenna2019graphical}, no method for automatically determining the graph structure is provided.  
In the NIST challenge, manually constructed graph networks are used for \pgm.

\mypara{Our Contributions}
In this paper, we propose \method,
for differentially private synthetic data generation.
The first novel contribution is that, instead of using graphical models as the summarization/representation of a dataset,
we propose to \textbf{use a set of large number of low-degree marginals to represent a dataset}.  
For example, in the experiments, given around 100 attributes, our method uses all one-way marginals and 
around 500 two-way marginals.  
A two-way marginal (specified by two attributes) is a frequency distribution table, showing the number of records with 
each possible combination of values for the two attributes.  
At a high level, graphical models can be viewed as a parametric approach to data summarization, and our approach 
can be viewed as a non-parametric one.  The advantage of our approach is that it makes weak assumptions about the conditional independence among attributes, and simply tries to capture correlation relationships that are in the dataset.  

This method is especially attractive under DP for several reasons.  
First, since counting the number of records has a low sensitivity of 1, counting queries can be answered accurately.  
Second, since a marginal issues many counting queries (one for each cell) with the same privacy cost of one counting query, it is arguably the most efficient way to extract information from a dataset under DP.  
Third, using either advanced composition theorem~\cite{DRV10} or zero-Concentrated DP~\cite{bun2016concentrated}, the variance of noises added to each marginal grows only linearly with the number of marginals under the same privacy budget.  
Furthermore, when one attribute is included in multiple marginal, one can use averaging 
to reduce the variance.  
As a result, one can afford to get a large number of marginals with reasonable accuracy. 

There are two main challenges for using a set of marginals for private data synthesis.  
The first challenge is how to select which marginals to use.  Using too many marginals (such as all 2-way marginals) results in higher noises, and slow down data synthesis. 
The second challenge is how to synthesize the dataset given noisy marginals.  

The second contribution is that we propose \textbf{a new method to automatically and privately select the marginals.} 
We first propose a metric \margselect (stands for \margselectfull) that measures the correlation between pairwise attributes.  
\margselect is easy to compute and has low global sensitivity.  
Given \margselect scores, we then propose a greedy algorithm that selects the pairs to form marginals.

The third contribution is that we develop a method that {\bf iteratively update a synthetic dataset to make it match the target set of marginals}.  
When the number of attribute is small enough so that the full contingency table can be stored and manipulated directly, one can use methods such as multiplicative update~\cite{arora2012multiplicative} to do this. 
However, with tens or even over one hundred attributes, it is infeasible to represent the full contingency table.  

The key idea underlying our approach is to view the dataset being synthesized as a proxy of the joint distribution to be estimated, and directly manipulate this dataset.  
In particular, given a set of noisy marginals, we start from a randomly generated dataset where each attribute matches one-way marginal information in the set, and then gradually ``massage'' the synthetic dataset so that its distribution is closer and closer to each pairwise marginal.  
We model this problem as a network flow problem and propose \gumfull (short for \gum), a method to ``massage'' the dataset to be consistent with all the noisy marginals.
We believe that \gum can be of independent interest outside the privacy community. 
Essentially, it can be utilized more broadly as a standalone algorithm and it allows us to generate synthetic dataset from dense graphical models.

To summarize, the main contributions of this paper are:
\begin{itemize}[leftmargin=*]
    \setlength\itemsep{-0.5em}
    \item A simple yet efficient method to capture correlations within the dataset.
    \item A new method to automatically and privately select marginals that capture sufficient correlations.
    \item A data synthesis algorithm \gum that can also be used standalone to handle dense graphical models.
    \item An extensive evaluation which demonstrates the performance improvement of the proposed method on real-world dataset and helps us understand the intuition of different techniques.
\end{itemize}

\mypara{Roadmap}
In Section~\ref{sec:back}, we present background knowledge of DP and composition theorem, and formally define the data synthesis problem.  
We then introduce a general framework of private data synthesis in Section~\ref{sec:framework}.
We present our proposed marginal selection method and data synthesis method in Section~\ref{sec:marg_select} and Section~\ref{sec:synthesizing}, respectively.
Experimental results are presented in Section~\ref{sec:exp}.
We discuss related work in Section~\ref{sec:related} and limitations in Section~\ref{sec:limitation}.
Finally, we provide concluding remarks in Section~\ref{sec:conc}.

\section{Preliminaries}
\label{sec:back}

\subsection{Differential Privacy}
Differential privacy~\cite{DMNS06} is designed for the setting 
where there is a \textbf{trusted data curator}, which gathers data from individual users, processes the data in a way that satisfies DP, and then publishes the results.
Intuitively, the DP notion requires 
that any single element in a dataset has only a limited impact on the output.

\begin{definition}[$(\epsilon,\delta)$-Differential Privacy] \label{def:non-pure-dp}
	An algorithm $\AA$ satisfies $(\epsilon,\delta)$-differential privacy ($(\epsilon,\delta)$-DP), where $\epsilon>0, \delta \geq 0$,
	if and only if for any two neighboring datasets $D$ and $D'$, we have
	\begin{equation*}
	\forall{T\subseteq\! \mathit{Range}(\AA)}:\; \Pr{\AA(D)\in T} \leq e^{\epsilon}\, \Pr{\AA(D')\in T}+\delta,\label{eq:npdp}
	\end{equation*}
	where $\mathit{Range}(\AA)$ denotes the set of all possible outputs of the algorithm $\AA$.
\end{definition}
In this paper we consider two datasets $D$ and $D^{\prime}$ to be \textit{neighbors}, denoted as $D \simeq D'$, if and only if either  $D=D^{\prime} + r$ or $D^{\prime} = D + r$, where $D + r$ denotes the dataset resulted from adding the record $r$ to the dataset $D$. 

\subsection{Gaussian Mechanism}
\label{subsec:gauss}
There are several approaches for designing mechanisms that satisfy $(\epsilon, \delta)$-differential privacy.  
In this paper, we use the Gaussian mechanism.  
The approach computes a function $f$ on the dataset $D$ in a differentially privately way, by adding to $f(D)$ a random noise. 
The magnitude of the noise depends on $\Delta_f$, the \emph{global sensitivity} or the $\ell_2$ sensitivity of $f$.  
Such a mechanism $\AA$ is given below:
$$
\begin{array}{crl}
& \AA(D) & =f(D)+\calN\left(0, \Delta_f^2 \sigma^2 \mathbf{I} \right)
  \\
\mbox{where} &  \Delta_f & = \max\limits_{(D,D') : D \simeq D'} || f(D) - f(D')||_2.
\end{array}
$$
In the above, $\calN (0, \Delta_f^2 \sigma^2 \mathbf{I})$ denotes a multi-dimensional random variable sampled from the normal distribution with mean $0$ and standard deviation  $\Delta_f \sigma$, and $\sigma=\sqrt{2\ln\frac{1.25}{\delta}}/\epsilon$.

\subsection{Composition via Zero Concentrated DP}

For a sequential of $k$ mechanisms $\mathcal{A}_1, \ldots, \mathcal{A}_k$ satisfying $(\epsilon_i, \delta_i)$-DP for $i=1,\ldots, k$ respectively, the basic composition result~\cite{dpbook} shows that the privacy composes linearly, i.e., the sequential composition satisfies $(\sum_i^k\epsilon_i, \sum_i^k\delta_i)$-DP. When $\epsilon_i=\epsilon$ and $\delta_i=\delta$, the advanced composition bound from~\cite{DRV10} states that the composition satisfies ($\epsilon \sqrt{2k\log(1/\delta')}+k\epsilon(e^\epsilon-1)$, $k\delta+\delta'$)-DP. 

To enable more complex algorithms and data analysis task via the composition of multiple differentially private building blocks, zero Concentrated Differential Privacy (\zcdp for short) offers elegant composition properties.
\revision{}{The general idea is to connect $(\epsilon, \delta)$-DP to R\'enyi divergence, and use the useful property of R\'enyi divergence to achieve tighter composition property.
In another word, for fixed privacy budget $\epsilon$ and $\delta$, \zcdp can provide smaller standard deviation for each task compared to other composition techniques.}

Formally, \zcdp is defined as follows:
\begin{definition}[Zero-Concentrated Differential Privacy (\zcdp)~\cite{bun2016concentrated}]
A randomized mechanism $\mathcal{A}$ is $\rho$-zero concentrated differentially private (i.e., $\rho$-\zcdp) if for any two neighboring databases $D$ and $D'$ and all $\alpha \in (1, \infty)$,
\begin{align*}
\mathcal{D}_\alpha (\mathcal{A}(D)||\mathcal{A}(D')) \overset{\Delta}{=} \frac{1}{\alpha -1} \log \big(\mathbb{E}\left[ e^{(\alpha-1)L^{(o)}} \right] \big) \le \rho \alpha 
\end{align*}
Where $\mathcal{D}_\alpha (\mathcal{A}(D)||\mathcal{A}(D'))$ is called $\alpha$-R\'enyi divergence between the distributions of $\mathcal{A}(D)$ and $\mathcal{A}(D')$.
\revision{}{$L^{o}$ is the privacy loss random variable with probability density function $f(x)=\log\frac{\Pr{\mathcal{A}(D)=x}}{\Pr{\mathcal{A}(D')=x}}$.}
\end{definition}

\zcdp has a simple linear composition property~\cite{bun2016concentrated}.
Due to the space limitation, we defer the introduction of \zcdp's composition property to Appendix~\ref{app:zcdp_composition}.

\subsection{Problem Definition}
In this paper, we consider the following problem: \emph{Given a dataset \doo, we want to generate a synthetic dataset \ds that is statistically similar to \doo.}
Generating synthetic dataset \ds allows data analyst to handle arbitrary kinds of data analysis tasks on the same set of released data, which is more general than prior work focusing on optimizing the output for specific tasks (\eg,~\cite{qardaji2014priview,xiao2010differential,abadi2016deep,li2010optimizing}).

More formally, a dataset $D$ is composed of $n$ records each having $d$ attributes.  
The synthetic dataset \ds is said to be similar to \doo if $f(\ds)$ is close to $f(\doo)$ for any function $f$.  
In this paper, we consider three statistical measures: marginal queries, range queries, and classification models.
In particular, a marginal query captures the joint distribution of a subset of attributes.
A range query counts the number of records whose corresponding values are within the given ranges.
Finally, we can also use the synthetic dataset to train classification models and measure the classification accuracy.

\begin{table*}
    \footnotesize
	\centering
	\begin{tabular}{c|c|c|c|c}
		\toprule
		\backslashbox{Method}{Step}	& Marginal Selection & Noise Addition & Post Processing & Data Synthesis \\ 
		\midrule
		\priview~\cite{qardaji2014priview} & Covering design & Equal budget + Laplace & Max-entropy Estimation & -\\
		\privbayes~\cite{zhang2017privbayes} & Bayesian network + Info Gain (EM) & Equal budget + Laplace & - & Sampling \\
		\pgm~\cite{mckenna2019graphical} & - (not dense) & Equal budget + Gaussian & Markov Random Field & Sampling \\
		\method & Optimization + Greedy & Weighted budget + Gaussian & Consistency & \gum\\
		\bottomrule
	\end{tabular}
	\caption{Summary of existing methods on different steps.  The four steps are all new.  Our marginal selection method enables private auto selection of marginals.  \gum enables usage of dense graphical model.}
	\label{tbl:framework_summary}
\end{table*}

\section{A Framework of Private Data Synthesis}
\label{sec:framework}

\revision{}{In this section, we first propose a general framework for generating differentially private synthetic datasets, and then review some existing studies in this framework. 
\method follows this framework and proposes novel techniques for each of the component in the framework.}

To generate the synthetic dataset in a differentially private way, 
one needs to first transform the task 
to estimate a function $f$ with low sensitivity $\Delta_f$.
One straightforward approach is to obtain the noisy full distribution, \ie, the joint distribution of all attributes.
Given the detailed information about the distribution, one can then generate a synthetic dataset by sampling from the distribution.
However, when there are many attributes in the dataset, 
computing or even storing the full distribution requires exponentially large space. 
To overcome this issue, one promising approach is to estimate many low-degree joint distributions, also called marginals, which are distributions of only a subset of attributes.  
More specifically, to generate a synthetic dataset, there are four steps: (1) marginal selection, (2) noise addition, (3) post-processing, and (4) data synthesis. 

The current best-performing approaches on private data synthesis all follow this approach.  
Table~\ref{tbl:framework_summary} summarizes these four steps of existing work and our proposed method.  In what follows, we review these steps in the reverse order.

\subsection{Data Synthesis}

To synthesize a dataset, existing work uses graphical models to model the generation of the data.  
In particular, \privbayes~\cite{zhang2017privbayes} uses a differentially private Bayesian network.  It is a generative model that can be represented by a directed graph.  In the graph, each node $v$ represents an attribute, and each edge from $u$ to $v$ corresponds to $\Pr{v|u}$, the probability of $u$ causing $v$.  As each attribute can take multiple values, all possible $\Pr{v=y|u=x}$ are needed.  When a node $v$ has more than one nodes $U=\{u_1,\ldots,u_k\}$ connected to it, $\Pr{v|U}$ is needed to sample $v$.  Because the causality is a single-direction relationship, the graph cannot contain cycles.
To sample a record, we start from the node with in-degree $0$.  We then traverse the graph to obtain the remaining attributes following the generation order specified by the Bayesian network.

More recently, \cite{mckenna2019graphical} proposed to sample from 
differentially private Markov Random Field (MRF).
Different from Bayesian network, MRF is represented by undirected graphs, and each edge $u, v$ contains the joint distribution $\Pr{v,u}$.  
Moreover, cycles or even cliques are allowed in this model.  
The more complex structures enable capturing higher dimensional correlations, but will make the sampling more challenging.  
In particular, one first merge cliques into nodes and form a tree structure, which is called junction tree.  
The data records can then be sampled from it.  
The main shortcoming of \pgm is that, when the graph is dense, the domain of cliques in the junction tree could be too large to handle.

\subsection{Marginal Selection}
To build a graphical model, joint distributions in the form of $\Pr{v,u}$ are needed (note that conditional distributions $\Pr{v|u}$ can be calculated from joint distributions).
The goal is to capture all the joint distributions.  However, by the composition property of DP, having more marginals leads to more noise in each of them.  We do not want to select too many marginals which leads to excessive noise on each of them.  

\privbayes chooses the marginals by constructing the Bayesian network.  
In particular, it first randomly assigns an attribute as the first node, and then selects other attributes one by one using Exponential Mechanism (EM, refer to Appendix~\ref{app:em}).  
The original Bayesian network uses mutual information as the metric to select the most correlated marginals.  
In the setting of DP, the sensitivity for mutual information is high.  
To reduce sensitivity, the authors of~\cite{zhang2017privbayes} proposed a function that is close to mutual information.

Another method \priview~\cite{qardaji2014priview} uses a data independent method to select the marginals.  
In particular, a minimal set of marginals are selected so that all pairs or triples of attributes are contained in some marginal.  
When some attributes are independent, capturing the relationship among them actually increases the amount of noise.  
This approach cannot scale with the number of attributes $d$.

\mypara{Noise Addition}
Given the marginals, the next step is to add noise to satisfy DP.  The classic approach is to split the privacy budget equally into those marginals and add Laplace noise (refer to Appendix~\ref{app:lap}).

\mypara{Post Processing}
The DP noise introduces inconsistencies, including (1) some estimated probabilities being negative, (2) the estimated probabilities do not sum up to $1$, and (3) two marginals that contain common attributes exist inconsistency.

In \privbayes, negative probabilities are converted to zeros.  In \pgm, consistencies are implicitly handled by the estimation procedure of the Markov Random Field.

\section{Differentially Private Marginal Selection}
\label{sec:marg_select}

In the phase of obtaining marginals, there are two sources of errors.  
One is information loss when some marginals are missed; the other is noise error incurred by DP.  
\privbayes chooses few marginals; as a result, useful correlation information from other marginals is missed.  
On the other hand, \priview is data-independent and tries to cover all the potential correlations; and when there are more than a few dozen attributes, the DP noise becomes too high.

To balance between the two kinds of information loss, we propose an effective algorithm \marginal that is able to choose marginals that capture more useful correlations even under very low privacy budget.

\subsection{Dependency Measurement}
\label{subsec:depend_measure}
To select marginals that capture most of the correlation information, one needs a metric to measure the correlation level.
In Bayesian network, mutual information is used to capture pair-wise correlation.  As the sensitivity for mutual information is high, the authors of~\cite{zhang2017privbayes} proposed a function that can approximate the mutual information.  
However, the function is slow (quadratic to the number of users in the dataset) to compute.

To compute correlation in a simple and efficient way, in this subsection, we propose a metric which we call \margselectfull (\margselect for short).  For any two attributes $a, b$, \margselect calculates the $\ell_1$ distance between the 2-way marginal $\mathsf{M}_{a,b}$ and 2-way marginal generated assuming independence $\mathsf{M}_a \times \mathsf{M}_b$, where a marginal $\mathsf{M}_{A}$ specified by a set of attributes $A$ is a frequency distribution table, showing the frequency with 
each possible combination of values for the attributes, and $\times$ denote the outer product, \ie, 
$
\margselect_{a,b}=|\mathsf{M}_{a,b} - \mathsf{M}_a \times \mathsf{M}_b|_1
$.

\begin{figure}[t]
    \footnotesize
    \subfloat[$1$-way marginal for gender.]{
		\begin{minipage}[t]{0.23\textwidth}
			\centering
			\begin{tabular}[c]{c|c}
				\toprule \label{table:m1}
				$v$					& $\mathsf{M_{\mbox{gender}}}(v)$ \\ 
				\midrule
				$\langle$male,$*\rangle$    & 0.40    \\ 
				$\langle$female,$*\rangle$  & 0.60      \\ 
				\bottomrule
			\end{tabular}
		\end{minipage}
	}
	\subfloat[$1$-way marginal for age.]{
		\begin{minipage}[t]{0.23\textwidth}
			\centering
			\begin{tabular}[c]{c|c}
				\toprule \label{table:m2}
				$v$					& $\mathsf{M_{\mbox{age}}}(v)$ \\ 
				\midrule
				$\langle*$,teenager$\rangle$    & 0.20    \\ 
				$\langle*$,adult   $\rangle$    & 0.30      \\ 
				$\langle*$,elderly$\rangle$     & 0.50      \\ 
				\bottomrule
			\end{tabular}
		\end{minipage}
	} \\
	\subfloat[$2$-way marginal assume indepent]{
	    \label{subfig:marginal_dependence}
		\begin{minipage}[t]{0.23\textwidth}
			\begin{tabular}[c]{c|c}
				\toprule \label{table:f}
				$v$					&  \\ 
				\midrule
				$\langle$male, teenager$\rangle$    & 0.08    \\ 
				$\langle$male, adult$\rangle$       & 0.12      \\ 
				$\langle$male, elderly$\rangle$     & 0.20    \\ 
				$\langle$female, teenager$\rangle$  & 0.12    \\ 
				$\langle$female, adult$\rangle$     & 0.18  \\ 
				$\langle$female, elderly$\rangle$   & 0.30      \\ 
				\bottomrule
			\end{tabular}
		\end{minipage}
	}
	\subfloat[Actual $2$-way marginal]{
	    \label{subfig:marginal_actual}
		\begin{minipage}[t]{0.23\textwidth}
			\begin{tabular}[c]{c|c}
				\toprule \label{table:f}
				$v$					&  \\ 
				\midrule
				$\langle$male, teenager$\rangle$    & 0.10    \\ 
				$\langle$male, adult$\rangle$       & 0.10      \\ 
				$\langle$male, elderly$\rangle$     & 0.20    \\ 
				$\langle$female, teenager$\rangle$  & 0.10    \\ 
				$\langle$female, adult$\rangle$     & 0.20  \\ 
				$\langle$female, elderly$\rangle$   & 0.30      \\ 
				\bottomrule
			\end{tabular}
		\end{minipage}
	} 
    \caption{Example of the calculation of \margselect.}
    \label{fig:example_indif}
\end{figure}

Figure \ref{fig:example_indif} gives an example to illustrate the calculation of \margselect.
The $2$-way marginal in Figure \ref{subfig:marginal_dependence} is directly calculated by the $1$-way marginal of gender and age, without analyzing the dataset; and Figure \ref{subfig:marginal_actual} gives the actual $2$-way marginal.
In this example,
$\margselect = 0.08 \cdot n$, where $n$ is the number of records.
The advantage of using \margselect is that it is easy to compute, and it has low sensitivity in terms of its range, $[0,2n]$:
\begin{lemma}
\label{lemma:l1_sensitivity}
The sensitivity of \margselect metric is $4$: $\Delta_{\margselect}=4$.
\end{lemma}
The proof is deferred to Appendix~\ref{appendix:proofs}.
Given $d$ attributes, we use the Gaussian mechanism to privately obtain all \margselect scores.
\revision{}{To evaluate the impact of noise, one should consider both sensitivity and range of the metrics.
We theoretically and empirically analyze the noise-range ratio of entropy-based metrics and \margselect in Appendix~\ref{app:comparison_metrics}, and show that \margselect has smaller noise-range ratio than entropy-based metrics.}
More specifically, given the overall privacy parameters $(\epsilon, \delta)$, we first compute the parameter $\rho$ using Proposition~\ref{prop:CDPtoDP}.  We then use $\rho'<\rho$ for publishing all the \margselect scores for all $m={d\choose 2}$ pairs of attributes.  In particular, with the composition theory of \zcdp, we can show that publishing all \margselect scores with Gaussian noise $\mathcal{N}(0, 8m/\rho' \mathbf{I})$ satisfies $\rho'$-\zcdp (its proof is also deferred to Appendix~\ref{appendix:proofs}).
\begin{theorem}
\label{thm:dependency_privacy}
Given $d$ attributes, publishing all $m=d(d-1)/2$ \margselect scores with Gaussian noise $\mathcal{N}(0, 8m/\rho' \mathbf{I})$ satisfies $\rho'$-\zcdp.
\end{theorem}

\subsection{Marginal Selection}
\label{subsec:marginal_selection}
Given the dependency scores \margselect, the next step is to choose the pairs with high correlation, and use the Gaussian mechanism to publish marginals on those pairs.  
In this process, there are two error sources.
One is the noise error introduced by the Gaussian noise; the other is the dependency error when some of the marginals are not selected.
If we choose to publish all $2$-way marginals, the noise error will be high and there is no dependency error; when we skip some marginals, the error for those marginals will be dominated by the dependency error.

\mypara{Problem Formulation}
Given $m$ pairs of attributes, each pair $i$ is associated with an indicator variable $x_i$ that equals $1$ if pair $i$ is selected, and $0$ otherwise.  
Define $\psi_i$ as the noise error introduced by the Gaussian noise and $\phi_i$ as its dependency error.
The marginal selection problem is formulated as the following optimization problem:
\begin{align*}
\mathsf{minimize~} & \sum_{i=1}^m \left[\psi_i x_i + \phi_i(1 - x_i)\right] \\
\mathsf{subject~to~} & x_i \in \{0, 1\}
\end{align*}

\revision{}{Notice that the dependency error $\phi_i$ has positive correlation with $\margselect_i$, \ie, larger $\margselect_i$ incurs larger $\phi_i$.
Thus, we approximate $\phi_i$ as $\margselect_i + \mathcal{N}(0, m^2\rho'^2 \mathbf{I})$, and it is fixed in the optimization problem.
}

The noise error $\psi_i$ is dependent on the privacy budget $\rho_i$ allocated to the pair $i$.  In particular, we first show that given the true marginal $\mathsf{M}_i$, we add Gaussian noise with scale $1/\rho_i$ to achieve $\rho_i$-\zcdp.  
\begin{theorem}
\label{thm:marginal_gauss}
(1) The marginal $\mathsf{M}$ has sensitivity $\Delta_{\mathsf{M}}=1$; (2) Publishing marginal $\mathsf{M}$ with noise $\mathcal{N}(0, 1/2\rho \mathbf{I})$ satisfies $\rho$-\zcdp.
\end{theorem}
The proof of Theorem~\ref{thm:marginal_gauss} is deferred to Appendix~\ref{appendix:proofs}.  

To make $\psi_i$ and $\phi_i$ comparable, we use the expected $\ell_1$ error of the Gaussian noise on marginal $i$.  
That is, if the marginal size is $c_i$, after adding Gaussian noise with scale $\sigma_i$, we expect to see the $\ell_1$ error of $c_i\sqrt{\frac{2}{\pi}}\sigma_i$.  
Thus, with privacy budget $\rho_i$, $\psi_i=c_i\sqrt{\frac{1}{\pi\rho_i}}$.
The optimization problem is transformed to:
\setlength{\belowdisplayskip}{0pt} \setlength{\belowdisplayshortskip}{0pt}
\setlength{\abovedisplayskip}{0pt} \setlength{\abovedisplayshortskip}{0pt}
\begin{align*}
\mathsf{minimize~} & \sum_{i=1}^m \left[c_i\sqrt{\frac{1}{\pi\rho_i}} x_i + \phi_i(1 - x_i)\right] \\
\mathsf{subject~to~} & x_i \in \{0, 1\}\\
                    & \sum x_i\rho_i=\rho
\end{align*}

\mypara{Optimal Privacy Budget Allocation}
We first assume the pairs are selected (\ie, variables of $x_i$ are determined), and we want to allocate different privacy budget to different marginals to minimize the overall noise error.
In this case, the optimization problem can be rewritten as:
\setlength{\belowdisplayskip}{0pt} \setlength{\belowdisplayshortskip}{0pt}
\setlength{\abovedisplayskip}{0pt} \setlength{\abovedisplayshortskip}{0pt}
\begin{align*}
\mathsf{minimize~} & \sum_{i:x_i=1} c_i \sqrt{\frac{1}{\rho_i}} \\
\mathsf{subject~to~} & \sum_{i:x_i=1} \rho_i=\rho
\end{align*}
For this problem, we can construct the Lagrangian function $\mathcal{L}=\sum_{i} \frac{c_i}{\sqrt{\rho_i}} + \mu \cdot \left( \sum_{i}\rho_i - \rho \right)$.
By taking partial derivative of $\mathcal{L}$ for each of $\rho_i$, we have $\rho_i = \left( \frac{2\mu}{c_i} \right)^{-2/3}$.
The value of $\mu$ can be solved by equation $\sum_{i} \rho_i=\rho$.
As a result, $\mu = \frac{1}{2} \cdot \left( \frac{\rho}{\sum_{i} c_i^{2/3}} \right)^{-3/2}$, and we have
\setlength{\belowdisplayskip}{0pt} \setlength{\belowdisplayshortskip}{0pt}
\setlength{\abovedisplayskip}{0pt} \setlength{\abovedisplayshortskip}{0pt}
\begin{align}
\label{eq:optimal_allocation}
\rho_i = \frac{c_i^{2/3}}{\sum_{j} c_j^{2/3}} \cdot \rho
\end{align}
That is, allocating privacy budget proportional to the $\frac{2}{3}$ power of the number of cells achieves the minimum overall noise error.

\mypara{A Greedy Algorithm to Select Pairs}
We propose a greedy algorithm to select pairs of attributes, as shown in Algorithm \ref{alg:marginal_selection}. 
Given the \margselect scores of all pairs of attributes $\langle\phi_i\rangle$, size of all marginals $\langle c_i\rangle$, and the total privacy budget $\rho$, the goal is to determine $x_i$ for each $i\in\{1,\ldots,m\}$, or equivalently, output a set of pairs $X=\{i: x_i=1\}$ that minimize the overall error.
We handle this problem by iteratively including marginals that give the maximal utility improvement.  
In particular, in each iteration $t$, we select one marginal that brings the maximum improvement to the overall error.
More specifically, we consider each marginal $i$ that is not yet included in $X$ (\ie, $i\in\bar{X}$, where $\bar{X}=\{1,\ldots,m\}\setminus X$): In Line~\ref{ln:allocate}, we allocate the optimal privacy budget $\rho_i$ according to Equation~\ref{eq:optimal_allocation}. 
We then calculate the error in Line~\ref{ln:calculate}, and select one with maximum utility improvement (in Line~\ref{ln:select}).  
After the marginal is selected, we then include it in $X$.
The algorithm terminates when the overall error no longer improves.
The algorithm is guaranteed to terminate since the noise error would gradually increase when more marginals are selected.
When the noise error is larger than any of the remaining dependency error, the algorithm terminates.

\begin{algorithm}[!h]
\footnotesize
\SetCommentSty{small}
\LinesNumbered
\caption{Marginal Selection Algorithm}
\label{alg:marginal_selection}

\KwIn{Number of pairs $m$, privacy budget $\rho$, dependency error $\langle\phi_i\rangle$, marginal size $\langle c_i\rangle$;}
\KwOut{Selected marginal set $X$;}

$X \gets \varnothing$; 
$t \gets 0$;
$E_0 \gets \sum_{i \in \bar{X}} \phi_i$;

\While{True}
{
    \ForEach{marginal $i \in \bar{X}$}
    {
        Allocate $\rho$ to marginals $j \in X \cup \{i\}$;\label{ln:allocate}\\
        $E_t(i) = \sum_{j \in X \cup \{i\}} c_j\sqrt{\frac{1}{\pi\rho_j}} + \sum_{j \in \bar{X} \setminus \{i\}} \phi_j$\label{ln:calculate};
    }

    $\ell \gets \argmin_{i \in \bar{X}} E_t(i)$\label{ln:select};
    
    $E_t \gets E_t(\ell)$;
    
    \If{$E_t \geq E_{t-1}$}
    {
        \textbf{Break}
    }
    $X \gets X \cup \{l\}$; \\
    $t\gets t+1$;
}

\end{algorithm}
\vspace{-1em}

\subsection{Post Processing}
The purpose of post processing is to ensure the noisy marginals are consistent.
By handling such inconsistencies, we avoid impossible cases and ensure there exists a solution (\ie, a synthetic dataset) that satisfies all the noisy marginals.  
For the case when multiple marginals contain the same set of attributes, and their estimations on the shared attributes do not agree, we use the weighted average method~\cite{DWH+11,qardaji2014priview}.  
Note that~\cite{DWH+11,qardaji2014priview} both assume the privacy budget is evenly distributed.  
We extend it to the uneven case.  

\mypara{Consistency under Uneven Privacy Budget Allocation}
When different marginals have some attributes in common, those attributes are actually estimated multiple times.  Utility will increase if these estimates are utilized together.  
For example, when some marginals are estimated twice, the mean of the estimates is actually more accurate than each of them.
More formally, assume a set of attributes $A$ is shared by $s$ marginals $\mathsf{M}_1, \mathsf{M}_2, \ldots, \mathsf{M}_s$, where $A=\mathsf{M}_1 \cap \ldots \cap \mathsf{M}_s$.  
We can obtain $s$ estimates of ${A}$ by summing from cells in each of the marginals.

In \cite{qardaji2014priview}, the authors proposed an optimal method to determine the distribution of the weights when privacy budget is evenly distributed among marginals.
The main idea is to take the weighted average of estimates from all marginals in order to minimize the variance of marginals on $A$.
\revision{}{We adopt the weighted average technique, and extend it to hand the case where privacy budget is unevenly allocated.}
In particular, we allocate a weight $w_i$ for each marginal $i$.  
The variance of the weighted average can be represented by $\sum_{i} w_i^2\cdot \frac{g_i}{\rho_i}$, where $\rho_i$ is the privacy budget and $g_i$ is the number of cells that contribute to one cell of the marginal on $A$.  
Here the Gaussian variance is $1/\rho_i$.  
By summing up $g_i$ cells, and multiplying the result by $w_i$, we have the overall variance $w_i^2\frac{g_i}{\rho_i}$.  
The weights should add up to $1$.  
More formally, we have the following optimization problem:

\setlength{\belowdisplayskip}{0pt} \setlength{\belowdisplayshortskip}{0pt}
\setlength{\abovedisplayskip}{0pt} \setlength{\abovedisplayshortskip}{0pt}
\begin{align*}
\mathsf{minimize } & \sum_{i} w_i^2\cdot \frac{g_i}{\rho_i} \\
\mathsf{subject~to } & \sum_{i}w_i=1
\end{align*}

By constructing the Lagrangian function and following the same derivative procedure as we did for obtaining optimal $\rho_i$ (Equation~\eqref{eq:optimal_allocation}), we have $w_i=\frac{\rho_i/g_i}{\sum_i \rho_i/g_i}$ is the optimal strategy.

\mypara{Overall Consistency}
In addition to the inconsistency among marginals, some noisy marginals may contain invalid distributions (\ie, some probability estimations are negative, and the sum does not equal to $1$).  
Given the invalid distribution, it is known that projecting it to a valid one with minimal $\ell_2$ distance achieves the maximal likelihood estimation.  
This is discovered in different settings (\eg, \cite{lee2015maximum,wang2019consistent,aistats:Bassily19}); and there exists efficient algorithm for this projection.  

The challenge emerges when we need to handle the two inconsistencies simultaneously, one operation invalidate the consistency established in another one.  
We iterate the two operations multiple times to ensure both consistency constraints are satisfied.

\section{Synthetic Data Generation}
\label{sec:synthesizing}
Given a set of noisy marginals, the data synthesis step generates a new dataset \ds so that its distribution is consistent with the noisy marginals.  Existing methods~\cite{zhang2017privbayes,mckenna2019graphical} put these marginals into a graphical model, and use the sampling algorithm to generate the synthetic dataset.  As each record is sampled using the marginals, the synthetic dataset distribution is naturally consistent with the distribution.

The drawback of this approach is that when the graph is dense, existing algorithms do not work.  
To overcome this issue, we use an alternative approach.  
Instead of sampling the dataset using the marginals, we initialize a random dataset and update its records to make it consistent with the marginals.

\subsection{Strawman Method: Min-Cost Flow (\synflow)}
Given the randomly initiated dataset \ds, for each noisy marginal, we update \ds to make it consistent with the marginal.
A marginal specified by a set of attributes is a frequency distribution table for each possible combination of values for the attributes.
The update procedure can be modeled as a graph flow problem.
In particular, given a marginal, a bipartite graph is constructed.  
Its left side represents the current distribution on \ds; and the right side is for the target distribution specified by the marginal.  
Each node corresponds to one cell in the marginal and is associated with a number.  Figure~\ref{fig:mcfp} demonstrates an example of this flow graph.
Now in order to change \ds to make it consistent with the marginal, we change records in \ds.

\begin{figure}[ht]
    \centering

    \includegraphics[width=0.35\textwidth]{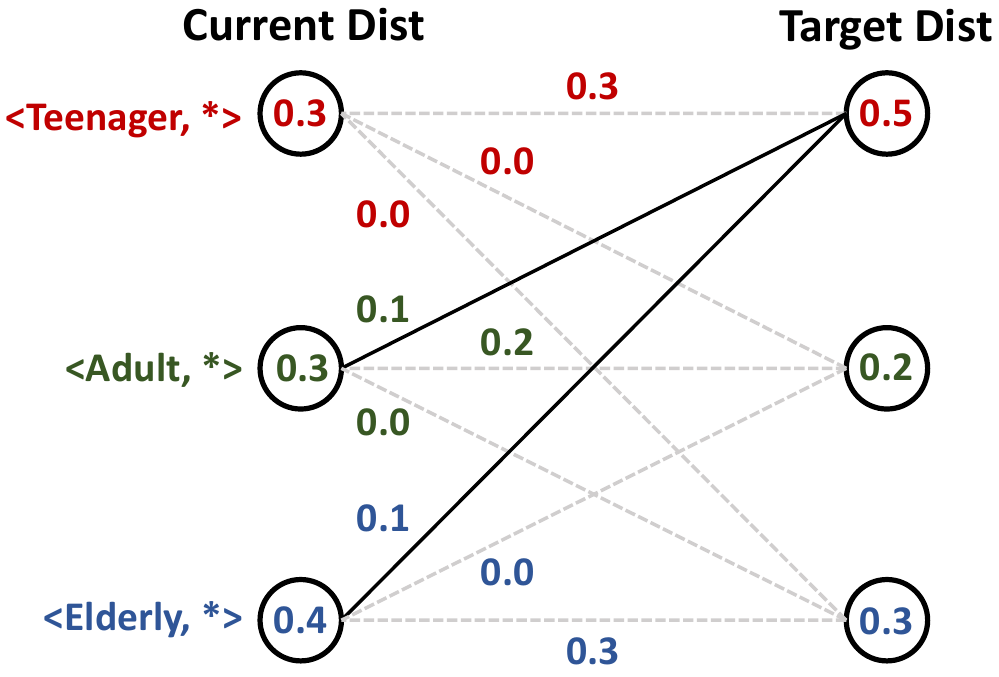}

    \caption{Running example of \synflow.  
    The left nodes represent current distribution from \ds; and the right nodes give the target distribution specified by the noisy marginal.  
    The min-cost flow is to move $0.1$ from adult to teenager, and $0.1$ from elderly to teenager.  
    To change the distribution, we find matching records from \ds and change their corresponding attributes.} 
    \label{fig:mcfp}
\end{figure}

The \synflow method enforces a min-cost flow in the graph and updates \ds by changing the values of the records on the flow.
For example, in Figure~\ref{fig:mcfp}, there are two changes to \ds.  
First, one third of the adults needs to be changed to teenagers.  
Note that we change only the related attribute and keep the other attributes the same.  Second, one fourth of the elderly are changed to teenager.  
We iterate over all the noisy marginals and repeat the process multiple times until the amount of changes is small.  
The intuition of using min-cost flow is that, the update operations make the minimal changes to \ds, and by changing the dataset in this minimal way, the consistency already established in \ds (with previous marginals) can be maintained.
The min-cost flow can be solved by the off-the-shelf linear programming solver, e.g.,~\cite{ahuja1988network}.

When all marginals are examined, we randomly shuffle the whole dataset \ds.  
Since the modifying procedure would invalidate the consistency established from previous marginals, \synflow needs to iterate multiple times to ensure that \ds is almost consistent with all marginals.

\begin{figure*}[!htpb]
    \footnotesize
    \centering
	\subfloat[Dataset before updating.]{
		\begin{tabular}[c]{c|c c c}
			\toprule \label{table:dataset}
			& Income  & Gender    & Age  \\ 
			\midrule
			$v_1$     & high          & male          & teenager     \\ 
			$v_2$     & high          & male          & adult     \\ 
			$v_3$     & \brown{high}  & \brown{male}  & \brown{adult} \\ 
			$v_4$     & \brown{high}  & \brown{male}  & \brown{teenager} \\ 
			$v_5$     & \blue{high}   & \blue{female} & \blue{elderly}  \\ 
			\bottomrule
		\end{tabular}
	}
	\subfloat[Marginal table for $\{$Income, Gender$\}$,
	where red and blue stands for over-counted and under-counted cells, respectively.
	]{
		\begin{tabular}[c]{c|c c}
			\toprule \label{table:m1}
			$v$		&  $\mathsf{S_{\{\mbox{I,G}\}}}(v)$ & $\mathsf{T_{\{\mbox{I,G}\}}}(v)$ \\ 
			\midrule
			$\langle$low, male,$*\rangle$    & 0.0   & 0.0    \\ 
			$\langle$low, female,$*\rangle$  & 0.0   & 0.0   \\
			$\langle$high, male,$*\rangle$   & \red{0.8}   & \red{0.2}    \\ 
			$\langle$high, female,$*\rangle$ & \blue{0.2}  & \blue{0.8}   \\ 
			\bottomrule
		\end{tabular}
	} 
	\subfloat[Dataset after updating.]{
		\begin{tabular}[c]{c|c c c}
			\toprule \label{table:dataset}
			& Income  & Gender    & Age   \\ 
			\midrule
			$v_1$   & high           & male      & teenager     \\ 
			$v_2$   & high           & male      & adult     \\ 
			$v_3$   & \blue{high}   & \blue{female}  & \blue{elderly} \\ 
			$v_4$   & \blue{high}   & \blue{female}  & \brown{teenager} \\ 
			$v_5$   & \blue{high}   & \blue{female}  & \blue{elderly}  \\ 
			\bottomrule
		\end{tabular}
	}
    \caption{Example of the synthesized dataset before and after updating procedure.  
    In (a), blue stands for the records to be added, and brown stands for the records to be changed.
    In (c), $v_4$ only changes income and gender attributes, while $v_3$ changes the whole record which is duplicated from $v_5$.
    Notice that in this example, we have $\alpha=2.0, \beta=0.5$ and the marginal distribution in (c) do not completely match $\mathsf{T_{\{\mbox{I,G}\}}}(v)$ of $[0.0, 0.0, 0.2, 0.8]$; instead, it becomes $[0.0, 0.0, 0.4, 0.6]$.
    }
    \label{fig:update_example}
\end{figure*}

\subsection{Gradually Update Method (\gum)}
\label{subsec:gum}
Empirically, we find that the convergence performance of \mcf is not good (we will demonstrate it via experiment in Section~\ref{sec:exp}).  
We believe that this is because \mcf always changes \ds to make it completely consistent with the current marginal in each step.
Doing this reduces the error of the target marginal close to zero, but increases the errors for other marginals to a large value.

To handle this issue, we borrow the idea of multiplicative  update~\cite{arora2012multiplicative} and propose a new approach that \underline{G}radually \underline{U}pdate \ds based on the \underline{M}arginals; and we call it \gum.  
\gum also adopts the flow graph introduced by \mcf, but differs from \mcf in two ways:  First, \gum does not make \ds fully consistent with the given marginal in each step.  
Instead, it changes \ds in a multiplicative way, so that if the original frequency in a cell is large, then the change to it will be more.  
In particular, we set a parameter $\alpha$, so that for cells that have values are lower than expected (according to the target marginal), we add at most $\alpha$ times of records, i.e., $\min\left\{ n^t - n^s, \alpha n^s \right\}$~\footnote{Notice that $\alpha$ could be greater than $1$ since $n^s < n^t$. 
In the experiments, we always set $\alpha$ to be less than $1$ to achieve better convergence performance.}, where $n^t$ is the number in the marginal and $n^s$ is the number from \ds.
On the other hand, for cells with values higher than expected, we will reduce $\min\left\{ n^s - n^t, \beta n^s \right\}$ records that satisfy it. 
As the total number of record is fixed, given $\alpha$, $\beta$ can be calculated.

Figure~\ref{fig:update_example} gives a running example.  
Before updating, we have 4 out of 5 records have the combination $\langle high, male \rangle$, and 1 record has $\langle high, female \rangle$.  
To get closer to the target marginal of 0.2 and 0.8 for these two cells,
we want to change 2 of the $\langle high, male \rangle$ records to be $\langle high, female \rangle$.
In this example, we have $\alpha=2.0, \beta=0.5$~\footnote{\revision{}{We have $\alpha=\frac{n^t-n^s}{n^s}$ for under-counted cells and $\beta=\frac{n^s-n^t}{n^s}$ for over-counted cells.
The number of records for under-counted cell $\langle$high, female,$*\rangle$ increase from $1$ to $3$; thus $\alpha=\frac{3-1}{1}=2$.
The number of records for over-counted cell $\langle$high, male,$*\rangle$ decrease from $4$ to $2$; thus $\beta=\frac{4-2}{4}=0.5$.}}
and do not completely match the target marginal of 0.2 and 0.8.
To this end, one approach is to simply change the Gender attribute value from male to female in these two records as in \mcf.  
We call this a {\bf Replace} operation.  Replacing will affect the joint distribution of other marginals, such as $\{Gender, Age\}$.
An alternative is to discard an existing $\langle high, male \rangle$ record, and {\bf Duplicate} an existing $\langle high, female \rangle$ record (such as $v_5$ in the example).  
Duplicating an existing record help preserve joint distributions between the changed attributes and attributes not in the marginal.  
However, Duplication will not introduce new records that can better reflect the overall joint distribution.  
In particular, if there is no record that currently has the combination $\langle high, female, elderly \rangle$, duplication cannot be used.  

Therefore, we need to use a combination of Replacement and Duplication (which is the case in Figure~\ref{fig:update_example}).
Furthermore, once the synthesized dataset is getting close to the distribution, we would prefer Duplication to Replacement, since at that time there should be enough records to reflect the distribution and Replacement disrupts the joint distribution between attributes in a marginal and those not in it.  
In Appendix~\ref{app:comparison_update}, we empirically compare different record updating strategies and validate that introducing the Duplication operation can effectively improve the convergence performance.

\subsection{Improving the Convergence}
\label{subsec:improve_convergence}
Given the general data synthesize method, we have several optimizations to improve its utility and performance.
First, to bootstrap the synthesizing procedure, we require each attribute of \ds follows the 1-way noisy marginals when we initialize a random dataset \ds.  

\mypara{Gradually Decreasing $\alpha$}
The update rate $\alpha$ should be smaller with the iterations to make the result converge.
From the machine learning perspective, gradually decreasing $\alpha$ can effectively improve the convergence performance.  
There are some common practices~\cite{decay} of setting $\alpha$.

\begin{itemize}[leftmargin=*]
    \setlength\itemsep{-0.5em}
    \item Step decay: $\alpha = \alpha_0 \cdot k^{\lfloor \frac{t}{s} \rfloor}$, where $\alpha_0$ is the initial value, $t$ is the iteration number, $k$ is the decay rate, and $s$ is the step size (decrease $\alpha$ every $s$ iterations).
    The main idea is to reduce $\alpha$ by some factor every few iterations.  
    
    \item Exponential decay: $\alpha = \alpha_0 \cdot e^{-kt}$, where $k$ is a hyperparameter.
    This exponentially decrease $\alpha$ in each iteration. 
    
    \item Linear decay: $\alpha = \frac{\alpha_0}{1+kt}$.
    
    \item Square root decay: $\alpha = \frac{\alpha_0}{\sqrt{1+kt}}$.
    
\end{itemize}

We empirically evaluate the performance of different decay algorithms in Appendix~\ref{app:comparison_decay}, and find that step decay is preferable in all settings.
The step decay algorithm is also widely used to update the step size in the training of deep neural networks~\cite{krizhevsky2012imagenet}.

\mypara{Attribute Appending}
The selected marginals $\mathcal{X}$ output by Algorithm~\ref{alg:marginal_combine} can be represented by a graph $\mathcal{G}$.
We notice that some nodes have degree $1$, which means the corresponding attributes are included in exactly one marginal.
For these attributes, it is not necessary to involve them in the updating procedure.  Instead, we could append them to the synthetic dataset \ds after other attributes are synthesized.
In particular, we identify nodes from $\mathcal{G}$ with degree $1$.  We then remove marginals associated with these nodes from $\mathcal{X}$.  The rest of the noisy marginals are feed into \gum to generate the synthetic data but with some attributes missing.  For each of these missed attributes, we sample a smaller dataset \ds' with only one attribute, and we concatenate \ds' to \ds using the marginal associated with this attribute if there is such a marginal; otherwise, we can just shuffle \ds' and concatenate it to \ds.
Note that this is a one time operation after \gum is done. No synthesizing operation is needed after this step.

\mypara{Separate and Join}
We observe that, when the privacy budget is low, the number of selected marginals is relatively small, and the dependency graph is in the form of several disjoint subgraphs.
In this case, we can apply \gum to each subgraph and then join the corresponding attributes.
The benefit of Separate and Join technique is that, the convergence performance of marginals in one subgraph would not be affected by marginals in other subgraph, which would improve the overall convergence performance.

\mypara{Filter and Combine Low-count Values}
If some attributes have many possible values while most of them have low counts or do not appear in the dataset.
Directly using these attributes to obtain pairwise marginals may introduce too much noise.
To address this issue, we propose to filter and combine the low-count values.
The idea is to spend a portion of privacy budget to obtain the noisy one-way marginals.
After that, we keep the values that have count above a threshold $\theta$.  
For the values that are below $\theta$, we add them up, if the total is below $\theta$, we assign 0 to all these values.  
If their total is above $\theta$, then we create a new value to represent all values that have low counts. 
After synthesizing the dataset, this new value is replaced by the values it represents using the noisy one-way marginal. 
The threshold is set as $\theta = 3 \sigma$, where $\sigma$ is the standard deviation for Gaussian noises added to the one-way marginals.

\vspace{-1em}
\subsection{Putting Things Together: \method}
Algorithm~\ref{alg:dpsyn} illustrates the overall workflow of \method.
We split the total privacy budget into three parts.
The first part is used for publishing all $1$-way marginals, intending to filter and combine the values with low count or do not exist.
The second part is used to differentially privately select marginals.
The marginal selection method \marginal consists of two components, \ie, $2$-way marginal selection (Algorithm \ref{alg:marginal_selection}) and marginal combine (Algorithm \ref{alg:marginal_combine}).
The third part is used to obtain the noisy combined marginals.
After obtaining the noisy combined marginals, we can use them to construct synthetic dataset \ds without consuming privacy budget, since this is a post processing procedure.

\begin{algorithm}[th]
\footnotesize
\SetCommentSty{small}
\LinesNumbered
\caption{\method}
\label{alg:dpsyn}

\KwIn{Private dataset \doo, privacy budget $\rho$;}
\KwOut{Synthetic dataset \ds;}

Publish 1-way marginals using GM with $\rho_1=0.1\rho$;

Filter values with estimates smaller than $3\sigma$;

Select 2-way marginals with Algorithm~\ref{alg:marginal_selection} and $\rho_2=0.1\rho$;

Combine marginals using Algorithm~\ref{alg:marginal_combine};

Publish combined marginals using GM with $\rho_3=0.8\rho$;

Make noisy marginals consistent;

Construct \ds using \gum;

\end{algorithm}

\vspace{-2em}
\section{Evaluation}
\label{sec:exp}
In this section, we first conduct a high-level end-to-end experiment to illustrate the effectiveness of \method.
Then, we evaluate the effectiveness of each step of \method by fixing other steps.
As a highlight, our method consistently achieves better performance than the state-of-the-art in all steps.

\subsection{Experimental Setup}

\mypara{Datasets}
We run experiments on the following four datasets. 
\vspace{-2em}
\begin{itemize}[leftmargin=*]
    \setlength\itemsep{-0.5em}
    
    \item \revision{}{\mypara{UCI Adult~\cite{AN10}}
    This is a widely used dataset for classification from the UCI machine learning repository.}
    
	\item \mypara{US Accident~\cite{moosavi2019accident}}
	This is a countrywide traffic accident dataset, which covers $49$ states of the United States.
	
	\item \mypara{Loan~\cite{kaggle-loan}}
	This dataset contains loan data in lending club issued from 2007 to 2015.
	
	\item \mypara{Colorado~\cite{nist-challenge}}
	This is the census dataset of Colorado State in 1940. 
	This dataset is used in the final round of the NIST challenge~\cite{nist-challenge}.
\end{itemize}

The detailed information about the datasets are listed in Table~\ref{tab:datasets},
where the label column stands for the label used in the classification task.

\begin{table}[!h]
    \footnotesize
	\centering
	\begin{tabular}{c|c|c|c|c}
		\toprule
		Dataset	& Records & Attributes & Domain & Label \\ 
		
		\midrule
		Adult & $48,000$ & $15$ & $6\cdot 10^{17}$ & salary \\
		US Accident & $600,000$ & $30$ & $3\cdot 10^{39}$ & Severity \\
		Loan & $600,000$ & $81$ & $4\cdot 10^{136}$ & home\_ownership \\
		Colorado & $662,000$ & $97$ & $5\cdot 10^{162}$ & INCNONWG \\
		
		\bottomrule
	\end{tabular}
	\caption{Summary of datasets used in our experiments.}
	\label{tab:datasets}
\end{table}
\vspace{-1em}

\mypara{Tasks and Metrics}
We evaluate the statistical performance of the synthesized datasets on three data analysis tasks.
For each data analysis task, we adopt its commonly used metric to measure the performance.

\vspace{-1em}
\begin{itemize}[leftmargin=*]
    \setlength\itemsep{-0.5em}
    \item \mypara{Marginal Release}
    We compute all the $2$-way marginals and use the average $\ell_1$ error to measure the performance.
    
    \item \mypara{Range Query}
    We randomly sample $1000$ range queries, each contains $3$ attributes.
    We use the average $\ell_1$ error to measure the performance.
    In particular, we calculate $\frac{1}{|Q|}\sum_{q_i \in Q}|c_i - \hat{c_i}|$, where $Q$ is the set of randomly sampled queries, $c_i$ and $\hat{c_i}$ are the ratio of records that fall in the range of query $q_i$ in the original dataset and synthesized dataset, respectively.

    \item \mypara{Classification}
    We use the synthesized dataset to train an SVM classification model, and use misclassification rate to measure the performance.
\end{itemize}

\vspace{-1em}
\mypara{Competitors}
We compare each component of \method with a series of other methods, respectively.
\vspace{-1em}
\begin{itemize}[leftmargin=*]
    \setlength\itemsep{-0.5em}
    \item \mypara{Marginal Selection Methods}
    We compare our proposed \marginal method (Algorithm~\ref{alg:marginal_selection}) with \privbayes.  
    The computational complexity of dependency in original \privbayes method is too high.  
    Thus, we replace the dependency calculation part of \privbayes by our proposed \margselect metric, which we call \privbayesid.
    For Colorado dataset, the \pgm team open sourced a set of manually selected marginals in the NIST challenge~\cite{nist-challenge}, which serves as an alternative competitor. 
    
    \item \mypara{Noise Addition Methods} 
    We compare our proposed \weightgauss method with \equallap and \equalgauss methods.
    Both Gaussian methods use \zcdp to compose, and the Laplace mechanism use the naive composition, \ie, evenly allocate $\epsilon$ for each marginal. 
    
    \item \mypara{Data Synthesis Methods}
    We compare our proposed \method method with \privbayes, \pgm.
    Both \privbayes and \pgm use the selected marginals to estimate a graphical model, and sample synthetic records from it.
    Notice that we have two versions of synthesis methods for \method, \ie, \mcf and \gum.
\end{itemize}

\begin{figure*}[!htpb]
    \centering

    \subfloat[pair-wise marginal]{\includegraphics[width=0.3\textwidth]{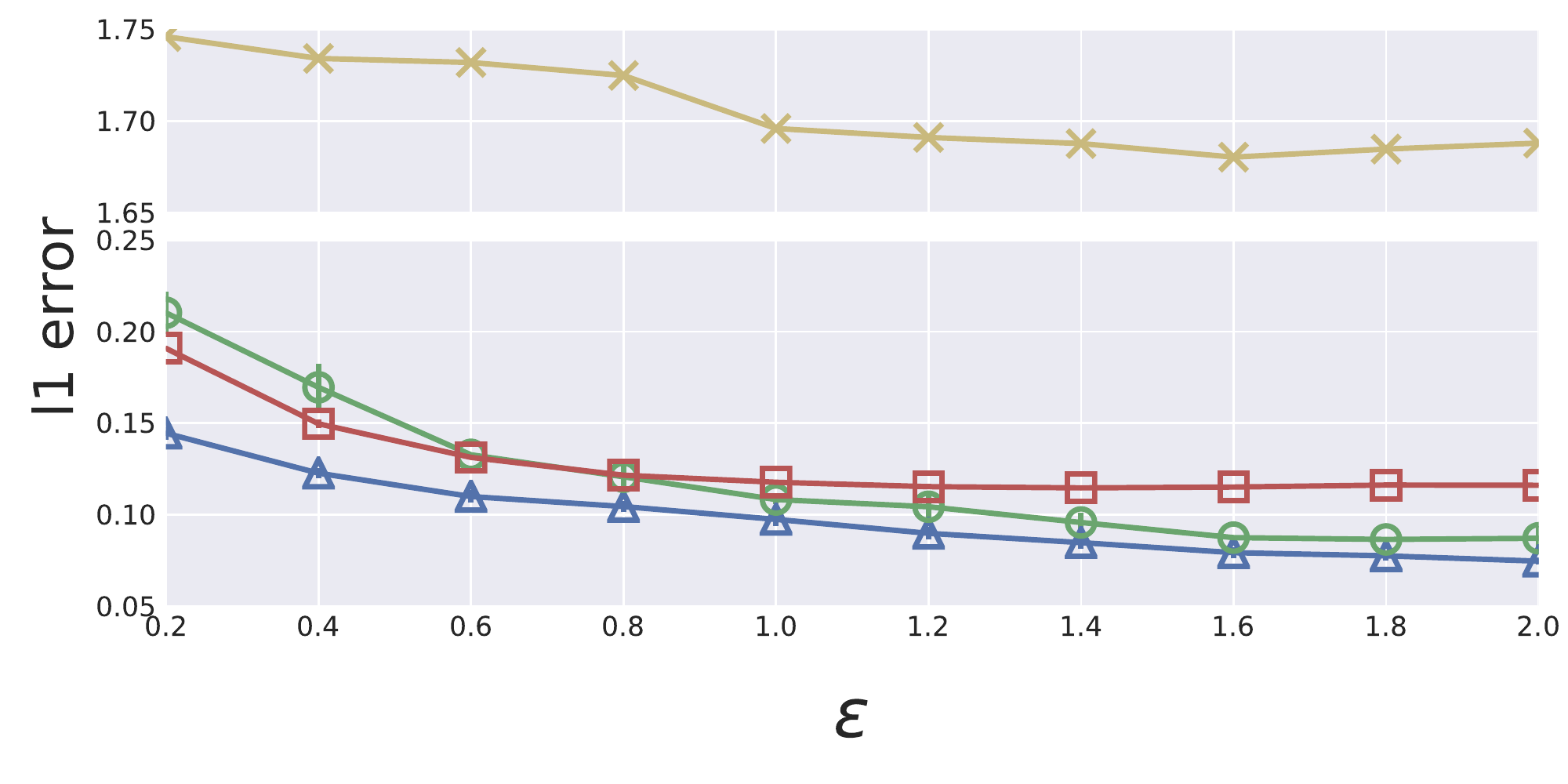}}
    \subfloat[range query]{\includegraphics[width=0.3\textwidth]{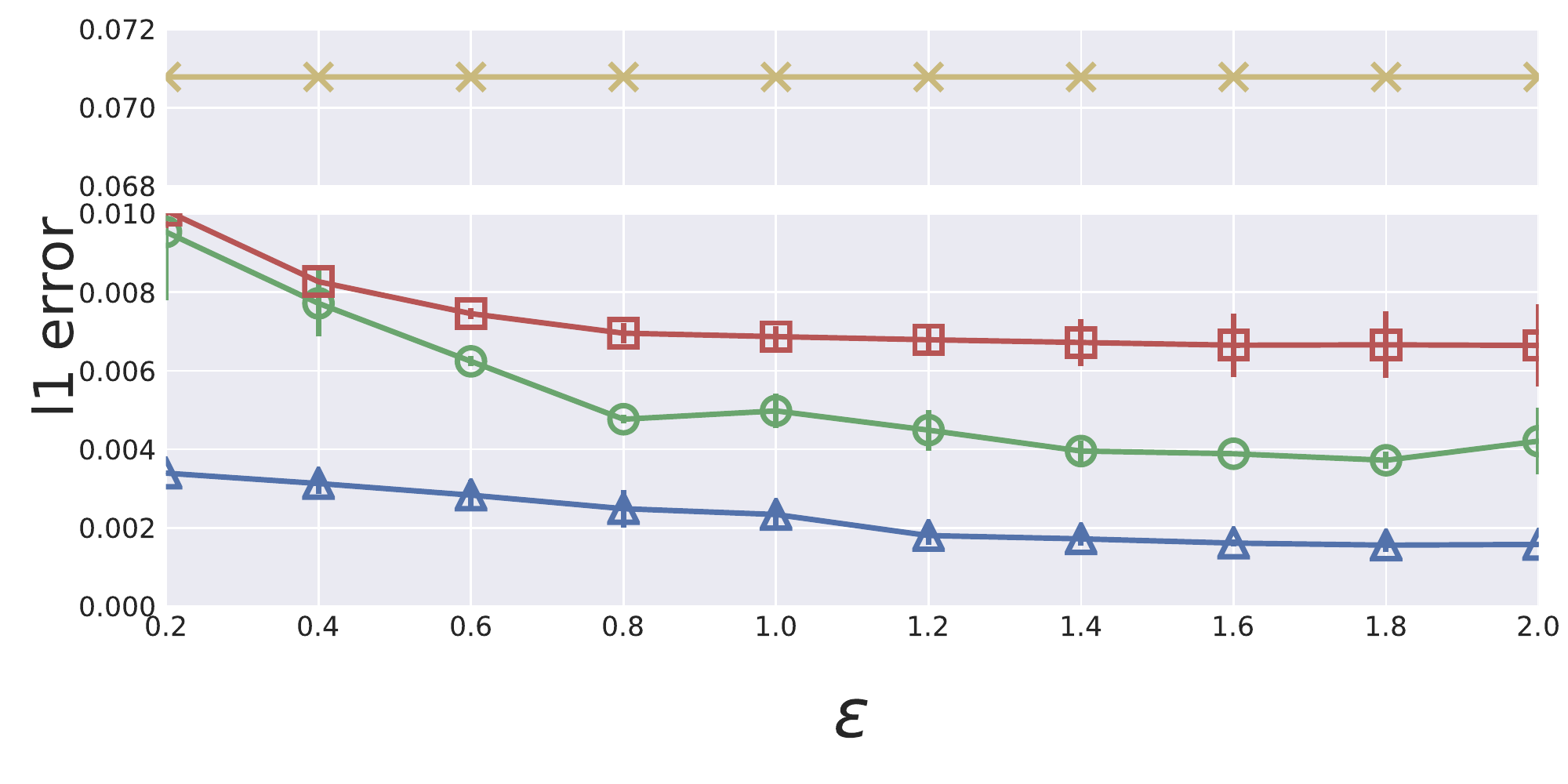}}
    \subfloat[classification]{\includegraphics[width=0.3\textwidth]{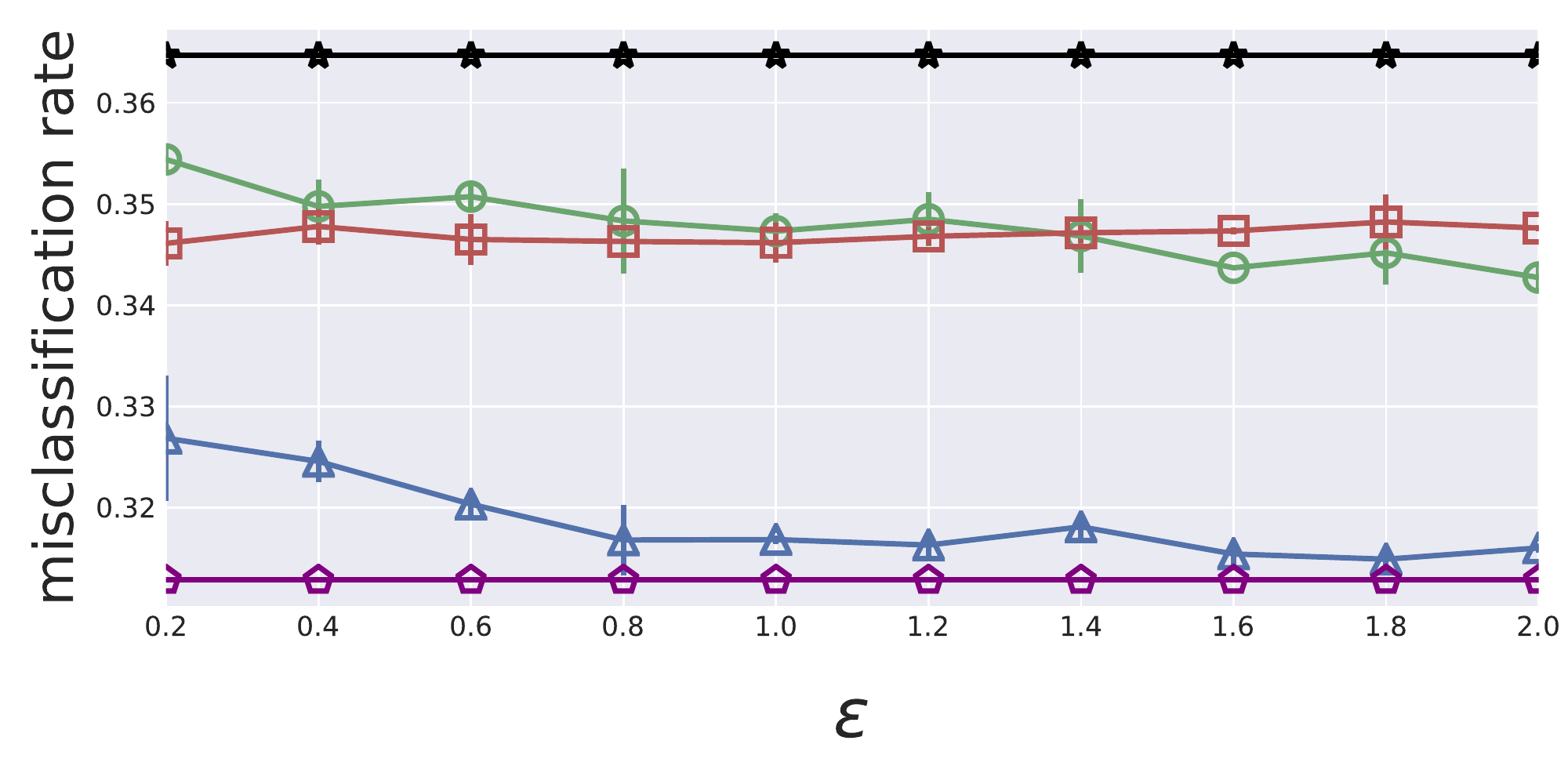}} \\ 
    
    \centering{US Accident} \\ [-2ex]

    \subfloat[pair-wise marginal]{\includegraphics[width=0.3\textwidth]{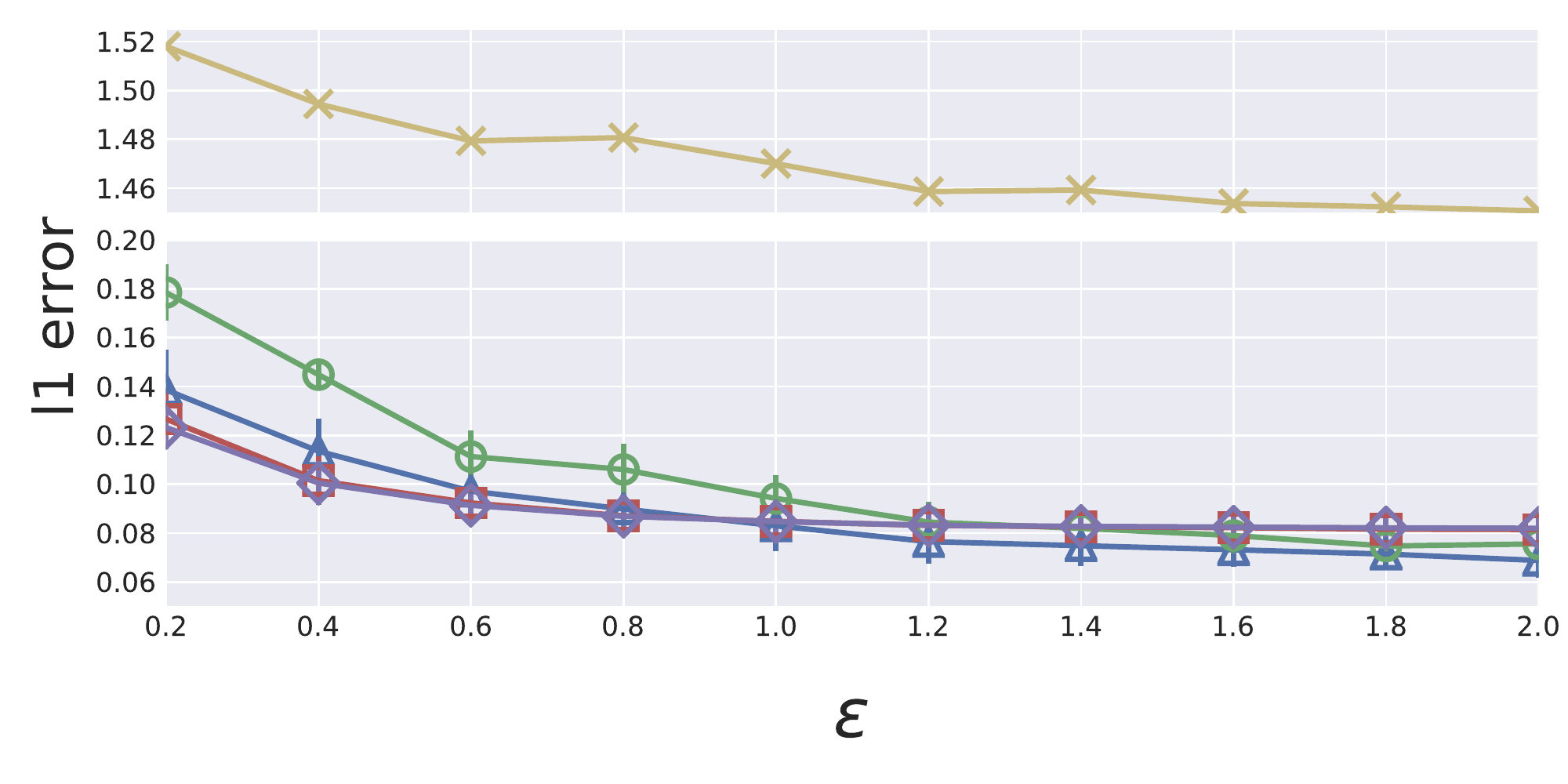}}
    \subfloat[range query]{\includegraphics[width=0.3\textwidth]{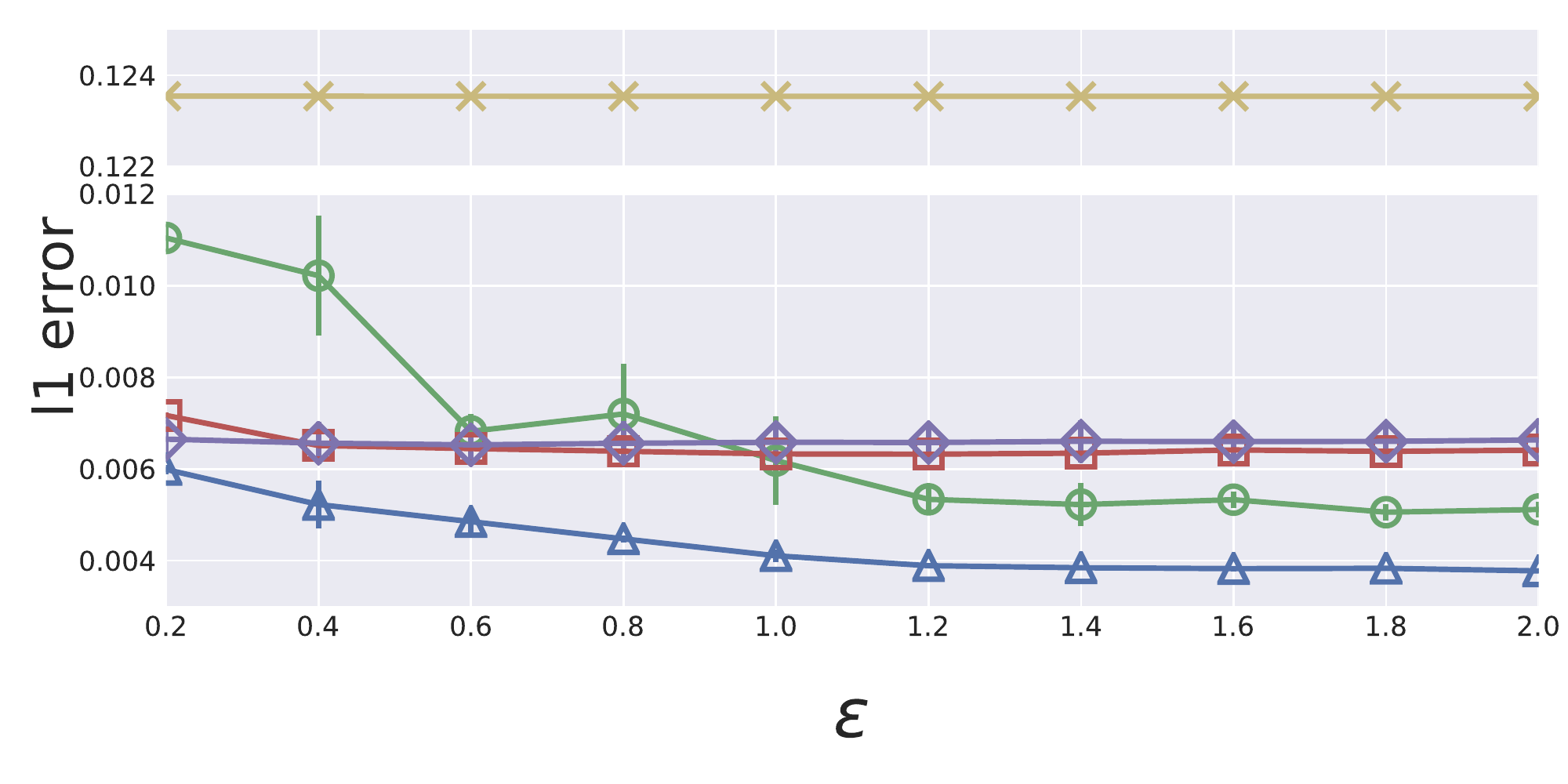}}
    \subfloat[classification]{\includegraphics[width=0.3\textwidth]{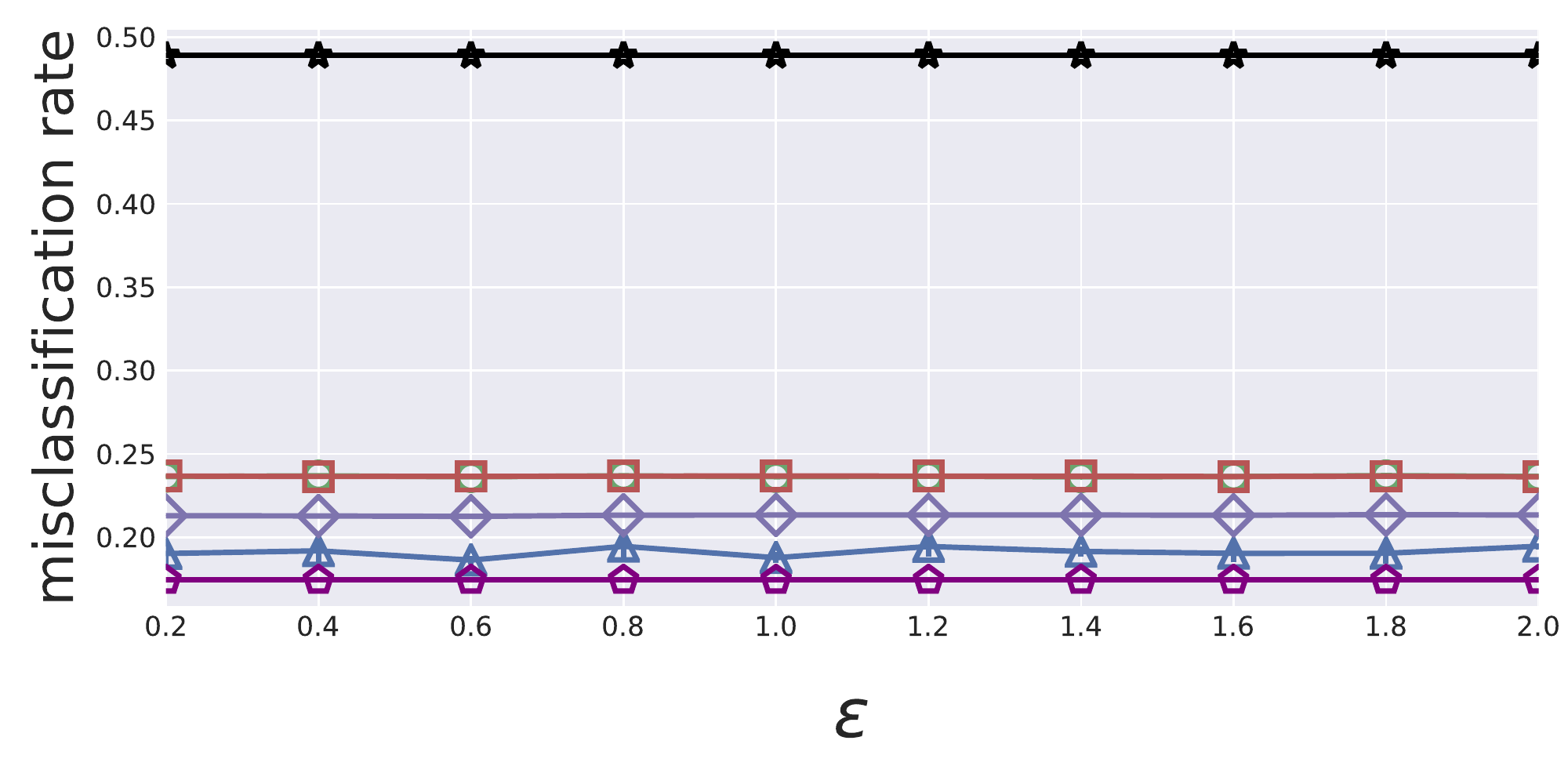}} \\ 
    
    \centering{Colorado} \\ 
    
    \subfloat{\includegraphics[width=0.8\textwidth]{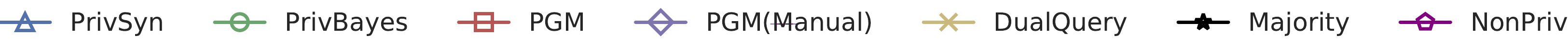}}  \\ [-2ex]

    \caption{End-to-end Comparison of different dataset generation methods.  
    \method is our proposed method.
    }
    \label{fig:comparison_e2e}
\end{figure*}

We also compare with a few other algorithms that do not follow the framework in Section~\ref{sec:framework}.
\begin{itemize}[leftmargin=*]
    \setlength\itemsep{-0.5em}
    \item \mypara{DualQuery} It generates records in a game theoretical manner.
    The main idea is to maintain a distribution over a workload of queries.
    One first samples a set of queries from the workload each time, and then generates a record that minimize the error of these queries.
    We refer the readers to Section~\ref{sec:related} for detailed discussion.
    
    \item For the classification task, we have another two competitors, \ie, \maj and \nonpriv.
    \maj represents the naive method that blindly predicts the label by the majority label.  Methods that perform worse than \maj means that the published dataset doesn't help the classification task, since the majority label can be outputted correctly even under very low privacy budget. 
    \nonpriv represents the method without enforcing differential privacy, it is the best case to aim for.  For \nonpriv, we split the original dataset into two disjoint parts, one for training and another for testing.
\end{itemize}

\mypara{Experimental Setting}
For \privbayes, \pgm and \method methods, we set the number of synthesized records the same as that of the original dataset.
Notice that we adopt unbounded differential privacy~\cite{KM11} in this paper, we cannot directly access the actual number of records in the original dataset.
Thus, we instead use the total count of marginals to approximate it.
For \dq method, the number of synthesized records is inherently determined by the privacy budget, the step size and the sample size~\cite{gaboardi2014dual}.
We use the same hyper-parameter settings as~\cite{gaboardi2014dual}, \ie, the step size is $2.0$ and sample size is $1000$.
We will illustrate the impact of the number of synthesized records on \method in Appendix~\ref{app:impact_number_records}.
\revision{}{By default, we set $\delta=\frac{1}{n^2}$ for all methods, where $n$ is the number of records in original dataset.}

All algorithms are implemented in Python 3.7 and all the experiments are conducted on a server with Intel Xeon E7-8867 v3 @ 2.50GHz and 1.5TB memory.
We repeat each experiment $5$ times and report the mean and standard deviation.
Due to space limitation, we put the experimental results of US Accident and Colorado in the main body, and defer the results of Adult and Loan datasets to Appendix~\ref{app:other_datasets}.
We also defer the comparison of synthesis methods to Appendix~\ref{app:comparison_synthesis}.

\begin{figure*}[!htpb]
    \centering

    \subfloat[pair-wise marginal]{\includegraphics[width=0.3\textwidth]{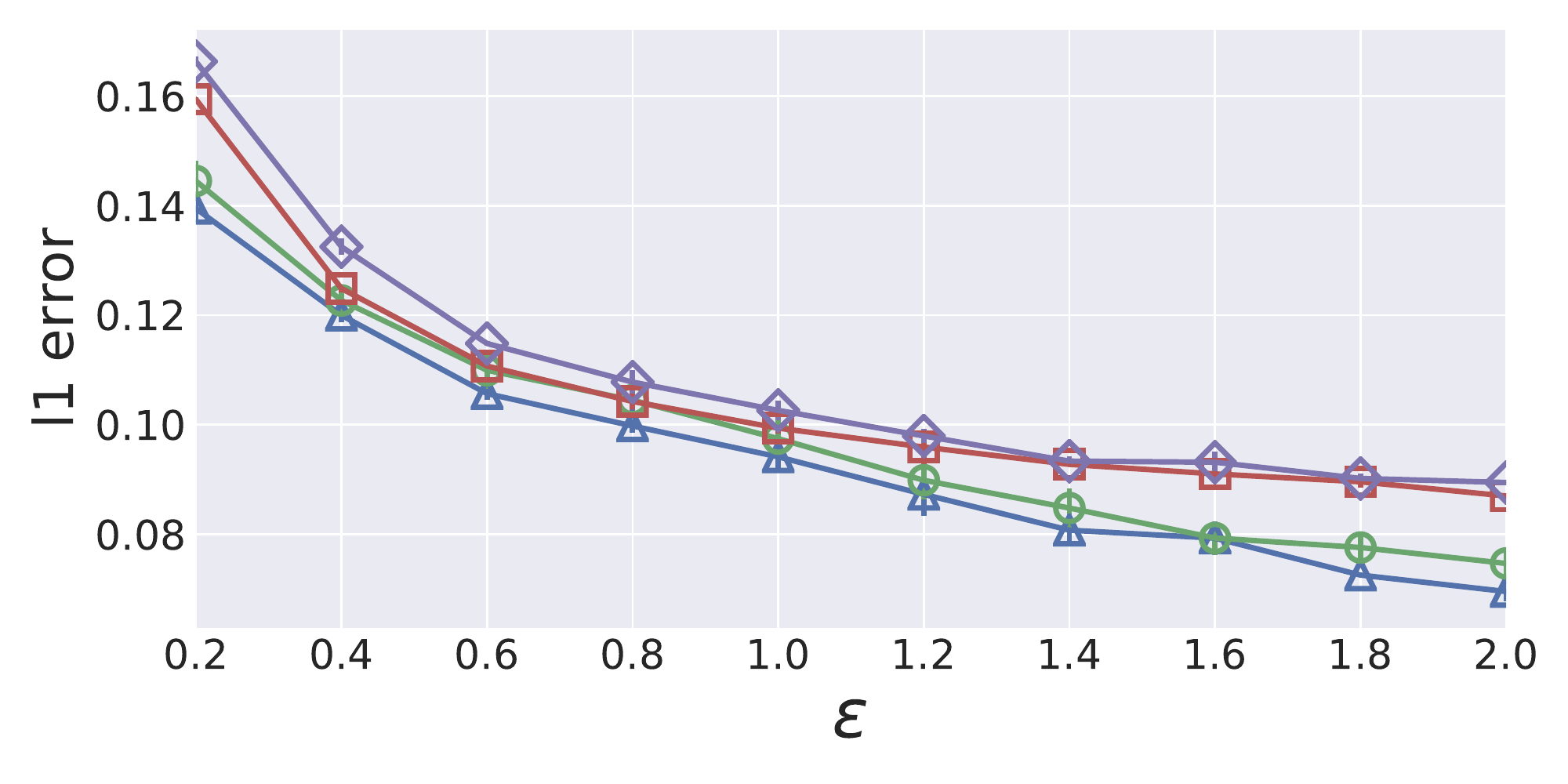}}
    \subfloat[range query]{\includegraphics[width=0.3\textwidth]{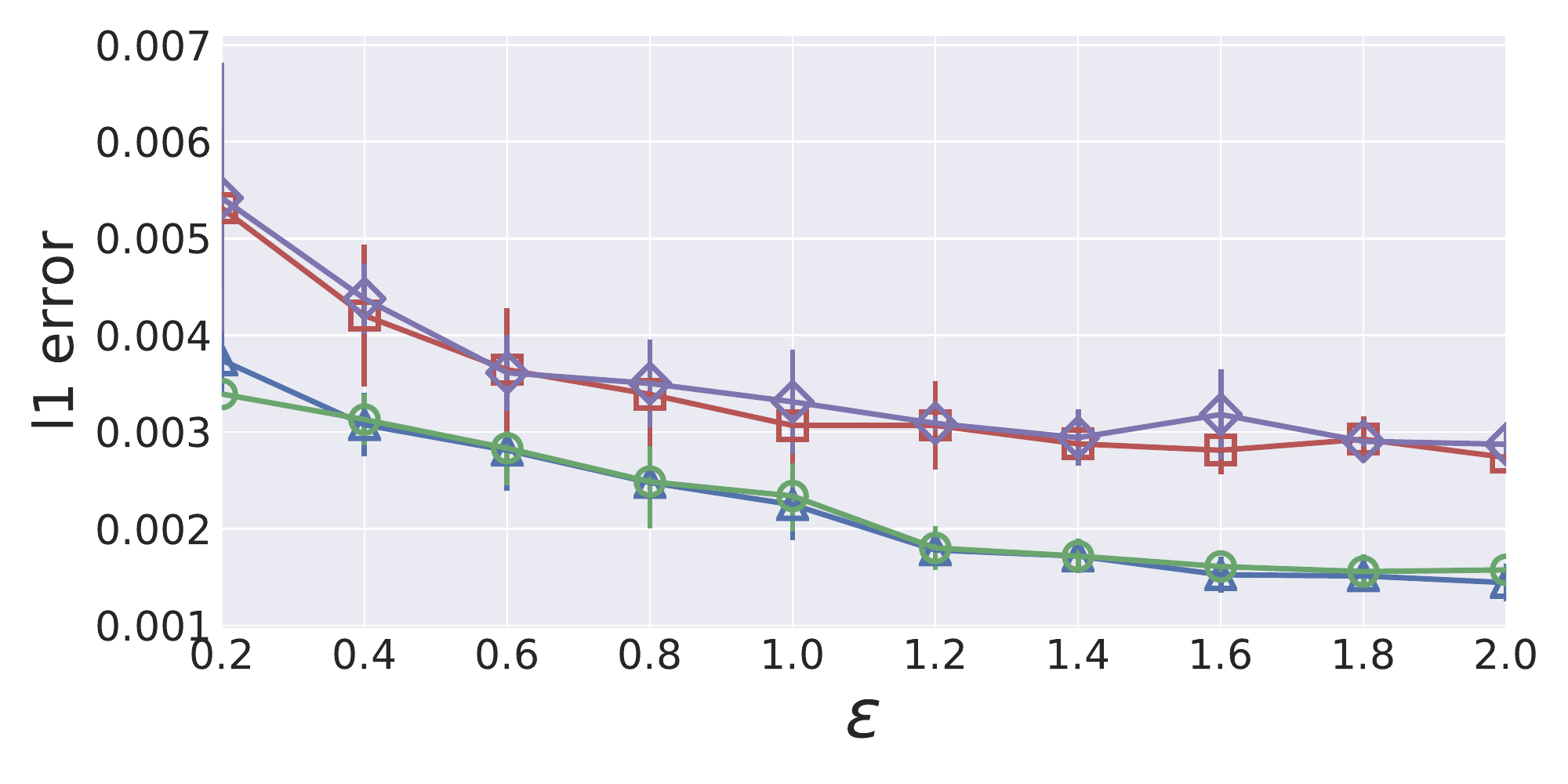}}
    \subfloat[classification]{\includegraphics[width=0.3\textwidth]{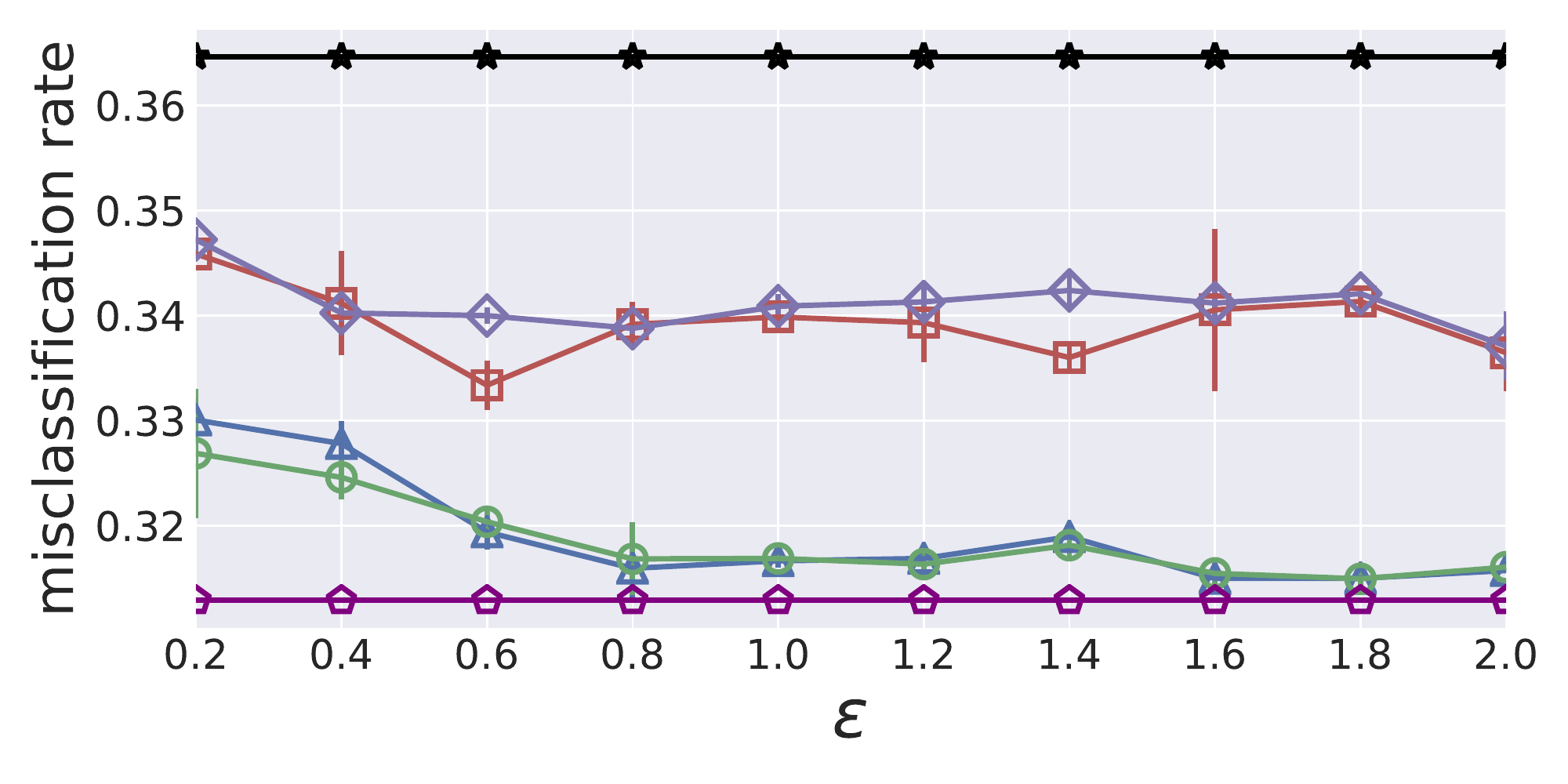}} \\ 
    
    \centering{US Accident} \\ [-2ex]

    \subfloat[pair-wise marginal]{\includegraphics[width=0.3\textwidth]{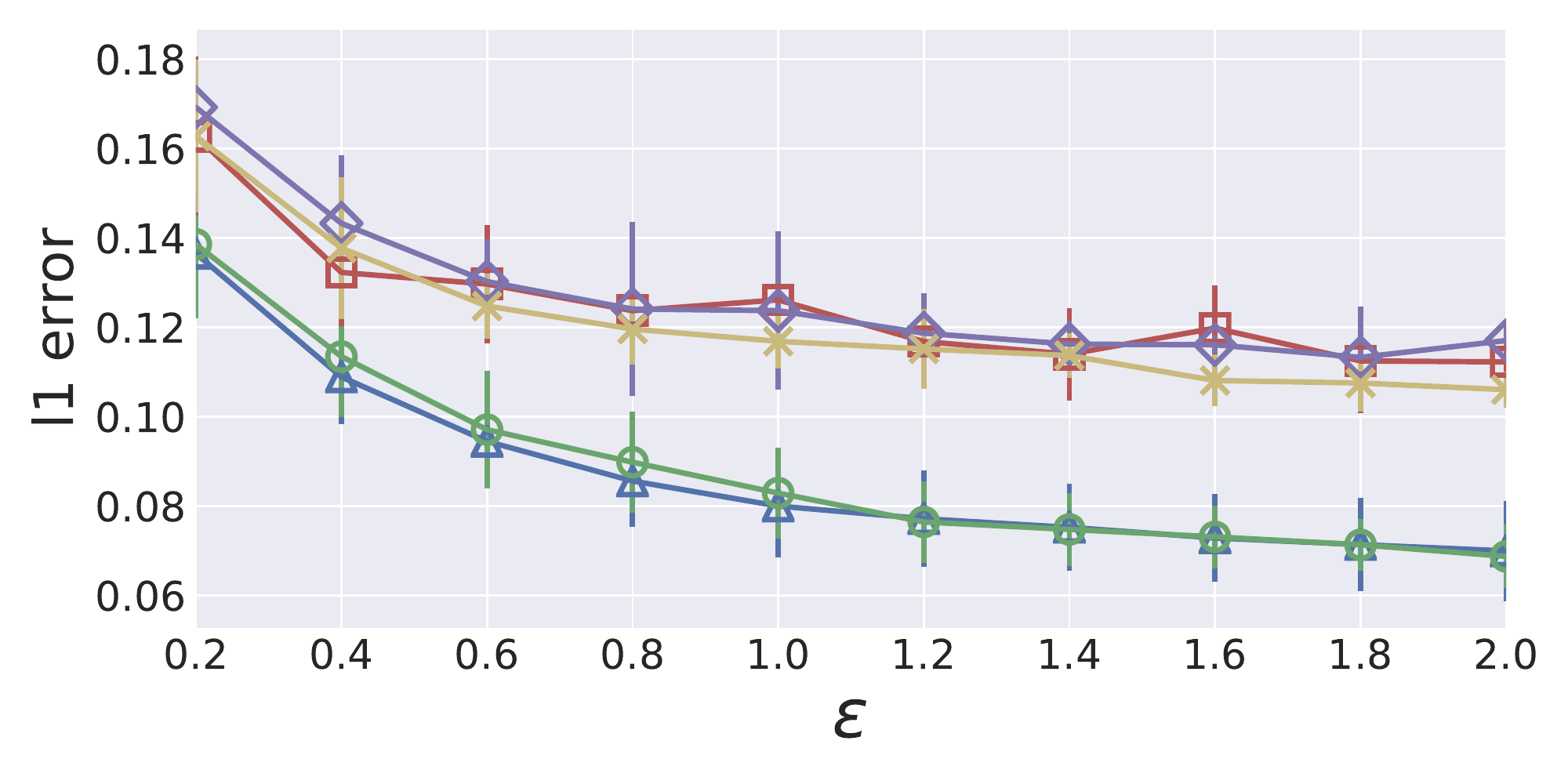}}
    \subfloat[range query]{\includegraphics[width=0.3\textwidth]{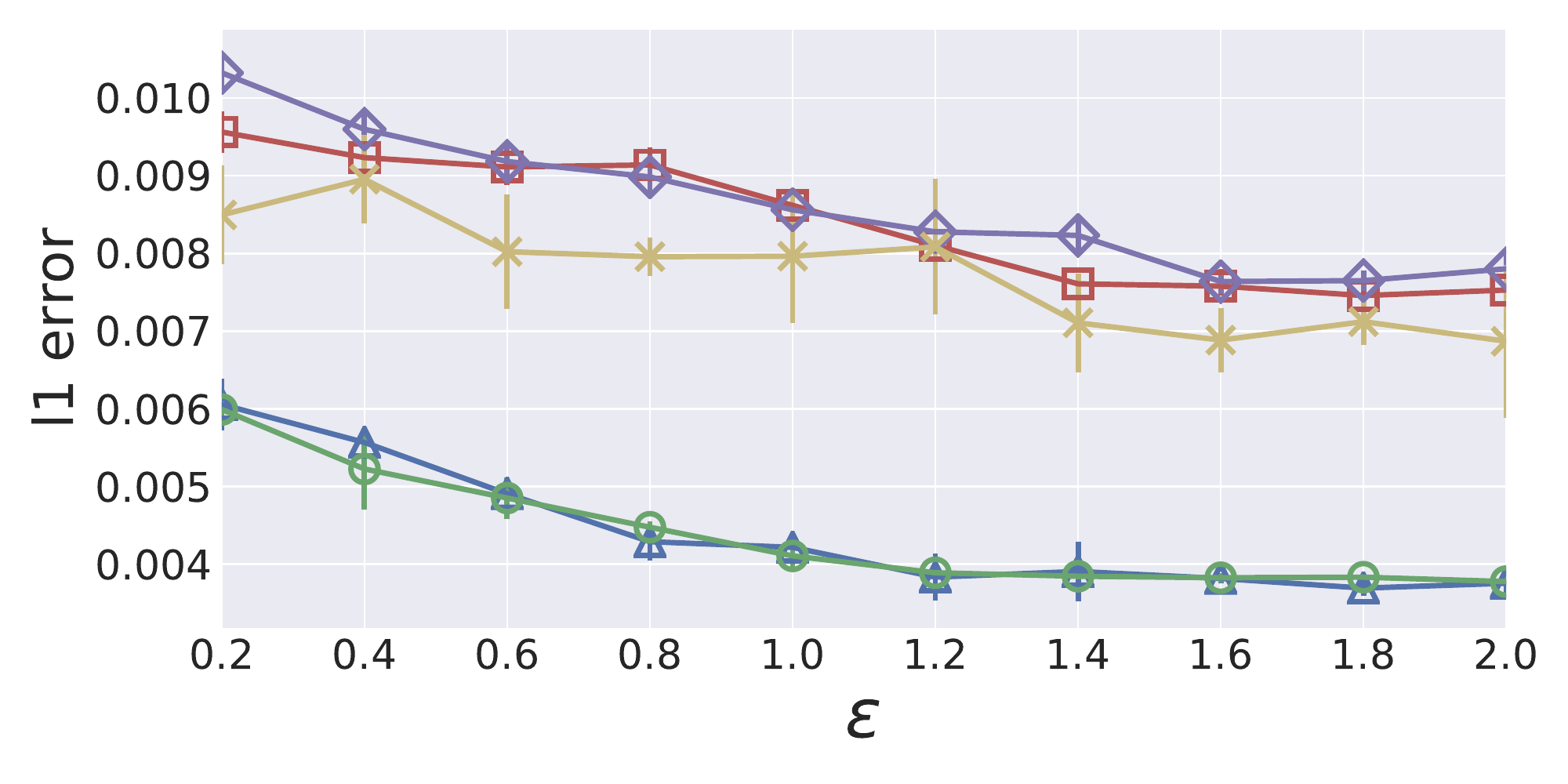}}
    \subfloat[classification]{\includegraphics[width=0.3\textwidth]{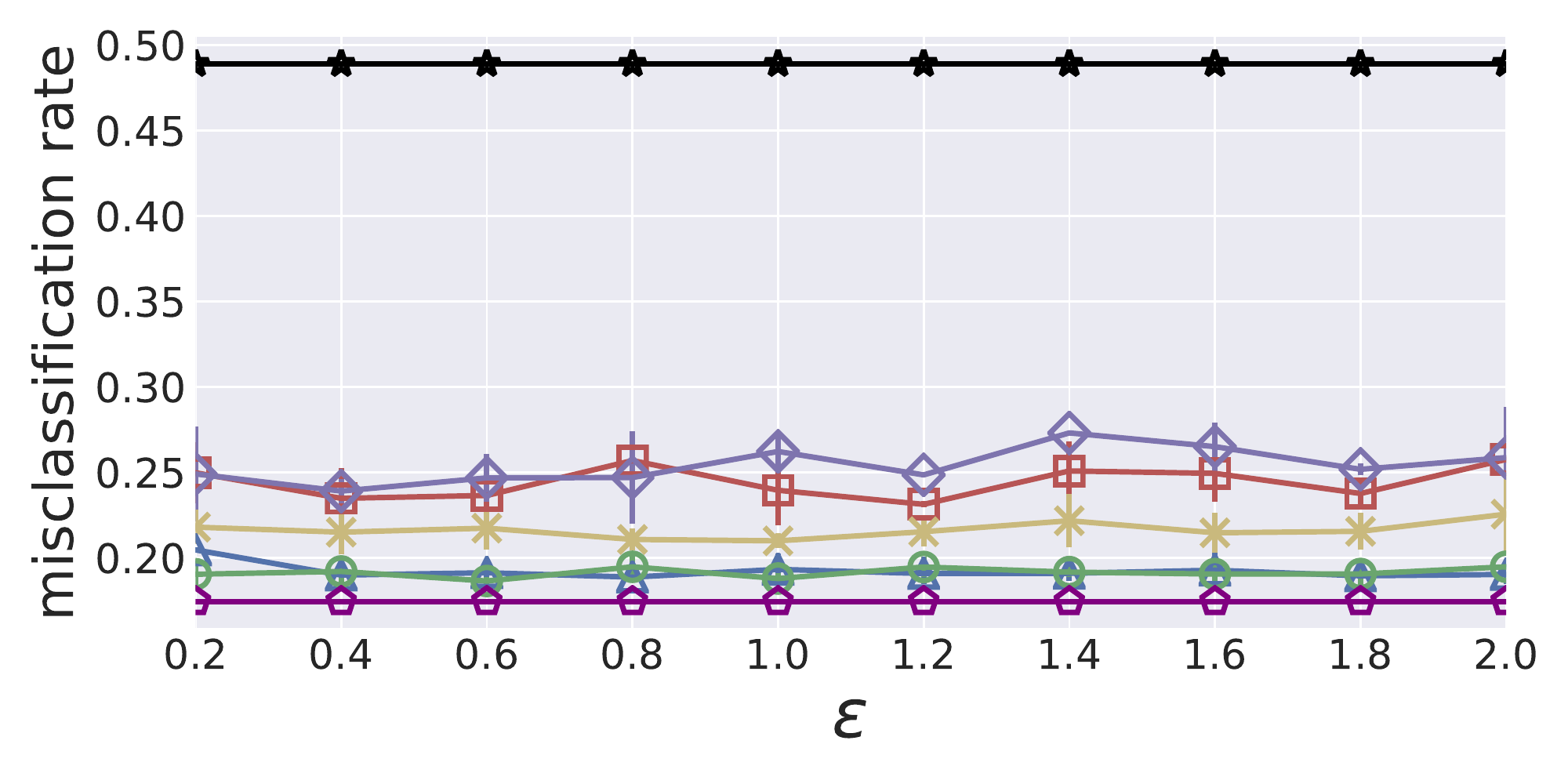}} \\ 

    \centering{Colorado} \\ 
    
    \subfloat{\includegraphics[width=0.8\textwidth]{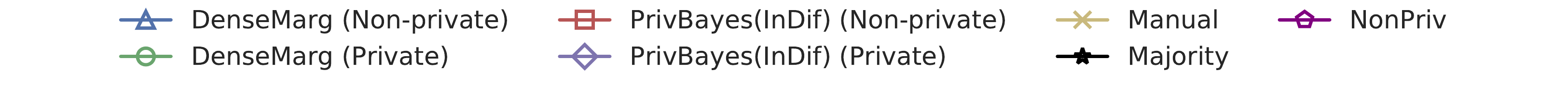}}  \\ [-2ex]

    \caption{Comparison of different marginal selection methods.
    \marginal is our proposed method.
    Non-private in the parenthese indicates that the marginal selection step do not consume privacy budget.}
    \label{fig:comparison_marginal}
\end{figure*}

\subsection{End-to-end Comparison}
\label{subsec:comparison_e2e}

\mypara{Setup}
For fair comparison, we use the optimal components and hyper-parameters for all methods.
Concretely, we use \privbayesid to select marginals for \privbayes and \pgm, since they can only handle sparse marginals.
Both \method and \dq can handle dense marginals; thus we use \marginal to select marginals for them.
For noise addition, we use \weightgauss for \privbayes, \pgm and \method.
\dq uses a game theoretical manner to generate synthetic datasets; thus it does not need the noise addition step.
For \privbayes, \pgm and \dq, we use the open-sourced code~\cite{pgm-code} by the author of \pgm to run the experiments.

\mypara{Results}
Figure~\ref{fig:comparison_e2e} illustrates the performance of different methods.
We do not show the classification performance of \dq since the misclassification rate is larger than \maj and the variance is large.
The experimental results show that our proposed \method
method consistently outperforms other methods for all datasets and all data analysis tasks.

For the pair-wise marginal task, the performance of \pgm and \privbayes is quite close to \method, meaning these two methods can effectively capture low-dimensional correlation.
However, the performance of range query task and classification task are much worse than \method, since range query and classification tasks require higher dimensional correlation.
\method can effectively preserve both low-dimensional and high-dimensional correlation.

The performance of \dq is significantly worse than other methods.
The reason is that generating each record consumes a portion of privacy budget, which limits the number of records generated by \dq.
In our experiments, the number of generated records by \dq is less than $300$ in all settings.
When the privacy budget is low, \eg, $\epsilon=0.2$, the number of generated records is less than $50$.
Insufficient number of records would lead to bad performance for all three data analysis tasks.

\subsection{Comparison of Marginal Selection \\ Methods}
\label{subsec:comparison_marginal}

\mypara{Setup}
We use \weightgauss method for noise addition, and use \gum for data synthesis.  
For each marginal selection method, we compare their performance in both private and non-private settings.
In the non-private setting, the marginal selection step do not consume privacy budget.
This can serve as a baseline to illustrate the robustness of different marginal selection methods.

\mypara{Results}
Figure~\ref{fig:comparison_marginal} illustrates the performance of different marginal selection methods.
For all datasets and all data analysis tasks, our proposed \marginal method consistently outperforms \privbayesid.
In the range query task, \marginal reduces the $\ell_1$ error by about $50\%$, which is much significant than that in pair-wise marginal release task.
This is because our range queries contain $3$ attributes, which requires higher dimensional correlation information than pair-wise marginal (contain $2$ attributes).
\marginal preserves more higher dimensional correlation by selecting more marginals than \privbayesid.

In all settings, the performance of \marginal in private setting and non-private setting are very close.
The reason is that \marginal tends to select the set of marginals with high \margselect, and adding moderate level of noise is unlikely to significantly change this set of marginals.
In our experiments, the overlapping ratio of the selected marginals between private setting and non-private setting is larger than $85\%$ in most cases.
This indicates that \marginal is very robust to noise.

\begin{figure*}[!htpb]
    \centering

    \subfloat[pair-wise marginal]{\includegraphics[width=0.3\textwidth]{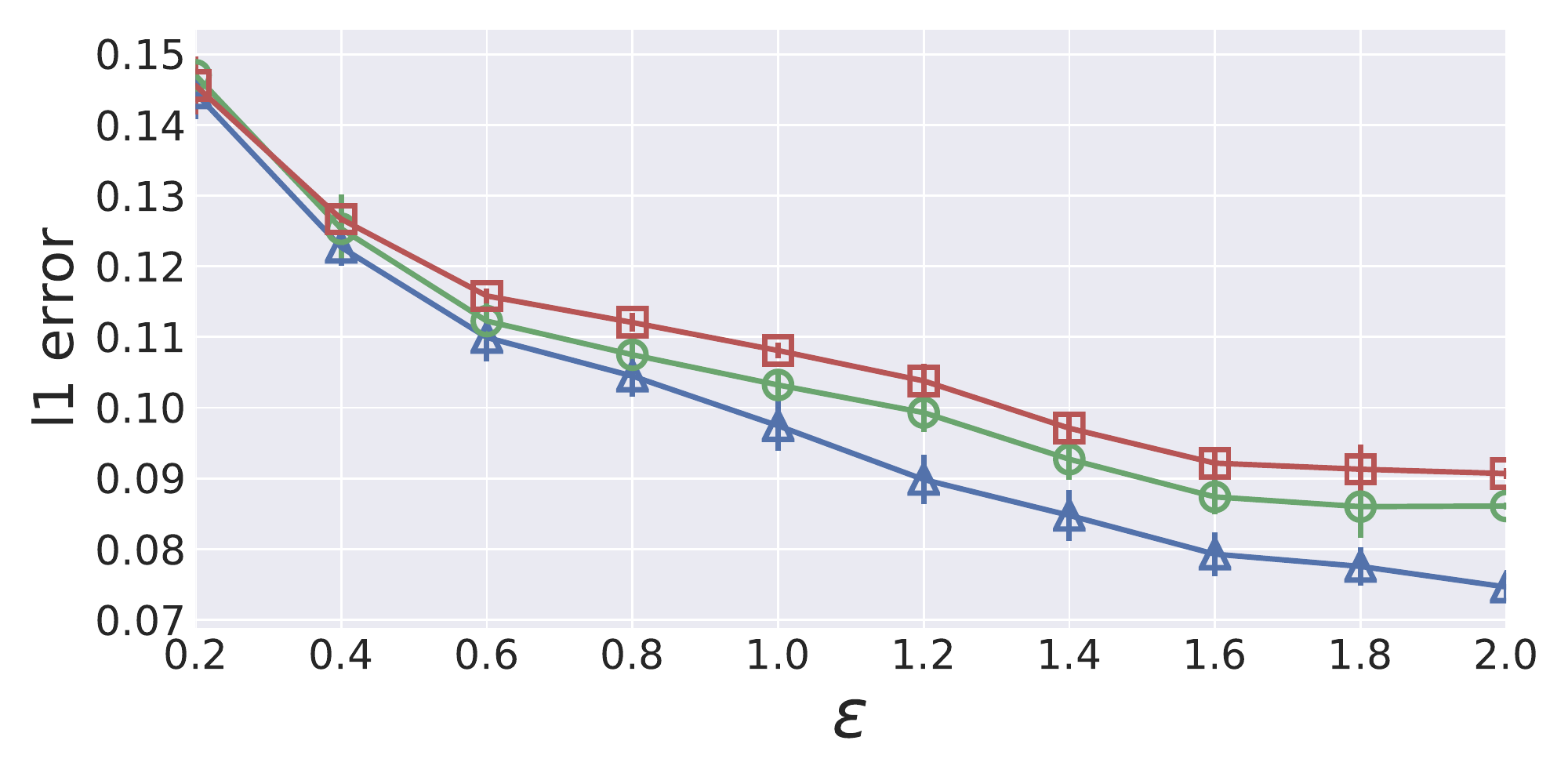}}
    \subfloat[range query]{\includegraphics[width=0.3\textwidth]{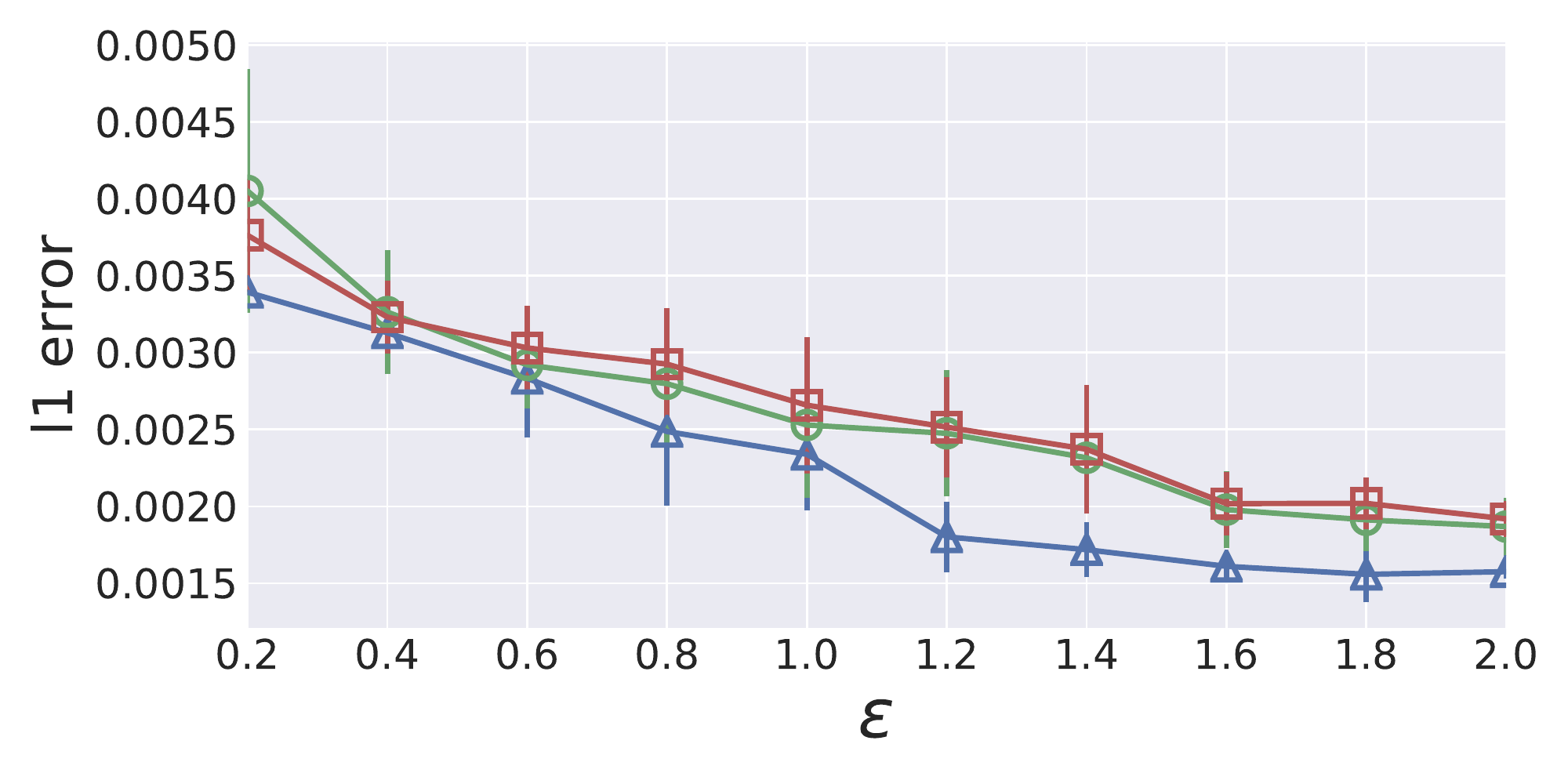}}
    \subfloat[classification]{\includegraphics[width=0.3\textwidth]{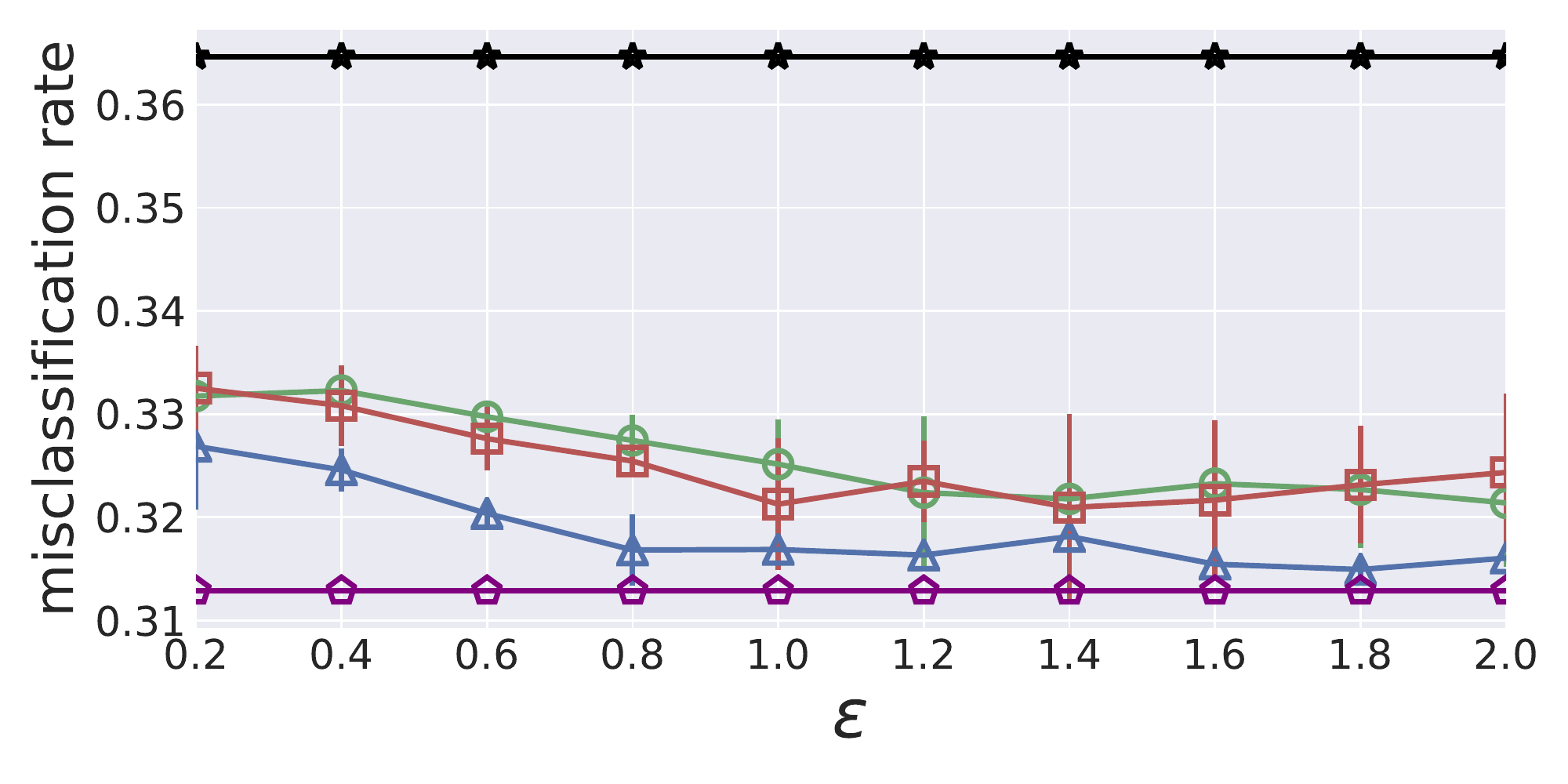}} \\ 
    
    \centering{US Accident} \\ [-2ex]

    \subfloat[pair-wise marginal]{\includegraphics[width=0.3\textwidth]{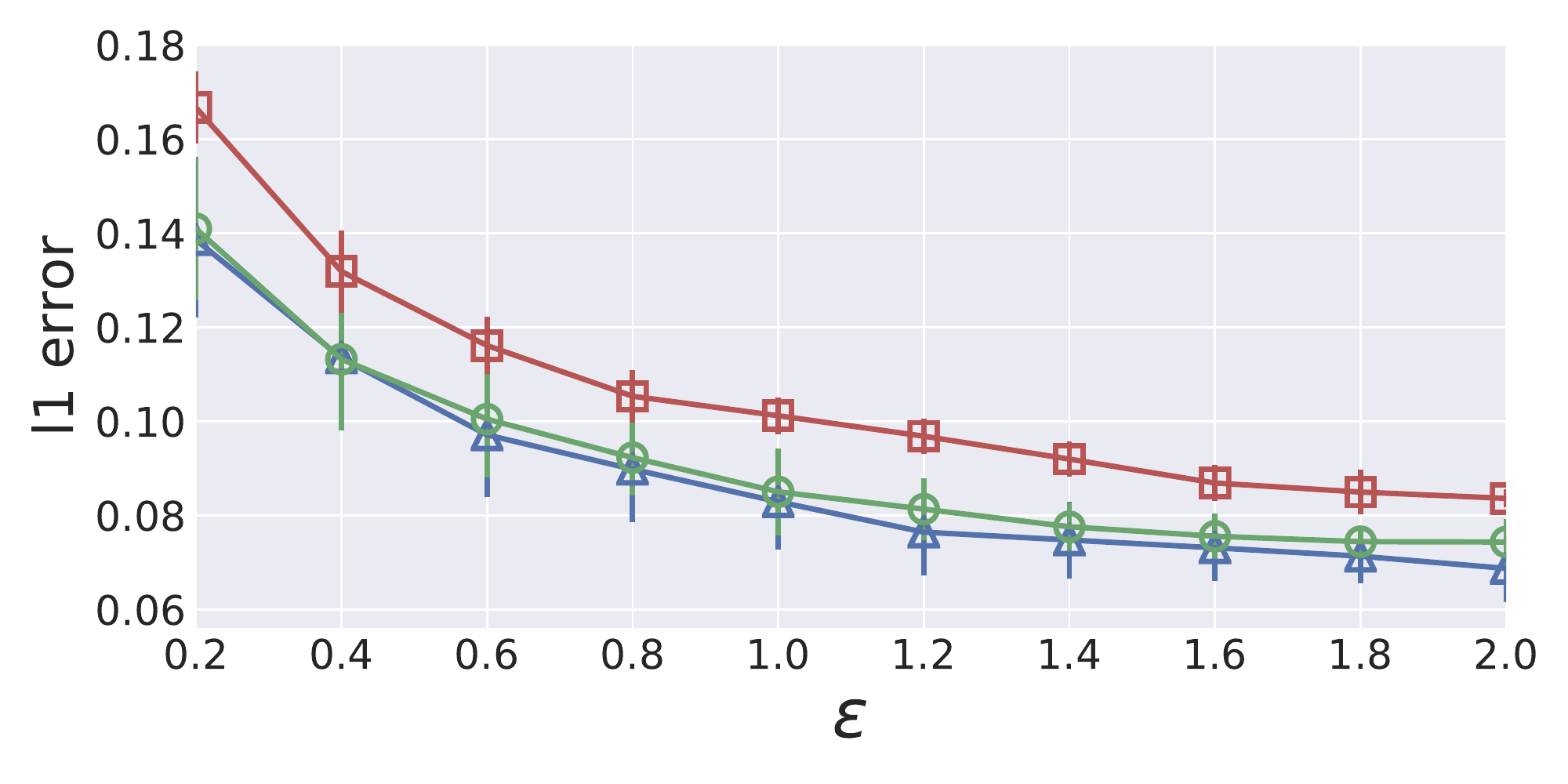}}
    \subfloat[range query]{\includegraphics[width=0.3\textwidth]{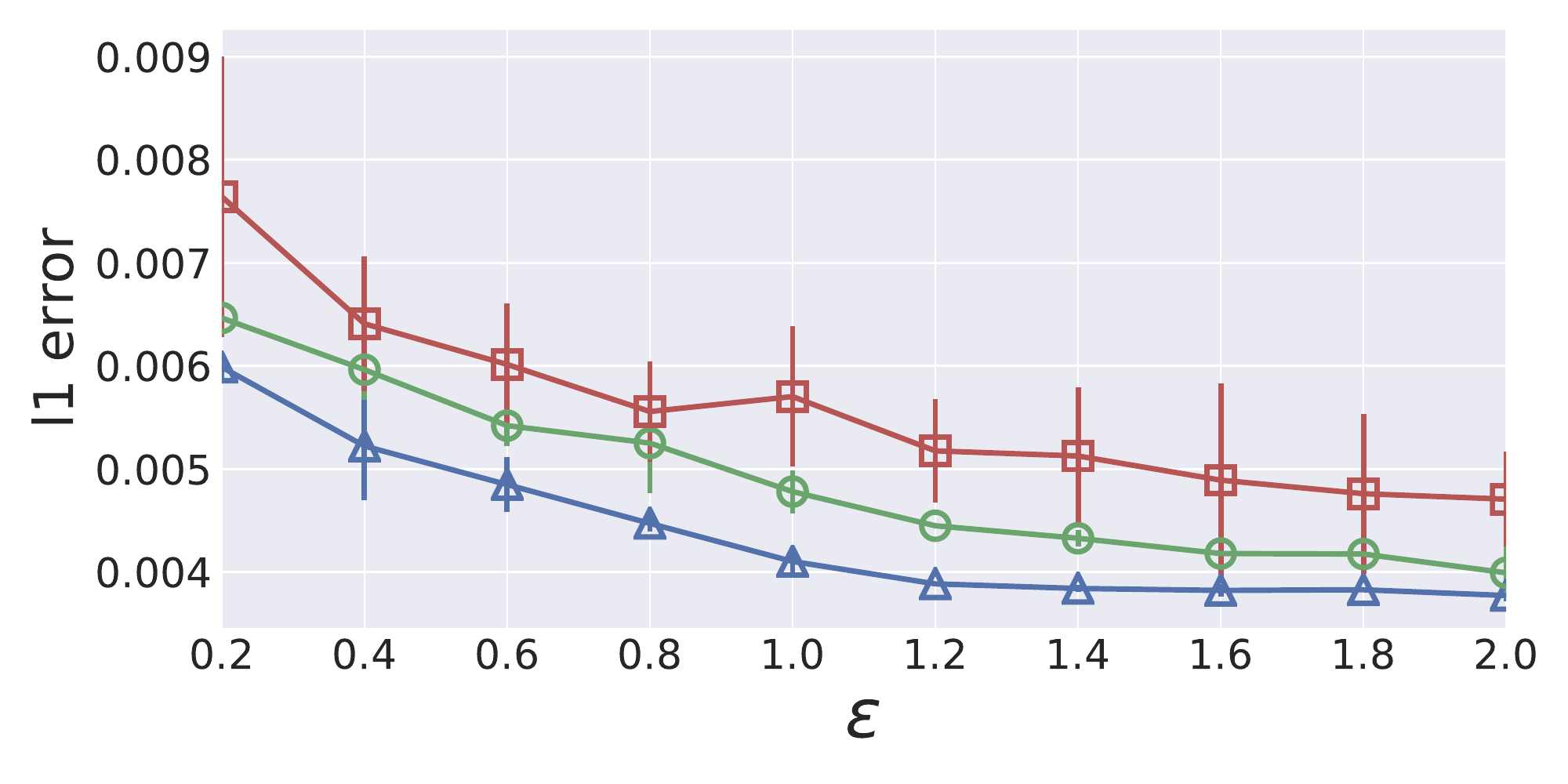}}
    \subfloat[classification]{\includegraphics[width=0.3\textwidth]{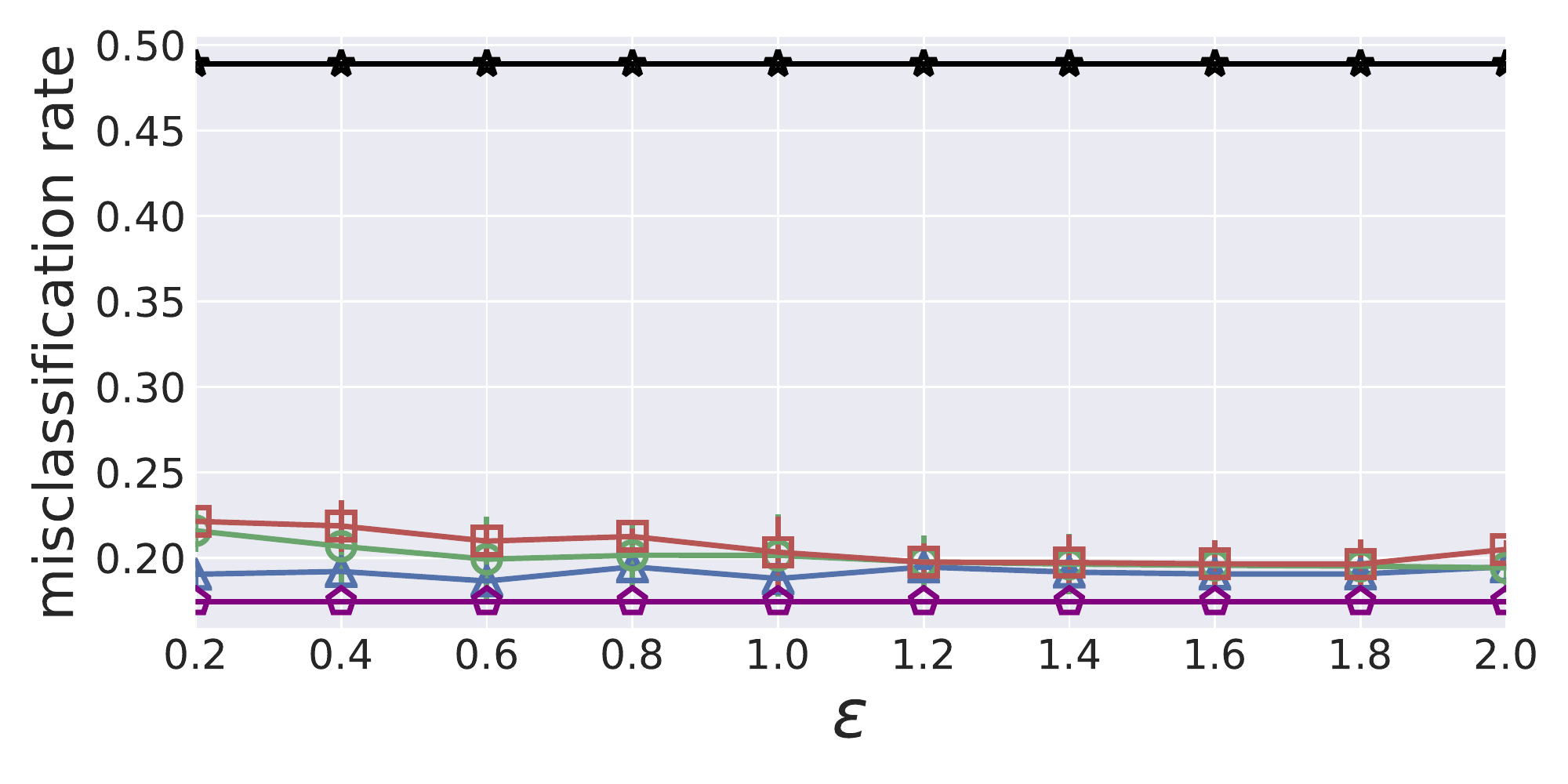}} \\ 
    
    \centering{Colorado} \\ 
    
    \subfloat{\includegraphics[width=0.8\textwidth]{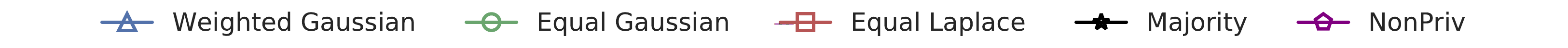}}  \\ [-2ex]

    \caption{Comparison of different noise addition methods.  
    \weightgauss is our proposed method.}
    \label{fig:comparison_noise_add}
\end{figure*}

\subsection{Comparison of Noise Addition Methods}
\label{subsec:comparison_noise_add}

\mypara{Setup}
We compare our proposed \weightgauss method with \equallap and \equalgauss methods.
Both Gaussian methods use \zcdp for composition.
The Laplace mechanism uses the naive composition, \ie, evenly allocate $\epsilon$ for all marginals. 
All methods use \marginal for marginal selection and \gum for data synthesis.

\mypara{Results}
Figure \ref{fig:comparison_noise_add} demonstrates the performance of different noise addition methods.
For all datasets and all data analysis tasks, our proposed \weightgauss method consistently outperforms the other two methods.
The advantage of \weightgauss increases when the privacy budget $\epsilon$ is larger.

In our experiment, both \weightgauss and \equalgauss methods use \zcdp to calculate the noise variance to each marginal, the main difference is that \weightgauss allocates privacy budget according to the number of cells, while \equalgauss evenly allocate privacy budget to all marginals.
The experimental results validate our analysis in Section \ref{subsec:marginal_selection} that \weightgauss is the optimal privacy budget allocation strategy.

\vspace{-1em}
\section{Related Work}
\label{sec:related}
Differential privacy (DP) has been the \textit{de facto} notion for protecting privacy.
Many DP algorithms have been proposed (see~\cite{dpbook,vadhan2017complexity} for theoretical treatments and~\cite{Li2016book} in a more practical perspective).
Most of the algorithms are proposed for specific tasks.
In this paper, we study the general task of generating a synthetic dataset with DP.  Compared to the ad-hoc methods, this approach may not achieve the optimal utility for the specific task.  But this approach is general in that given the synthetic dataset, any task can be performed, and there is no need to modify existing non-private algorithms.
There are a number of previous studies focus on generating synthetic dataset in a differentially private manner.
We classify them into three categoreis: graphical model based methods, game based methods and deep generative model based methods.
There are also some theoretical studies that discuss the hardness of differentially private data synthesis.

\mypara{Graphical Model Based Methods}
The main idea is to estimate a graphical model that approximates the distribution of the original dataset in a differentially private way.
\privbayes~\cite{zhang2017privbayes} and \bsg~\revision{}{(the initials of the authors' last names)}~\cite{bindschaedler2017plausible} approximate the data distribution using Bayesian Network.
These methods, however, need to call Exponential Mechanism~\cite{zhang2017privbayes} or Laplace Mechanism~\cite{bindschaedler2017plausible} many times, making the network structure inaccurate when the privacy budget is limited; and the overall utility is sensitive to the quality of the initial selected node.

\pgm~\cite{mckenna2019graphical} and \jtree~\cite{chen2015differentially} utilize Markov Random Field to approximate the data distribution.
\pgm takes as input a set of predefined low-dimensional marginals, and estimates a Markov Random Field that best matches these marginals.
\jtree first estimates a dependency graph by setting a threshold to the mutual information of pairwise attributes, and then obtains the Markov Random Field by transforming the dependency graph into a junction tree.
\pgm do not provide marginal selection method in the paper~\cite{mckenna2019graphical}.
\jtree proposes to use SVT to select marginals; however, Lyu \etal~\cite{lyu2016understanding} point out that JTree utilizes SVT in a problematic way.
The main limitation of graphical model based methods is that they cannot handle dense marginals that capture more correlation information.

\mypara{Game Based Methods}
There are works that formulate the dataset synthesis problem as a zero-sum game~\cite{hardt2012simple,gaboardi2014dual,vietrinew}.
Assume there are two players, data player and query player.
MWEM~\cite{hardt2012simple} method solves the game by having the data player use a no-regret learning algorithm, and the query player repeatedly best responds.
Dual Query~\cite{gaboardi2014dual} switches the role of the two players.
Concretely, the data player in MWEM maintains a distribution over the whole data domain, and query player repeatedly use exponential mechanism to select a query that have the worse performance from a workload of queries to update data player's distribution.
The main limitation of MWEM is that when the dataset domain is large (from $3 \cdot 10^{39}$ to $5 \cdot 10^{162}$ in our experiments), maintaining the full distribution is infeasible.
Thus, we do not compare with MWEM in our experiments.

In contrast, the query player in Dual Query maintains a distribution over all queries.
The query player each time samples a set of queries from the workload, and the data player generates a record that minimizes the error of these queries.
The shortcoming is that generating each record would consume a portion of privacy budget; thus one cannot generate sufficient records as discussed in Section~\ref{sec:exp}.
Moreover, both methods require a workload of queries in advance, and the generated dataset is guaranteed to be similar to the original dataset with respect to the query class.
This makes MWEM and Dual Query incapable of handling arbitrary kinds of tasks with satisfied accuracy.
The authors of~\cite{vietrinew} improve both MWEM and DualQuery by replacing their core components; however, this work follows the same framework with MWEM and QualQuery and do not address the main limitation of them.

\mypara{Deep Generative Model Based Methods}
Another approach is to train a deep generative model satisfying differential privacy, and use the deep generative model to generate a synthetic dataset.
The most commonly used deep generative model is the Generative Adversarial Network (GAN), and there are multiple studies focus on training GAN in a differentially private way~\cite{zhang2018differentially,beaulieu2019privacy,abay2018privacy,frigerio2019differentially,tantipongpipat2019differentially}.
The main idea is to utilize the DP-SGD framework~\cite{abadi2016deep} to add noise in the optimization procedure (i.e., stochastic gradient descent).
However, the preliminary application of GAN is to generate images.
Thus, the objective of GAN is to generate data records that look authentic, instead of approximating the original distribution, applying the GAN model to the current problem cannot generate a synthetic dataset with enough variations.
\revision{}{
In the NIST challenge~\cite{nist-challenge}, there are two teams adapting the GAN-based method to synthesize high-dimensional data, while their scores are much lower than \pgm and \privbayes.}
Thus, we do not compare this line of methods in our experiments.

\mypara{Theoretical Results}
There are a series of negative theoretical results concerning DP in the non-interactive setting~\cite{dwork2009complexity,ullman2011pcps,DN03,DMT07,DY08,DMNS06,DNR+09}.  
These results have been interpreted ``to mean that one cannot answer a linear, in the database size, number of queries with small noise while preserving privacy'' and to motivate ``an interactive approach to private data analysis where the number of queries is limited to be small -- sub-linear in the size $n$ of the dataset''~\cite{DNR+09}.  

We point out that, theoretical negative results notwithstanding, non-interactive publishing can serve an important role in private data publishing.  
The negative results essentially say that when the set of queries is sufficiently broad, one cannot guarantee that all of them are answered accurately.  
These results are all based on query sets that are broader than the natural set of queries in which one is interested.  
For example, suppose the dataset is one-dimensional where each value is an integer number in domain $[m]=\{1,2\ldots,m\}$.  
These results say that one cannot answer counting queries for arbitrary subsets of $[m]$ with error less than $\Theta(\sqrt{n})$, where $n$ is the size of the dataset.  
However, range queries, which are likely to be what one is interested in, can be answered with less error.
Moreover, these results are all asymptotic and do not rule out useful algorithms in practice.  When one plugs in actual parameters, the numbers that come out often have no bearing on practice.

\revisionstart
\section{Discussion and Limitations}
\label{sec:limitation}

In this section, we discuss the application scope and limitations of \method.

\mypara{Only Applicable to Tabular Data}
\method focuses on the tabular data and cannot handle other types of data such as image or streaming data.
Note that other existing methods (\privbayes, \pgm and \dq) also have this limitation.
We defer the application of \method to image dataset and sequential dataset to future work.

\mypara{Miss Some Higher Dimensional Correlation}
\method only considers low-degree marginals that may not capture some high-dimensional correlation information.
Notice that other marginal selection methods such as \privbayes and \bsg also use low-degree marginals to approximate the high-dimensional datasets and also have this limitation.
To capture higher dimensional correlation, one possibility is to consider all triple-wise marginals or higher-dimensional marginals; however, doing this may introduce too much noise for each of the marginal, resulting in inaccurate selection.
In practice, low-dimensional marginals are sufficient to capture enough correlation information, as illustrated on the four real-world datasets used in our experiments.

\revisionend

\section{Conclusion}
\label{sec:conc}

In this paper, we present \method for publishing a synthetic dataset under differential privacy.  
We identify the core steps in the process and our proposed method improves on these steps.  
We extensively evaluate different methods on multiple datasets to demonstrate the superiority of our method.

\newpage
{
    \bibliographystyle{plain}
	\bibliography{refs/privacy}

\begin{thebibliography}{10}

\bibitem{decay}
\url{http://cs231n.github.io/neural-networks-3/\#anneal}.

\bibitem{pgm-code}
\url{https://github.com/ryan112358/private-pgm}.

\bibitem{abadi2016deep}
Mart{\'\i}n Abadi, Andy Chu, Ian Goodfellow, H~Brendan McMahan, Ilya Mironov,
  Kunal Talwar, and Li~Zhang.
\newblock Deep learning with differential privacy.
\newblock In {\em Proceedings of the 2016 ACM SIGSAC Conference on Computer and
  Communications Security}, pages 308--318. ACM, 2016.

\bibitem{abay2018privacy}
Nazmiye~Ceren Abay, Yan Zhou, Murat Kantarcioglu, Bhavani Thuraisingham, and
  Latanya Sweeney.
\newblock Privacy preserving synthetic data release using deep learning.
\newblock In {\em Joint European Conference on Machine Learning and Knowledge
  Discovery in Databases}, pages 510--526. Springer, 2018.

\bibitem{abowd2018us}
John~M Abowd.
\newblock The us census bureau adopts differential privacy.
\newblock In {\em Proceedings of the 24th ACM SIGKDD International Conference
  on Knowledge Discovery \& Data Mining}, pages 2867--2867, 2018.

\bibitem{ahuja1988network}
Ravindra~K Ahuja, Thomas~L Magnanti, and James~B Orlin.
\newblock Network flows.
\newblock 1988.

\bibitem{arora2012multiplicative}
Sanjeev Arora, Elad Hazan, and Satyen Kale.
\newblock The multiplicative weights update method: a meta-algorithm and
  applications.
\newblock {\em Theory of Computing}, 8(1):121--164, 2012.

\bibitem{AN10}
A.~Asuncion and D.J. Newman.
\newblock {UCI} machine learning repository, 2010.

\bibitem{aistats:Bassily19}
Raef Bassily.
\newblock Linear queries estimation with local differential privacy.
\newblock In {\em {AISTATS}}, 2019.

\bibitem{beaulieu2019privacy}
Brett~K Beaulieu-Jones, Zhiwei~Steven Wu, Chris Williams, Ran Lee, Sanjeev~P
  Bhavnani, James~Brian Byrd, and Casey~S Greene.
\newblock Privacy-preserving generative deep neural networks support clinical
  data sharing.
\newblock {\em Circulation: Cardiovascular Quality and Outcomes},
  12(7):e005122, 2019.

\bibitem{bindschaedler2017plausible}
Vincent Bindschaedler, Reza Shokri, and Carl~A Gunter.
\newblock Plausible deniability for privacy-preserving data synthesis.
\newblock {\em Proceedings of the VLDB Endowment}, 10(5), 2017.

\bibitem{BLR08}
Avrim Blum, Katrina Ligett, and Aaron Roth.
\newblock A learning theory approach to non-interactive database privacy.
\newblock In {\em STOC}, pages 609--618, 2008.

\bibitem{bun2016concentrated}
Mark Bun and Thomas Steinke.
\newblock Concentrated differential privacy: Simplifications, extensions, and
  lower bounds.
\newblock In {\em Theory of Cryptography Conference}, pages 635--658. Springer,
  2016.

\bibitem{chen2015differentially}
Rui Chen, Qian Xiao, Yu~Zhang, and Jianliang Xu.
\newblock Differentially private high-dimensional data publication via
  sampling-based inference.
\newblock In {\em Proceedings of the 21th ACM SIGKDD International Conference
  on Knowledge Discovery and Data Mining}, pages 129--138. ACM, 2015.

\bibitem{DWH+11}
Bolin Ding, Marianne Winslett, Jiawei Han, and Zhenhui Li.
\newblock Differentially private data cubes: optimizing noise sources and
  consistency.
\newblock In {\em SIGMOD Conference}, pages 217--228, 2011.

\bibitem{DN03}
Irit Dinur and Kobbi Nissim.
\newblock Revealing information while preserving privacy.
\newblock In {\em PODS}, pages 202--210, 2003.

\bibitem{DNR+09}
C~Dwork, M~Naor, O~Reingold, G.N Rothblum, and S~Vadhan.
\newblock On the complexity of differentially private data release: efficient
  algorithms and hardness results.
\newblock {\em STOC}, pages 381--390, 2009.

\bibitem{DRV10}
C~Dwork, G~Rothblum, and S~Vadhan.
\newblock Boosting and differential privacy.
\newblock {\em Foundations of Computer Science (FOCS), 2010 51st Annual IEEE
  Symposium on}, pages 51 -- 60, 2010.

\bibitem{DY08}
C~Dwork and S~Yekhanin.
\newblock New efficient attacks on statistical disclosure control mechanisms.
\newblock {\em Advances in Cryptology--CRYPTO 2008}, pages 469--480, 2008.

\bibitem{Dwo06}
Cynthia Dwork.
\newblock Differential privacy.
\newblock In {\em ICALP}, pages 1--12, 2006.

\bibitem{DMNS06}
Cynthia Dwork, Frank McSherry, Kobbi Nissim, and Adam Smith.
\newblock Calibrating noise to sensitivity in private data analysis.
\newblock In {\em TCC}, pages 265--284, 2006.

\bibitem{DMT07}
Cynthia Dwork, Frank McSherry, and Kunal Talwar.
\newblock The price of privacy and the limits of {LP} decoding.
\newblock In {\em STOC}, pages 85--94, 2007.

\bibitem{dwork2009complexity}
Cynthia Dwork, Moni Naor, Omer Reingold, Guy~N Rothblum, and Salil Vadhan.
\newblock On the complexity of differentially private data release: efficient
  algorithms and hardness results.
\newblock In {\em Proceedings of the forty-first annual ACM symposium on Theory
  of computing}, pages 381--390, 2009.

\bibitem{dpbook}
Cynthia Dwork and Aaron Roth.
\newblock The algorithmic foundations of differential privacy.
\newblock {\em Foundations and Trends in Theoretical Computer Science},
  9(3-4):211--407, 2014.

\bibitem{frigerio2019differentially}
Lorenzo Frigerio, Anderson~Santana de~Oliveira, Laurent Gomez, and Patrick
  Duverger.
\newblock Differentially private generative adversarial networks for time
  series, continuous, and discrete open data.
\newblock In {\em IFIP International Conference on ICT Systems Security and
  Privacy Protection}, pages 151--164. Springer, 2019.

\bibitem{gaboardi2014dual}
Marco Gaboardi, Emilio Jes{\'u}s~Gallego Arias, Justin Hsu, Aaron Roth, and
  Zhiwei~Steven Wu.
\newblock Dual query: Practical private query release for high dimensional
  data.
\newblock In {\em International Conference on Machine Learning}, pages
  1170--1178, 2014.

\bibitem{hardt2012simple}
Moritz Hardt, Katrina Ligett, and Frank McSherry.
\newblock A simple and practical algorithm for differentially private data
  release.
\newblock In {\em Advances in Neural Information Processing Systems}, pages
  2339--2347, 2012.

\bibitem{johnson2018towards}
Noah Johnson, Joseph~P Near, and Dawn Song.
\newblock Towards practical differential privacy for sql queries.
\newblock {\em Proceedings of the VLDB Endowment}, 11(5):526--539, 2018.

\bibitem{kaggle-loan}
Kaggle.
\newblock Kaggle lending club loan data.
\newblock \url{https://www.kaggle.com/wendykan/lending-club-loan-data}.

\bibitem{KM11}
Daniel Kifer and Ashwin Machanavajjhala.
\newblock No free lunch in data privacy.
\newblock In {\em SIGMOD}, pages 193--204, 2011.

\bibitem{krizhevsky2012imagenet}
Alex Krizhevsky, Ilya Sutskever, and Geoffrey~E Hinton.
\newblock Imagenet classification with deep convolutional neural networks.
\newblock In {\em Advances in neural information processing systems}, pages
  1097--1105, 2012.

\bibitem{lee2014top}
Jaewoo Lee and Christopher~W Clifton.
\newblock Top-k frequent itemsets via differentially private fp-trees.
\newblock In {\em Proceedings of the 20th ACM SIGKDD international conference
  on Knowledge discovery and data mining}, pages 931--940, 2014.

\bibitem{lee2015maximum}
Jaewoo Lee, Yue Wang, and Daniel Kifer.
\newblock Maximum likelihood postprocessing for differential privacy under
  consistency constraints.
\newblock In {\em Proceedings of the 21th ACM SIGKDD International Conference
  on Knowledge Discovery and Data Mining}, pages 635--644, 2015.

\bibitem{li2010optimizing}
Chao Li, Michael Hay, Vibhor Rastogi, Gerome Miklau, and Andrew McGregor.
\newblock Optimizing linear counting queries under differential privacy.
\newblock In {\em Proceedings of the twenty-ninth ACM SIGMOD-SIGACT-SIGART
  symposium on Principles of database systems}, pages 123--134, 2010.

\bibitem{Li2016book}
Ninghui Li, Min Lyu, Dong Su, and Weining Yang.
\newblock {\em Differential Privacy: From Theory to Practice}.
\newblock Synthesis Lectures on Information Security, Privacy, and Trust.
  Morgan Claypool, 2016.

\bibitem{li2012privbasis}
Ninghui Li, Wahbeh Qardaji, Dong Su, and Jianneng Cao.
\newblock Privbasis: Frequent itemset mining with differential privacy.
\newblock {\em Proceedings of the VLDB Endowment}, 5(11):1340--1351, 2012.

\bibitem{lyu2016understanding}
Min Lyu, Dong Su, and Ninghui Li.
\newblock Understanding the sparse vector technique for differential privacy.
\newblock {\em arXiv preprint arXiv:1603.01699}, 2016.

\bibitem{Lyu2017vldb}
Min Lyu, Dong Su, and Ninghui Li.
\newblock Understanding the sparse vector technique for differential privacy.
\newblock {\em {PVLDB}}, 10(6):637--648, 2017.

\bibitem{mckenna2019graphical}
Ryan Mckenna, Daniel Sheldon, and Gerome Miklau.
\newblock Graphical-model based estimation and inference for differential
  privacy.
\newblock In {\em International Conference on Machine Learning}, pages
  4435--4444, 2019.

\bibitem{mcsherry2007mechanism}
Frank McSherry and Kunal Talwar.
\newblock Mechanism design via differential privacy.
\newblock In {\em 48th Annual IEEE Symposium on Foundations of Computer Science
  (FOCS'07)}, pages 94--103. IEEE, 2007.

\bibitem{moosavi2019accident}
Sobhan Moosavi, Mohammad~Hossein Samavatian, Srinivasan Parthasarathy, Radu
  Teodorescu, and Rajiv Ramnath.
\newblock Accident risk prediction based on heterogeneous sparse data: New
  dataset and insights.
\newblock In {\em Proceedings of ACM SIGSPATIAL'19}, pages 33--42.

\bibitem{nist-challenge}
NIST.
\newblock 2018 differential privacy synthetic data challenge.
\newblock
  \url{https://www.nist.gov/ctl/pscr/open-innovation-prize-challenges/past-prize-challenges/2018-differential-privacy-synthetic}.

\bibitem{qardaji2014priview}
Wahbeh Qardaji, Weining Yang, and Ninghui Li.
\newblock Priview: practical differentially private release of marginal
  contingency tables.
\newblock In {\em Proceedings of the 2014 ACM SIGMOD international conference
  on Management of data}, pages 1435--1446. ACM, 2014.

\bibitem{rogers2020linkedin}
Ryan Rogers, Subbu Subramaniam, Sean Peng, David Durfee, Seunghyun Lee,
  Santosh~Kumar Kancha, Shraddha Sahay, and Parvez Ahammad.
\newblock Linkedin's audience engagements api: A privacy preserving data
  analytics system at scale.
\newblock {\em arXiv preprint arXiv:2002.05839}, 2020.

\bibitem{tantipongpipat2019differentially}
Uthaipon Tantipongpipat, Chris Waites, Digvijay Boob, Amaresh~Ankit Siva, and
  Rachel Cummings.
\newblock Differentially private mixed-type data generation for unsupervised
  learning.
\newblock {\em arXiv preprint arXiv:1912.03250}, 2019.

\bibitem{ullman2011pcps}
Jonathan Ullman and Salil Vadhan.
\newblock Pcps and the hardness of generating private synthetic data.
\newblock In {\em Theory of Cryptography Conference}, pages 400--416. Springer,
  2011.

\bibitem{vadhan2017complexity}
Salil Vadhan.
\newblock The complexity of differential privacy.
\newblock In {\em Tutorials on the Foundations of Cryptography}, pages
  347--450. Springer, 2017.

\bibitem{vietrinew}
Giuseppe Vietri, Grace Tian, Mark Bun, Thomas Steinke, and Zhiwei~Steven Wu.
\newblock New oracle-efficient algorithms for private synthetic data release.

\bibitem{WXYZGY17}
Ning Wang, Xiaokui Xiao, Yin Yang, Zhenjie Zhang, Yu~Gu, and Ge~Yu.
\newblock Privsuper: A superset-first approach to frequent itemset mining under
  differential privacy.
\newblock In {\em Data Engineering (ICDE), 2017 IEEE 33rd International
  Conference on}, pages 809--820. IEEE, 2017.

\bibitem{wang2019consistent}
Tianhao Wang, Milan Lopuha{\"a}-Zwakenberg, Zitao Li, Boris Skoric, and Ninghui
  Li.
\newblock Locally differentially private frequency estimation with consistency.
\newblock In {\em NDSS}, 2020.

\bibitem{xiao2010differential}
Xiaokui Xiao, Guozhang Wang, and Johannes Gehrke.
\newblock Differential privacy via wavelet transforms.
\newblock {\em IEEE Transactions on knowledge and data engineering},
  23(8):1200--1214, 2010.

\bibitem{zhang2017privbayes}
Jun Zhang, Graham Cormode, Cecilia~M Procopiuc, Divesh Srivastava, and Xiaokui
  Xiao.
\newblock Privbayes: Private data release via bayesian networks.
\newblock {\em ACM Transactions on Database Systems (TODS)}, 42(4):25, 2017.

\bibitem{zhang2018differentially}
Xinyang Zhang, Shouling Ji, and Ting Wang.
\newblock Differentially private releasing via deep generative model.
\newblock {\em arXiv preprint arXiv:1801.01594}, 2018.

\end{thebibliography}
}

\appendix
\section{Exponential Mechanism}
\label{app:em}
The Exponential Mechanism (EM) computes a function $f$ on $D$ by sampling from the set of all possible answers in the range of $f$ according to an exponential distribution, with answers that are ``more accurate'' will be sampled with higher probability.  This is generally referred to as the {exponential mechanism}~\cite{mcsherry2007mechanism}.
This approach requires the specification of a set of quality functions $q_i: \mathbb{D} \rightarrow \mathbb{R}$.  Given the quality functions $q_1$ to $q_m$, the global sensitivity $\Delta_q$ is defined as:
$$ 
\Delta_q = \max_{i} \max_{(D,D'):D\simeq D'} |q_i(D) - q_i(D')| $$
The following method $\AA$ satisfies $\epsilon$-differential privacy:
\begin{equation}
 \Pr{\AA_q(D)=i} \propto \exp{\left(\frac{\epsilon }{2\,\Delta_q}q_i(D)\right)} \label{eq:exp}
\end{equation}

\section{Laplace Mechanism}
\label{app:lap}
The Laplace mechanism computes a function $f$ on the input $D$ in a differentially privately way, by adding to $f(D)$ a random noise.  The magnitude of the noise depends on $\mathsf{\Delta}_f$, the \emph{global sensitivity} or the $L_1$ sensitivity of $f$.  When $f$ outputs a single element, such a mechanism $\AA$ is given below:
$$
\begin{array}{crl}
& \AA_f(D) & =f(D) + \Lapp{\frac{\mathsf{\Delta}_f}{\epsilon}}
  \\
\mbox{where} &  \mathsf{\Delta}_f & = \max\limits_{(D,D') : D\simeq D'} ||f(D) - f(D')||_1,
\\
\mbox{and}& \Pr{\Lapp{\beta}=x} & = \frac{1}{2\beta} e^{-|x|/\beta}
\end{array}
$$
In the above, $\Lapp{\beta}$ denotes a random variable sampled from the Laplace distribution with scale parameter $\beta$.  When $f$ outputs a vector, $\AA$ adds a fresh sample of $\Lapp{\frac{\mathsf{\Delta}_f}{\epsilon}}$ to each element of the vector.

\section{Composition Property of \zcdp}
\label{app:zcdp_composition}

\zcdp has a simple linear composition property~\cite{bun2016concentrated}:
\begin{prop}
\label{prop:zcdp-compose}
Two randomized mechanisms $\mathcal{A}_1$ and $\mathcal{A}_2$ satisfy $\rho_1$-\zcdp and $\rho_2$-\zcdp respectively, their sequential composition  $\mathcal{A} = (\mathcal{A}_1,\mathcal{A}_2)$ satisfies ($\rho_1+\rho_2$)-\zcdp.
\end{prop}

The following two propositions restates the results from~\cite{bun2016concentrated}, which are useful for composing Gaussian mechanisms in differential privacy.

\begin{prop} 
\label{prop:CDPtoDP}
If $\AA$ provides $\rho$-\zcdp, then $\AA$ is $(\rho+2\sqrt{\rho\log(1/\delta)},\delta)$-differentially private for any $\delta>0$.
\end{prop}

\begin{prop}
\label{prop:Gaussian-zcdp}
The Gaussian mechanism which answers $f(D)$ with noise $\mathcal{N}(0, \Delta_f ^2 \sigma^2 \mathbf{I})$ satisfies ($\frac{1 }{2\sigma^2}$)-\zcdp.
\end{prop}

Given the privacy constraint $\epsilon$ and $\delta$, we can calculate the amount of noise for each task using Propositions~\ref{prop:zcdp-compose} to~\ref{prop:Gaussian-zcdp}.
In particular, we first use Proposition~\ref{prop:CDPtoDP} to compute the total $\rho$ allowed.  
Then we use Proposition~\ref{prop:zcdp-compose} to allocate $\rho_i$ for each task $i$.  
Finally, we use Proposition~\ref{prop:Gaussian-zcdp} to calculate $\sigma$ for each task.

\revisionstart
\begin{theorem}
\label{thm:sigma_calculation}
Given privacy budget $\epsilon$, $\delta$ and the number of tasks $m$, the standard deviation for each task is $\sigma=\left( \sqrt{2m\Delta_f^2\log{\frac{1}{\delta}}} + \sqrt{2m\Delta_f^2\log{\frac{1}{\delta}} + 2m\epsilon\Delta_f^2} \right) / 2\epsilon$.
\end{theorem}

\begin{proof}
Proposition~\ref{prop:CDPtoDP} states that $\rho$-\zcdp is equivalent to $(\rho+2\sqrt{\rho\log(1/\delta)},\delta)$-DP; thus we have $\epsilon = \rho+2\sqrt{\rho\log(1/\delta)}$.
Rearranging the above equation, we have
\begin{align}
\label{eq:cdp2dp}
    \sqrt{\rho}^2 + 2\sqrt{\log{\frac{1}{\delta}}} \cdot \sqrt{\rho} - \epsilon = 0
\end{align}

By solving Equation~\ref{eq:cdp2dp}, we get the relationship between $\rho$ and $\epsilon, \delta$:
$$
\sqrt{\rho} = \sqrt{\log{\frac{1}{\delta}} + \epsilon} - \sqrt{\log{\frac{1}{\delta}}}
$$

Assume the total privacy budget for \zcdp is $\rho$, it is obvious that the privacy budget for each task is $\rho_0=\frac{\rho}{m}$ based on Proposition~\ref{prop:zcdp-compose}.
From Proposition~\ref{prop:Gaussian-zcdp}, we have
\begin{align*}
    \sigma_0 
    &= \sqrt{\frac{\Delta_f^2}{2\rho_0^2}} = \sqrt{\frac{m\Delta_f^2}{2\rho}} \\
    &= \frac{\sqrt{k\Delta_f^2}}{\sqrt{2}\left( \sqrt{\log{\frac{1}{\delta}} + \epsilon} - \sqrt{\log{\frac{1}{\delta}}} \right)} \\
    &= \frac{\sqrt{2m\Delta_f^2\log{\frac{1}{\delta}}} + \sqrt{2m\Delta_f^2\log{\frac{1}{\delta}} + 2m\epsilon\Delta_f^2}}{2\epsilon}
\end{align*}
\end{proof}
\revisionend

Compared with $(\epsilon, \delta)$-DP, \zcdp provides a tighter bound on the cumulative privacy loss under composition, making it more suitable for algorithms consist of a large number of tasks.

\section{Missing Proofs}
\label{appendix:proofs}

\mypara{Proof of Lemma \ref{lemma:l1_sensitivity}: $\Delta_{\margselect}=4$}
\begin{proof}
Assume $D$ contains $n$ records and consider the two attributes $a$ and $b$.  
Denote the number of users for histogram on attribute $a$ as $a_1, a_2, \ldots$, and $b_1, b_2, \ldots$ for $b$.  
For the two-way marginal on $a, b$, denote the number of users for it as $\alpha_{11}, \alpha_{12}, \ldots$.

\revisionstart
The metric $\margselect_{ab}$ is 
\begin{align*}
    \margselect_{ab} = \sum_{ij} \abs{\frac{a_ib_j}{n} - \alpha_{ij}}
\end{align*}

If we add one user (wlog, whose values for $a$ and $b$ are $x$ and $y$), 
\begin{align*}
    \margselect_{ab}'
    = & \sum_{i\neq x, j\neq y} \abs{\frac{a_ib_j}{n+1} - \alpha_{ij}} \\
    & + \sum_{i\neq x} \abs{\frac{a_i(b_y+1)}{n+1} - \alpha_{iy}} \\
    & + \sum_{j\neq y} \abs{\frac{(a_x+1)b_j}{n+1} - \alpha_{xj}} \\
    & + \abs{\frac{(a_x+1)(b_y+1)}{n+1} - (\alpha_{xy}+1)}
\end{align*}

Since $\abs{s} - \abs{t} \leq \abs{s - t}$, the sensitivity is given by

\begin{align}
    &\Delta_{\margselect}
    = |\margselect_{ab} - \margselect_{ab}'| \nonumber \\
    \leq & \sum_{i\neq x, j\neq y} \abs{\frac{a_ib_j}{n(n+1)}} \nonumber \\
    & + \sum_{i\neq x} \abs{\frac{a_ib_y}{n(n+1)} - \frac{a_i}{n+1}} \nonumber \\
    & + \sum_{j\neq y} \abs{\frac{a_xb_j}{n(n+1)} - \frac{b_j}{n+1}} \nonumber \\
    & + \abs{\frac{(n+1)a_xb_y - n(a_x+1)(b_y+1) + n(n+1)}{n(n+1)}} \nonumber \\
    = & \frac{\sum_{i\neq x, j\neq y} a_ib_j - \sum_{i\neq x}(a_ib_y - na_i) - \sum_{j\neq y}(a_xb_j - nb_j)}{n(n+1)} \nonumber \\
    & + \frac{(n+1)a_xb_y - n(a_x+1)(b_y+1) + n(n+1)}{n(n+1)} \label{eq:absolute} \\
    = & \frac{(n-a_x)(n-b_y) - (n-a_x)(b_y-n) - (a_x-n)(n-b_y)}{n(n+1)} \nonumber \\
    & + \frac{(n+1)a_xb_y - n(a_x+1)(b_y+1) + n(n+1)}{n(n+1)} \label{eq:distribution} \\
    = & \frac{4\left( n^2 - (a_x+b_y)n + a_xb_y\right)}{n(n+1)} \nonumber \\
    = & \frac{4(n-a_x)(n-b_y)}{n(n+1)} \leq 4 \nonumber
\end{align}

In the above derivation, Equation~\eqref{eq:absolute} is due to $\frac{a_ib_j}{n(n+1)} \geq 0$, $\frac{a_ib_y}{n(n+1)} - \frac{a_i}{n+1} \leq 0$, $\frac{a_xb_j}{n(n+1)} - \frac{b_j}{n+1} \leq 0$ and $\frac{(n+1)a_xb_y - n(a_x+1)(b_y+1) + n(n+1)}{n(n+1)} = \frac{(n-a_x)(n-b_y)}{n(n+1)} \geq 0$.
Equation~\eqref{eq:distribution} is due to $\sum_{i\neq x}a_i=n-a_x$, $\sum_{j\neq y}b_j=n-b_y$ and $\sum_{i\neq x, j\neq y}a_ib_j=(n-a_x)(n-b_y)$.

\revisionend

\end{proof}

\mypara{Proof of Theorem \ref{thm:dependency_privacy}: 
Publishing $m$ \margselect scores with $\mathcal{N}(0, 8m/\rho' \mathbf{I})$ satisfies $\rho'$-\zcdp}
\begin{proof}
The proof is trivial given Lemma~\ref{lemma:l1_sensitivity}, Propositions~\ref{prop:zcdp-compose} and~\ref{prop:Gaussian-zcdp}: Because the sensitivity of \margselect is $4$, publishing it with $\mathcal{N}(0, 8m/\rho' \mathbf{I})$ satisfies $\rho'/m$-\zcdp.
For $m$ \margselect scores, by composition, publishing all of them satisfies $\rho'$-\zcdp.
\end{proof}

\mypara{Proof of Theorem \ref{thm:marginal_gauss}: (1) The marginal $\mathsf{M}$ has sensitivity $\Delta_{\mathsf{M}}=1$; 
(2) Publishing $\mathsf{M}$ with noise $\mathcal{N}(0, 1/2\rho \mathbf{I})$ satisfies $\rho$-\zcdp}
\begin{proof}
We first prove the marginal function has sensitivity $1$.  A marginal $\mathsf{M}_{A}$ specified by a set of attributes $A$ is a frequency distribution table, showing the number of record with each possible combination of values for the attributes.  For two marginals $\mathsf{M}_{A}$ and $\mathsf{M}_{A}'$, where $\mathsf{M}_{A}'$ is obtained by adding or removing one user to $\mathsf{M}_{A}$.  In general, for any $A$, it is obviously 
\[
\Delta_{\mathsf{M}} = |\mathsf{M} - \mathsf{M}'|_2 = 1
\]

Given this fact, by Propositions~\ref{prop:Gaussian-zcdp}, it is trivial that adding $\mathcal{N}(0, 1/2\rho \mathbf{I})$ to a marginal satisfies $\rho$-\zcdp.
\end{proof}

\begin{figure*}[!h]
    \centering

    \subfloat[pair-wise marginal]{\includegraphics[width=0.3\textwidth]{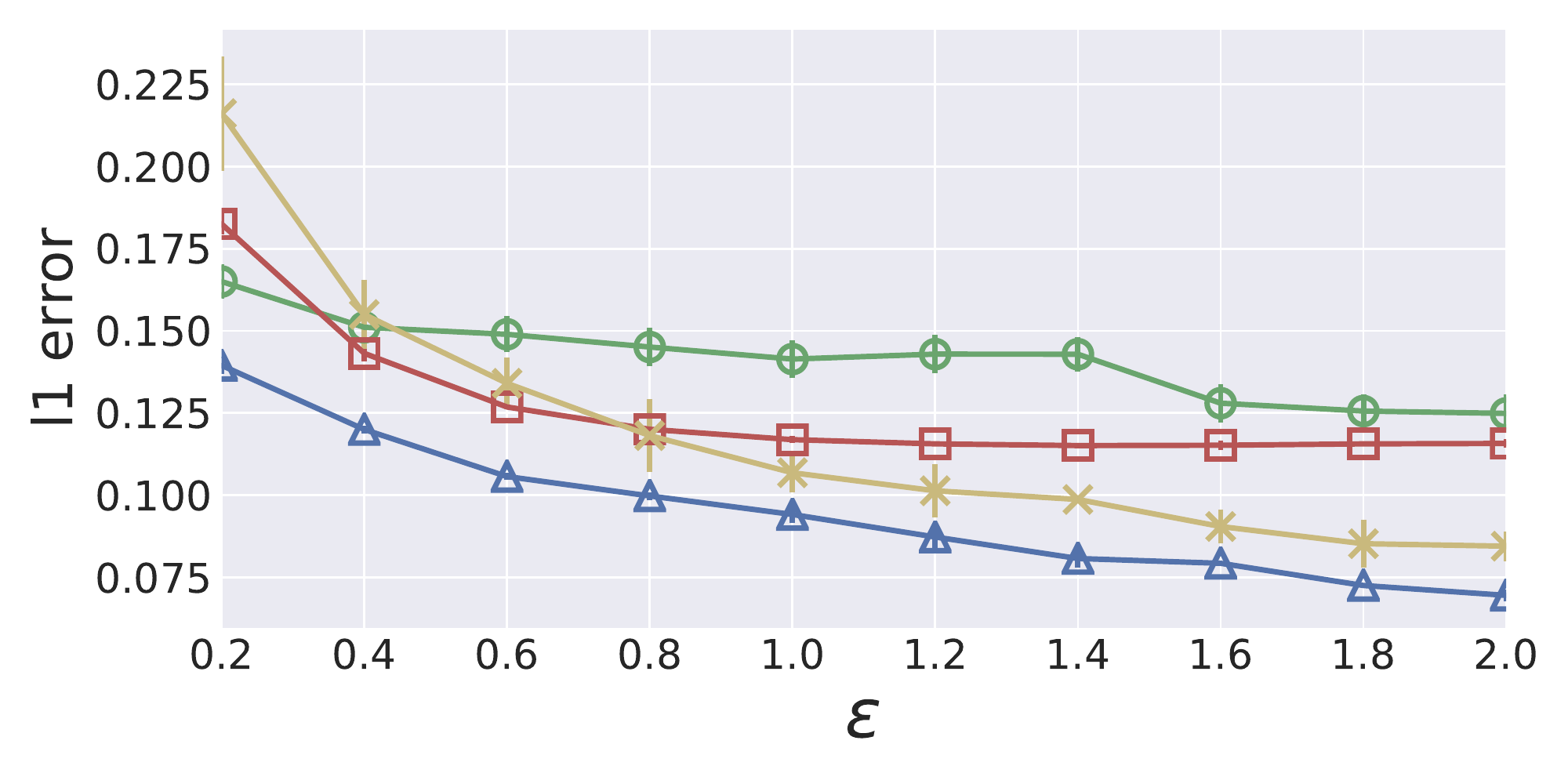}}
    \subfloat[range query]{\includegraphics[width=0.3\textwidth]{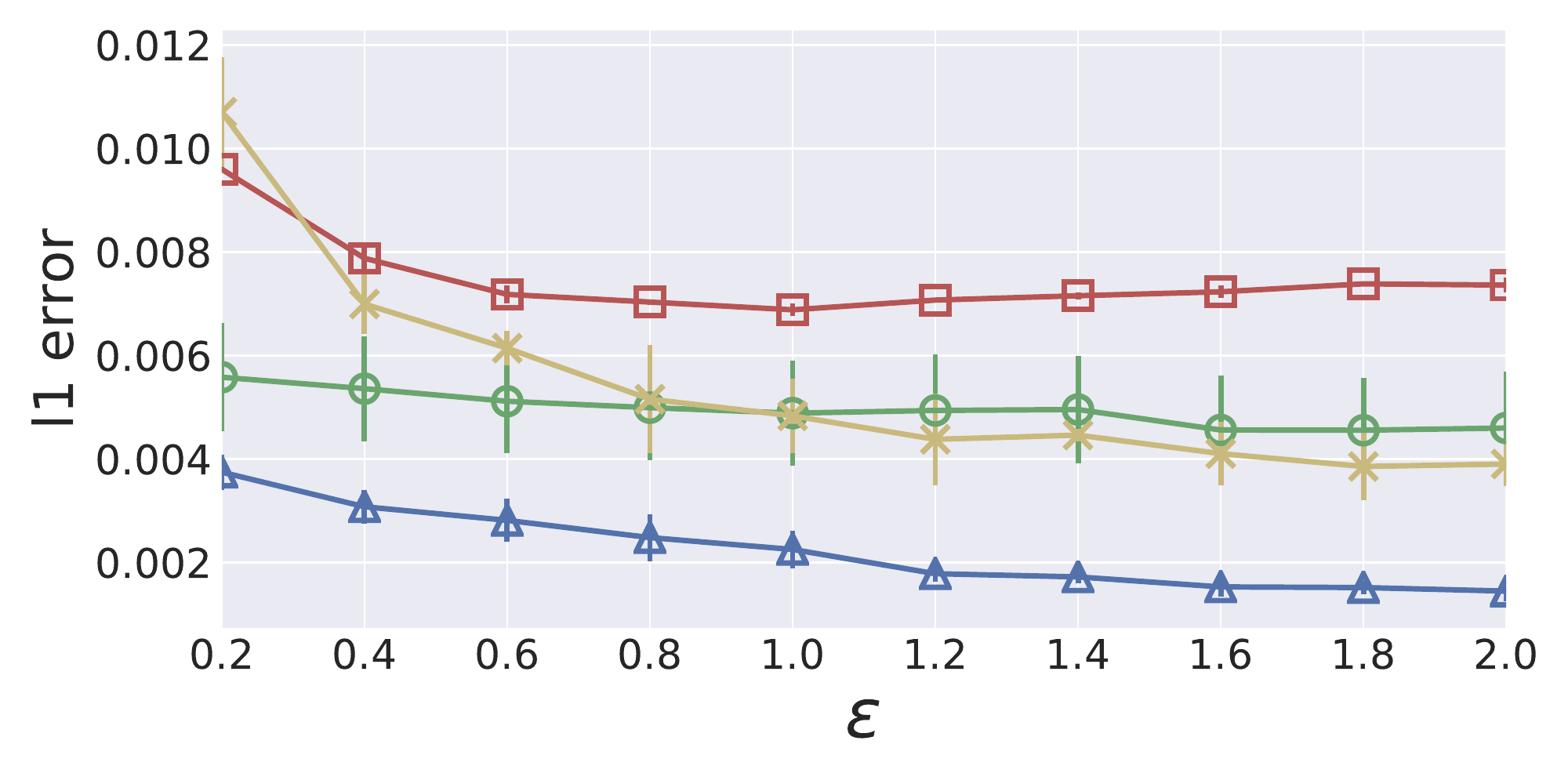}}
    \subfloat[classification]{\includegraphics[width=0.3\textwidth]{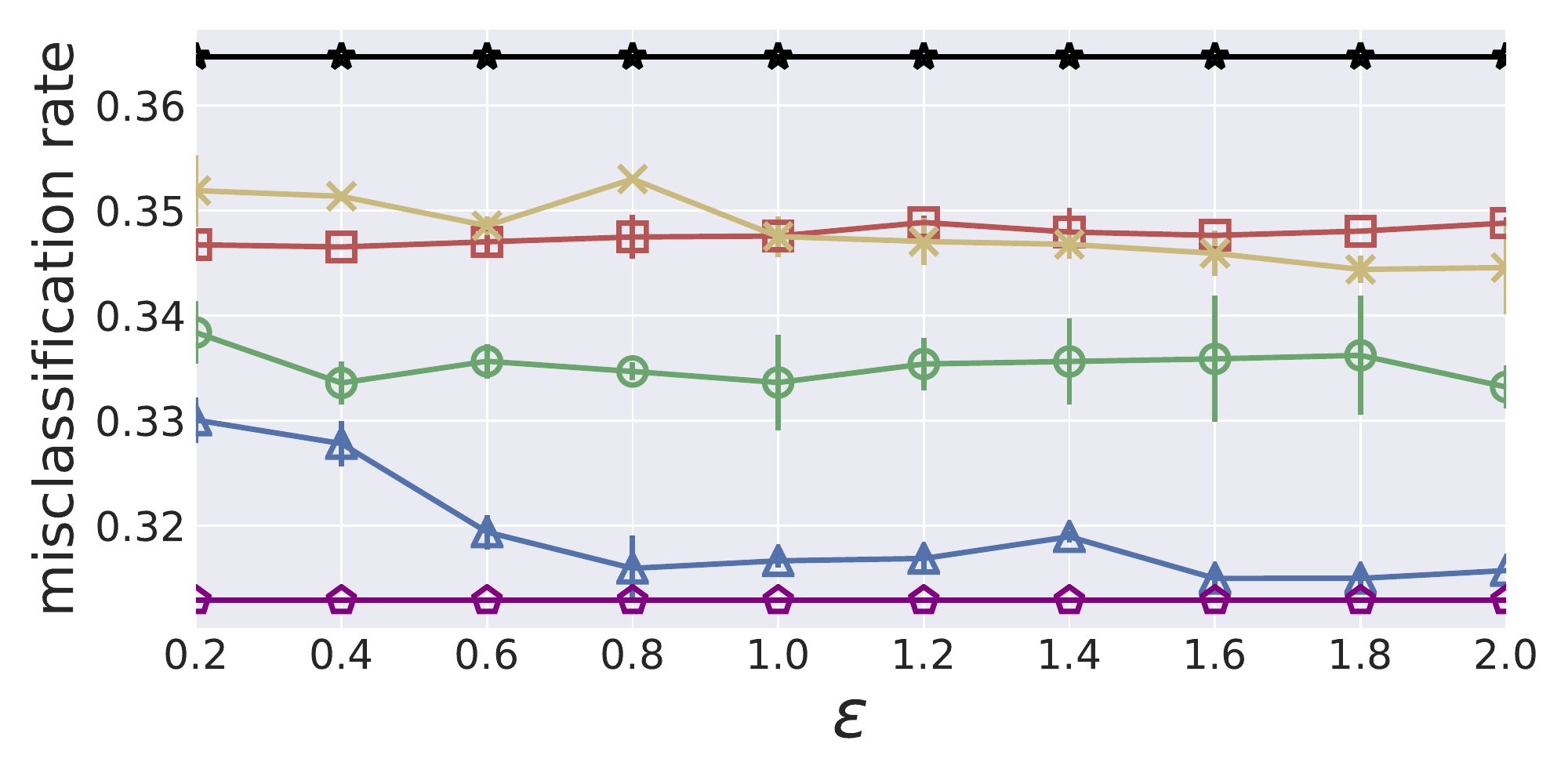}} \\ 
    
    \centering{US Accident} \\ [-2ex]

    \subfloat[pair-wise marginal]{\includegraphics[width=0.3\textwidth]{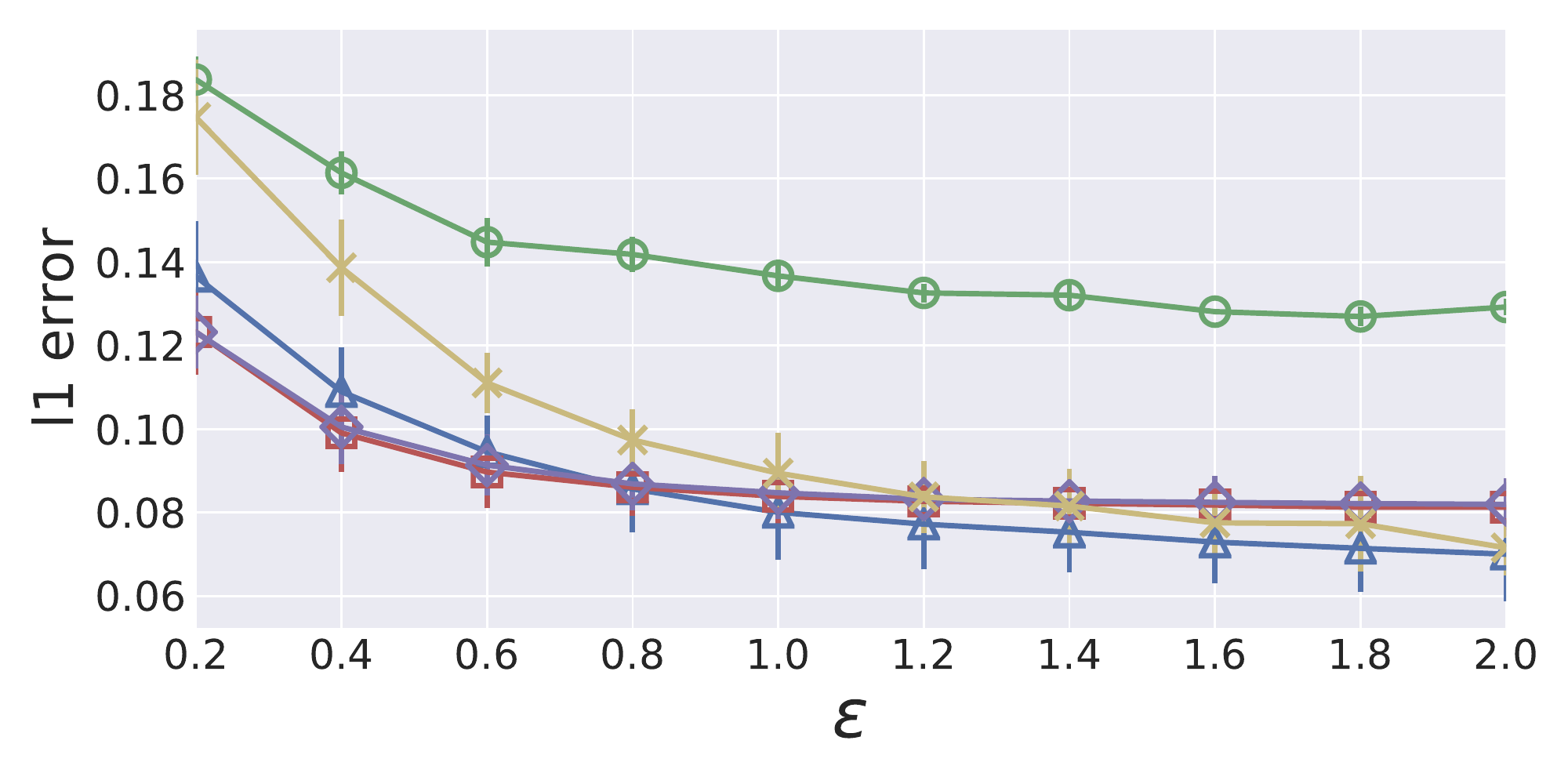}}
    \subfloat[range query]{\includegraphics[width=0.3\textwidth]{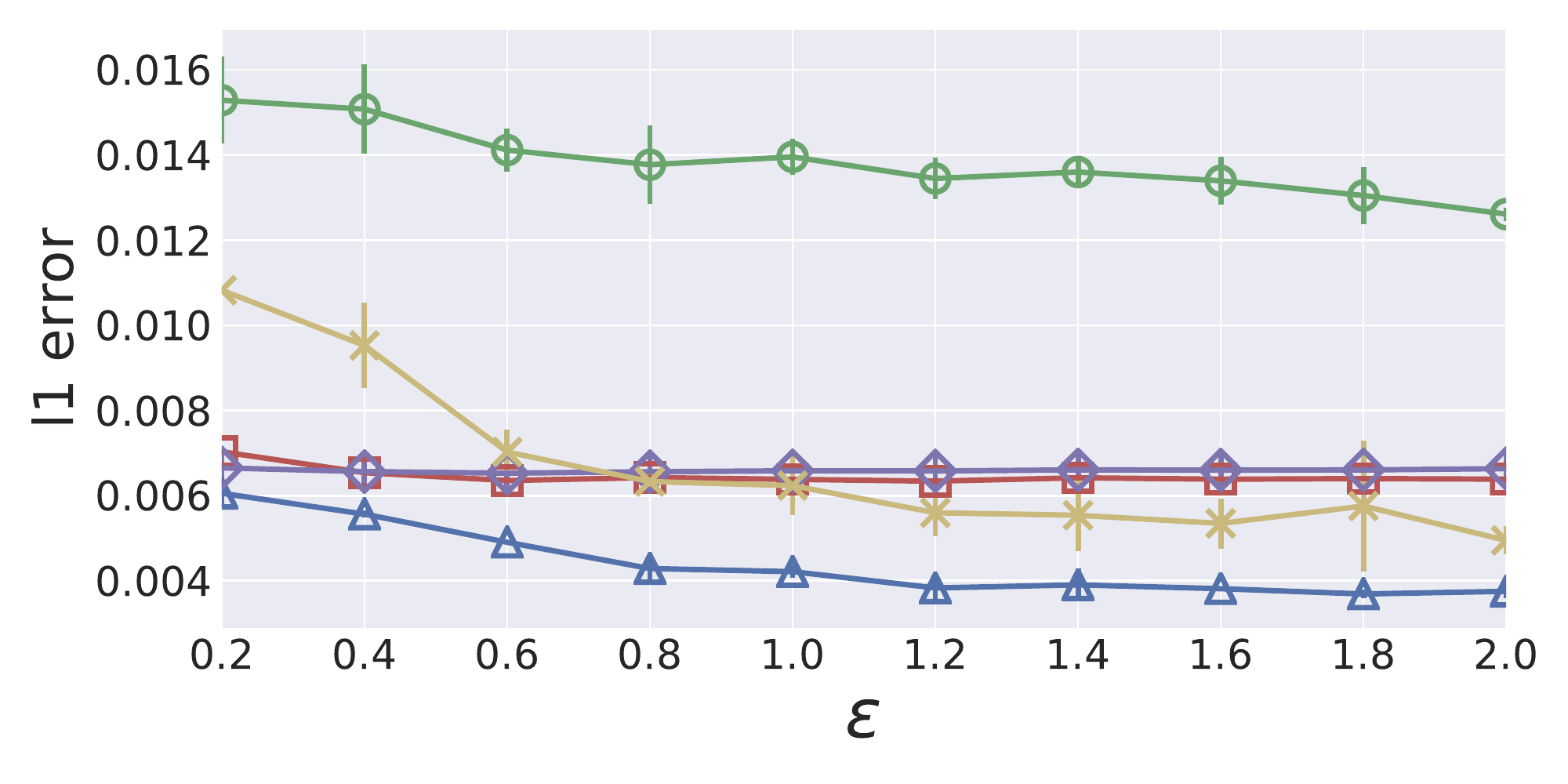}}
    \subfloat[classification]{\includegraphics[width=0.3\textwidth]{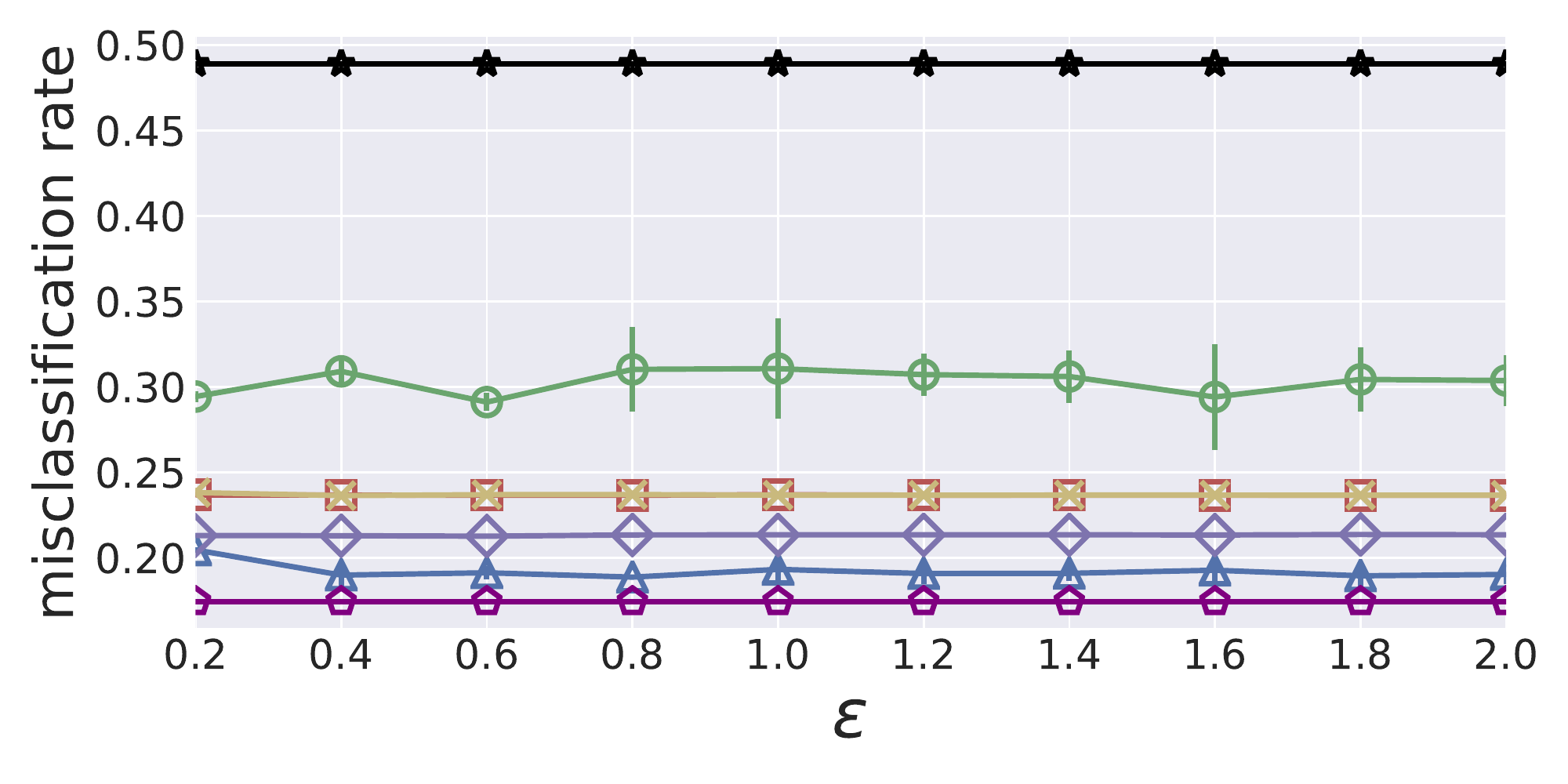}} \\ 
    
    \centering{Colorado} \\
    
    \subfloat{\includegraphics[width=0.8\textwidth]{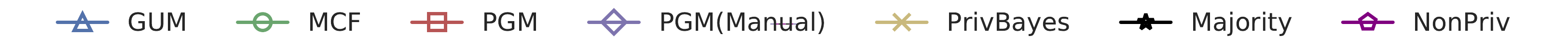}}  \\ [-2ex]

    \caption{Comparison of different synthesis methods.  
    \gum is our proposed method.}
    \label{fig:comparison_synthesis}
\end{figure*}

\revisionstart

\section{Comparison of InDif and Entropy-based Metrics}
\label{app:comparison_metrics}

To evaluate the impact of noise on the dependency metrics, one should consider both sensitivity and the range of the metrics.
In this section, we compare \margselect with two dependency metrics in the literature with respect to sensitivity and range.

\mypara{Mutual Information (MI)~\cite{zhang2017privbayes}}
\privbayes adopts mutual information to measure the dependency between attributes.
For attribute $A$ and $B$, their mutual information $I(A; B)$ is defined as~\footnote{All logarithms used in this section are to the base 2.}
\begin{align*}
    \sum_{a \in dom(A)}\sum_{b \in dom(B)}\Pr{A=a, B=b}\log\frac{\Pr{A=a, B=b}}{\Pr{A=a}\Pr{B=b}}
\end{align*}

From~\cite{zhang2017privbayes}, we know that the sensitivity of MI is $\frac{2}{n}\log\frac{n+1}{2} + \frac{n-1}{n}\log\frac{n+1}{n-1}$.
Besides, the range of MI is $[0, \log d]$, where $d = \max\{d_A, d_B\}$, $d_A$ and $d_B$ are the number of possible values for attribute $A$ and $B$, respectively.
Thus, the noise-range ratio of MI is defined as
\begin{align*}
    R_{MI} = \frac{1}{n} \cdot \frac{2\log\frac{n+1}{2} + (n+1)\log\frac{n+1}{n-1}}{\log d}
\end{align*}

\mypara{Symmetrical Uncertainty Coefficient (SUC)~\cite{bindschaedler2017plausible}}
\bsg adopts symmetrical uncertainty coefficient to measure the dependency between attributes, which is defined as
\begin{align*}
    corr(A, B) = 2 - 2\frac{H(A, B)}{H(A) + H(B)}
\end{align*}
where $H(\cdot)$ is the entropy function.

To achieve differential privacy, the authors in~\cite{bindschaedler2017plausible} propose to add noise to three entropy values in corr(A, B), respectively.
The authors prove that the sensitivity of entropy is $\frac{1}{n}\left[2 + \frac{1}{\ln2} + 2\log n\right]$.
Besides, the range of entropy is $[0, \log d]$.
Thus, the noise-range ratio of entropy is given by
\begin{align*}
    R_{Ent} = \frac{1}{n} \cdot \frac{2 + \frac{1}{\ln 2} + 2\log n}{\log d}
\end{align*}

\mypara{Comparison with InDif}
Recall that the sensitivity and range of \margselect is $4$ and $[0, 2n]$, respectively; thus, its noise-range ratio is given by
\begin{align*}
    R_{InDif} = \frac{1}{n} \cdot 2
\end{align*}

We list the noise-range ratio of three methods in Table~\ref{tab:noise_range_ratio} when $d$ varies.
We set $n=600000$ which is the case of three datasets in our experiments.
We observe that the noise-range ratio of \margselect is consistently smaller than the other two methods when $d \leq 100000$.
In the three datasets in our experiments, most of the attributes contains less than $100$ possible values, and the noise-range ratio of \margselect is $3$ times smaller than the other two methods.

\mypara{Comparison of Relative Errors}
To further evaluate the impact of noise on real-world datasets, we compare the relative errors between true values and noisy values of different metrics in Table~\ref{tab:relative_error} when $\epsilon=2.0$.
The relative errors are calculated as $\frac{1}{m}\sum_{i=1}^m\left|\frac{s_i - \Tilde{s_i}}{s_i}\right|$, where $m$ is the total number of pairwise marginals, $s_i$ and $\Tilde{s_i}$ are the true value and noisy value of marginal $i$, respectively.
We run each experiment $1000$ times and report the average relative error.

The experimental results show that the relative errors of \margselect are significantly smaller than MI and SUC.
The reason is that most of the MI values and SUC values are much smaller than their maximal value $\log{C}$, while most of the \margselect values are close to their maximal value $2n$.
For example, in the Colorado dataset, $78\%$ of the MI values and $87\%$ of the SUC values are smaller than $0.1$ (much smaller than $\log{C}$). In another hand, $37\%$ of the InDif values are larger than $0.5n$ (close to $2n$).

\begin{table}[!h]
    \footnotesize
	\centering
	\begin{tabular}{c|c|c|c|c|c|c}
		\toprule
		$d$	& $2$ & $50$ & $100$ & $1000$ & $10000$ & $100000$ \\ 
		\midrule
		$n \cdot R_{InDif}$ &$2.0$ & $2.0$ & $2.0$ & $2.0$ & $2.0$ & $2.0$ \\
		$n \cdot R_{MI}$	& $39.3$ & $7.0$ & $5.9$ & $3.9$ & $3.0$ & $2.4$ \\
		$n \cdot R_{Ent}$ & $41.8$ & $7.4$ & $6.3$ & $4.2$ & $3.1$ & $2.5$ \\
		\bottomrule
	\end{tabular}
	\caption{\revision{}{Noise-range ratio of different metrics when $n=600000$ and $d$ is varying.}}
	\label{tab:noise_range_ratio}
\end{table}

\begin{table}[!h]
    \footnotesize
	\centering
	\begin{tabular}{c|c|c|c|c}
		\toprule
		& Adult & Accident & Loan & Colorado \\ 
		\midrule
		InDif   & $0.017$ & $0.028$ & $0.161$ & $0.137$ \\
		MI	    & $34$ & $314$ & $543$ & $735$ \\
		SUC     & $396$ & $2858$ & $4205$ & $6933$ \\
		\bottomrule
	\end{tabular}
	\caption{\revision{}{Relative error of different metrics when $\epsilon=2.0$.}}
	\label{tab:relative_error}
\end{table}

\section{Computational Complexity Analysis}

In this section, we first theoretically analyze the computational complexity of different methods, and then empirically evaluate the running time and memory consumption.

\mypara{Time Complexity}
The computational time for all methods consist of two parts, marginal selection and dataset generation.

For \privbayes, the marginals are selected by constructing a Bayesian network.
The general idea is to start with a randomly selected node, then gradually add node to the Bayesian network that maximally increase MI of the selected nodes.
To reduce time complexity, \privbayes only consider at most $\gamma$ parents nodes in the selected nodes for each newly added node.
The number of pairs considered in iteration $i$ is $(d-i){i \choose \gamma}$, where $d$ is the number of attributes; thus summing over all iterations the computational complexity is bounded by $d\sum_{i=1}^d{i \choose \gamma}=d{d+1 \choose \gamma+1}$.
In the dataset generation step, \privbayes simply sample records one-by-one using the Bayesian network; thus the time complexity is $O\left( nd \right)$, where $n$ is the number of synthetic records.

For \pgm, except for marginal selection and dataset generation, it includes another component that learn the parameters of Markov random field.
The general idea is to use all the marginals and gradient decent technique to update the parameters.
The gradient decent process would repeat $t_{pg}$ times until convergence.
In practice, $t_{pg}$ is always set to be larger than $10000$.
Thus, the time complexity for learning Markov random field is $O\left( t_{pg}k_{pg} \right)$, where $k_{pg}$ is the number of marginals.
The time complexity for generating synthetic dataset is the same with \privbayes, \ie, $O\left( nd \right)$.
Notice that \pgm does not provide method to select marginals, we only report the time complexity for parameter learning and dataset generation in Table~\ref{tab:computational_complexity}.

For \method, there are $m = {d \choose 2} = \frac{d(d-1)}{2}$ possible pairwise marginals in the marginal selection step.
In iteration $i$ of Algorithm~\ref{alg:marginal_selection}, we need to check $m-i$ pairwise marginals;
thus, the time complexity is $\sum_{i=1}^{k_{ps}}(m-i)=k_{ps}m - \frac{k_{ps}(k_{ps}+1)}{2}=O\left( k_{ps}d^2 \right)$.
In the dataset generation step, we should go through all marginals $t_{ps}$ times to ensure consistency.
Thus, the time complexity is $t_{ps}k_{ps}$ and we typically set $t_{ps}=100$ in practice.

\mypara{Space Complexity}
The memory consumption of all methods consist of two parts, marginal tables and synthetic dataset.
The memory consumption of synthetic dataset for all methods are the same, \ie, $O\left( nd \right)$.
The memory consumption for marginal tables differs in the number of marginals $k_{\star}$ and the average number of cells for each marginal $C_{\star}$.
Specifically, \privbayes contains $d-1$ marginals where each marginal contains at most $\gamma + 1$ attributes.
The number of marginals for \pgm is unlimited; however, when the number of marginals is large, the Markov random field can be dense, resulting in large clique in the induced junction tree, which can be prohibitively large.
\method uses the $2$-way marginal; thus the average number of cells in each marginal is relatively small.
The number of marginals is typically in the range of $[100, 700]$ in our experiment.

\mypara{Empirical Evaluation}
Table~\ref{tab:running_time} and Table~\ref{tab:memory_consumption} illustrate the running time and memory consumption for all methods on four datasets in our experiment.

The empirical running time in Table~\ref{tab:running_time} shows that \privbayes performs best in terms of running time, since it requires only $d-1$ marginals and the sampling process is very fast.
\pgm uses the same set of marginals with \privbayes, while it needs additional time to learn the parameters of Markov random field, and the gradient decent process should repeat more than $10000$ times.
\method is slower than \privbayes and \pgm since it uses much more marginals.
For example, when $\epsilon=2.0$, the Colorado dataset has about $700$ marginals, while \privbayes and \pgm only have $96$ marginals.
Although \method costs more time than \privbayes and \pgm, it only takes less than 4 hours to generate large dataset such as Colorado ($97$ attributes with total domain of $5\cdot 10^{162}$), which is acceptable in practice considering its superior performance.

The empirical memory consumption in Table~\ref{tab:memory_consumption} shows that the memory consumption for all methods are similar for the same dataset.
The reason is that the memory consumption for all methods are dominated by the storage of synthetic datasets, and the storage of marginal tables are less than $10$ Megabytes for all datasets.

\begin{table}[!h]
    \footnotesize
	\centering
	\begin{tabular}{c|c|c}
		\toprule
		        	& Time Complexity & Space Complexity \\ 
		\midrule
		\privbayes  & $O\left( d{d+1 \choose \gamma+1} + nd \right)$ & $O\left( C_{pb}d + nd\right)$   \\
		\pgm	    & $O\left( t_{pg}k_{pg} + nd \right)$ & $O\left( C_{pg}k_{pg} + nd\right)$  \\
		\method     & $O\left( k_{ps}d^2 + t_{ps}k_{ps} \right)$ & $O\left( C_{ps}k_{ps} + nd\right)$   \\
		\bottomrule
	\end{tabular}
	\caption{\revision{}{Comparison of computational complexity for different methods.
	$n, d, k_{\star}$ stand for the number of records in synthetic dataset, the number of attributes and the number of marginals, respectively;
	$C_{\star}$ stands for the average number of cells in each marginal;
	$t_{\star}$ stands for the number of required iterations in each method.}}
	\label{tab:computational_complexity}
\end{table}

\begin{table}[!h]
    \footnotesize
	\centering
	\begin{tabular}{c|c|c|c|c}
		\toprule
		Datasets	& Adult & Accident & Loan & Colorado \\ 
		\midrule
		\privbayes  & 1 min  & 2 min & 7 min & 10 min  \\
		\pgm	    & 4 min  & 18 min & 40 min & 1 h 10 min  \\ 
		\method     & 4 min  & 40 min & 2 h 10 min & 3 h 30 min  \\
		\bottomrule
	\end{tabular}
	\caption{\revision{}{Comparison of running time for different methods.}}
	\label{tab:running_time}
\end{table}

\begin{table}[!h]
    \footnotesize
	\centering
	\begin{tabular}{c|c|c|c|c}
		\toprule
		Datasets	& Adult & Accident & Loan & Colorado \\ 
		\midrule
		\privbayes  & $0.06$ & $0.13$ & $0.36$ & $0.43$  \\
		\pgm	    & $0.06$ & $0.13$ & $0.36$ & $0.43$  \\
		\method     & $0.06$ & $0.13$ & $0.36$ & $0.43$  \\
		\bottomrule
	\end{tabular}
	\caption{\revision{}{Comparison of memory consumption of different methods.
	The unit is Gigabytes.}}
	\label{tab:memory_consumption}
\end{table}

\revisionend

\section{Combine Marginals}
\label{app:combine_marginals}
Till now, we assume two-way marginals are used.  When some marginals contain only a small number of possibilities (\eg, when some attributes are binary), extending to multi-way marginals can help capture more information.
In particular, given $X$, which contains indices of the marginals selected from Algorithm~\ref{alg:marginal_selection}, we first convert each index to its corresponding pair of attributes; we then build a graph $\mathcal{G}$ where each node represents an attribute and each edge corresponds to a pair.  We then find all the cliques of size greater than $2$ in the graph.  If a clique is not very big (smaller than a threshold $\gamma=5000$), and does not overlap much with existing selected attributes (with more than $2$ attributes in common), we merge the 2-way marginals contained in the clique into a multi-way marginal. 

Algorithm \ref{alg:marginal_combine} gives the pseudocode of our proposed marginal combining technique.
We first identify all possible cliques in graph $\mathcal{G}$ and sort them in decending order by their attribute size.
Then, we examine each clique $c$ to determine whether to include as a combined marginal.  
If the clique has a small domain size (smaller than a threshold $\gamma$) and does not contain more than $2$ attributes that is already in the selected attributes set $S$, we include this clique and remove all $2$-way marginals within it.

\begin{algorithm}[!htbp]
    \footnotesize
    \SetCommentSty{small}
    \LinesNumbered
    \caption{Marginal Combine Algorithm}
    \label{alg:marginal_combine}
    
    \KwIn{Selected pairwise marginals $X$,
    threshold $\gamma$}
    \KwOut{Combined marginals $\mathcal{X}$}
    
    Convert $X$ to a set of pairs of attributes;\\
    Construct graph $\mathcal{G}$ from the pairs; \\
    $S\gets\varnothing$; $\mathcal{X} \leftarrow \varnothing$\\
    
    \ForEach{clique size $s$ from $m$ to $3$}
    {
        $C_s \gets$ cliques of size $s$ in $\mathcal{G}$\\
        \ForEach{clique $c \in C_s$}
        {
            \If{$|c\cap S| \le 2$ and domain size of $c$ $\le \gamma$}
            {
                Append $c$ to $\mathcal{X}$ \\
                Append the attributes of $c$ to $S$
            }
        }
    }
    
\end{algorithm}

\section{Comparison of Synthesis Methods}
\label{app:comparison_synthesis}

To better understand the performance of different synthesis methods, we select marginals in a non-private setting and purely compare the performance of different synthesis methods.
This is different from the end-to-end evaluation in Section~\ref{subsec:comparison_e2e} that makes all steps private.
Other settings are the same as Section~\ref{subsec:comparison_e2e}.
We do not compare with \dq in this experiment since Section~\ref{subsec:comparison_e2e} has illustrated that its performance is much worse than other methods.

\mypara{Results}
Figure~\ref{fig:comparison_synthesis} shows the performance of different data synthesis methods.
Both \mcf and \gum exploit dense marginals selected by \marginal, while the performance of \mcf is even worse than the \pgm method and the \privbayes method that using spare marginals.
The reason is that, in each iteration, \mcf enforces the synthetic dataset \ds to fully match the marginal.  
This would severely destroy the correlation established by other marginals.
While \gum preserves the correlation of other marginals by gradually updating marginals in each iteration and using duplication technique.

Comparing Figures~\ref{fig:comparison_e2e} and \ref{fig:comparison_synthesis}, we observe that the experimental results in the private and non-private settings are similar, showing the robustness of \method.
This is consistent with the result in Section \ref{subsec:comparison_marginal}.

\begin{figure*}[!h]
    \centering

    \subfloat[$\epsilon=0.2$]{\includegraphics[width=0.3\textwidth]{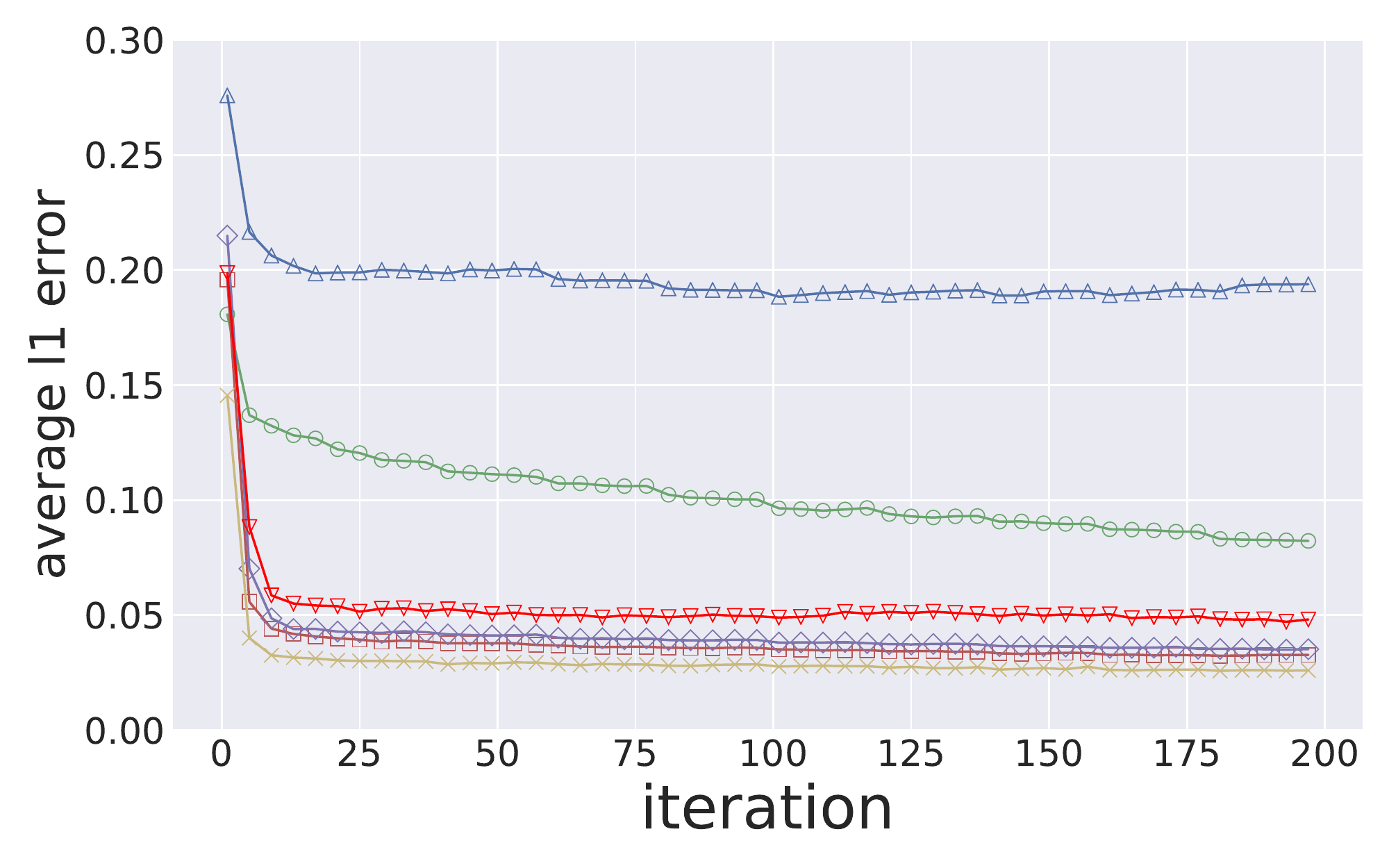}}
    \subfloat[$\epsilon=1.0$]{\includegraphics[width=0.3\textwidth]{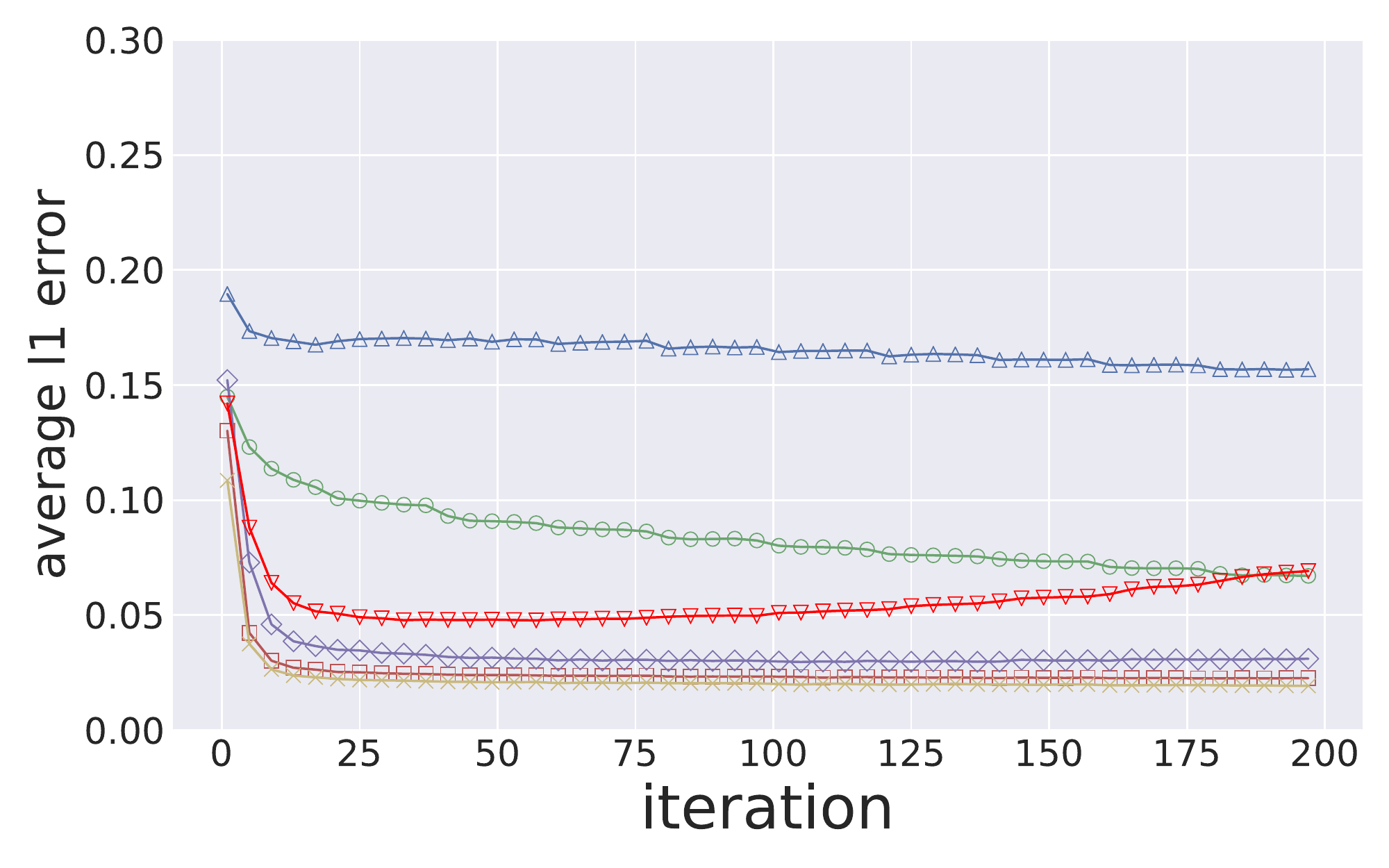}}
    \subfloat[$\epsilon=2.0$]{\includegraphics[width=0.3\textwidth]{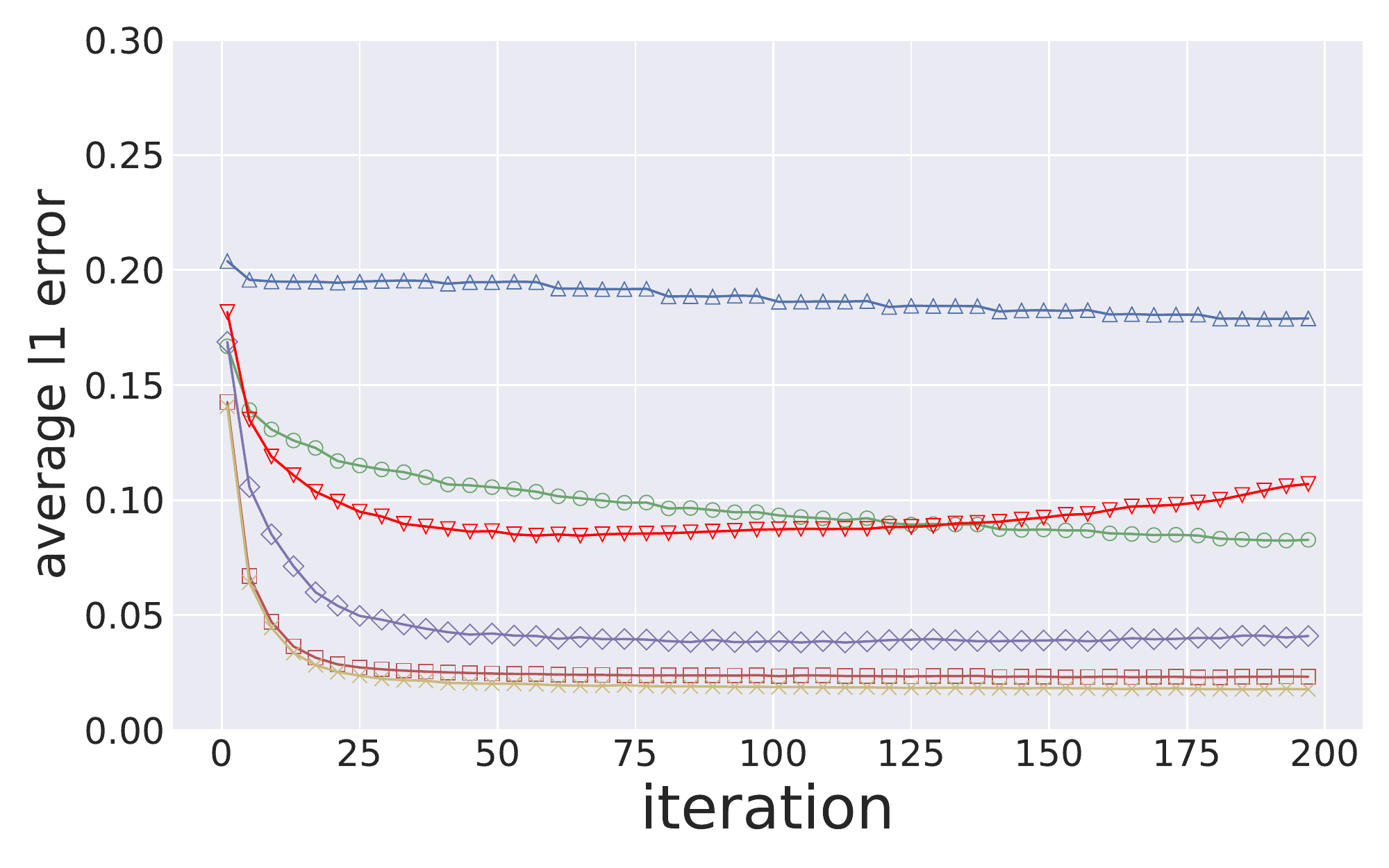}} \\ [1ex]
    
    \centering{US Accident} \\ [-2ex]
    
    \subfloat[$\epsilon=0.2$]{\includegraphics[width=0.3\textwidth]{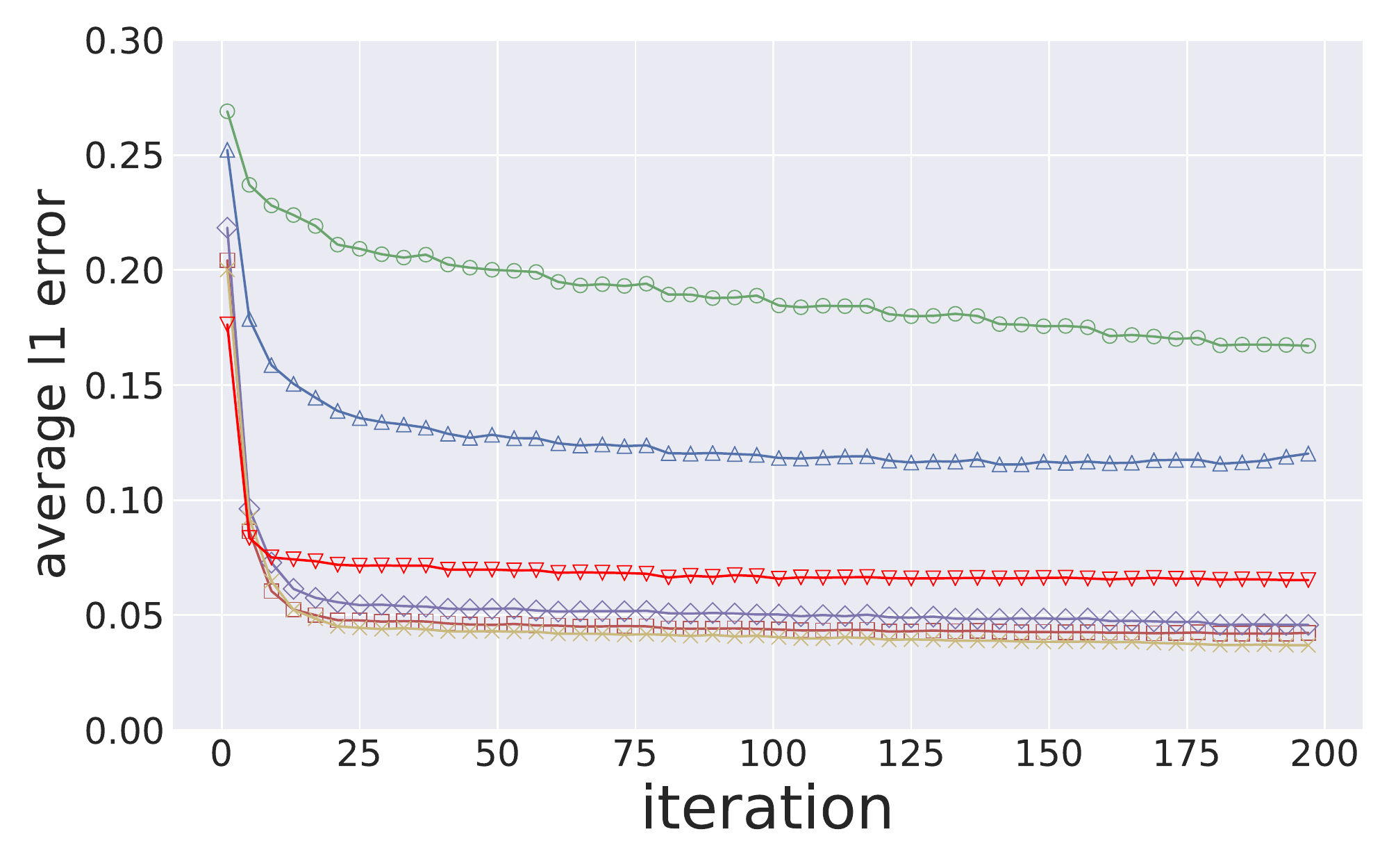}}
    \subfloat[$\epsilon=1.0$]{\includegraphics[width=0.3\textwidth]{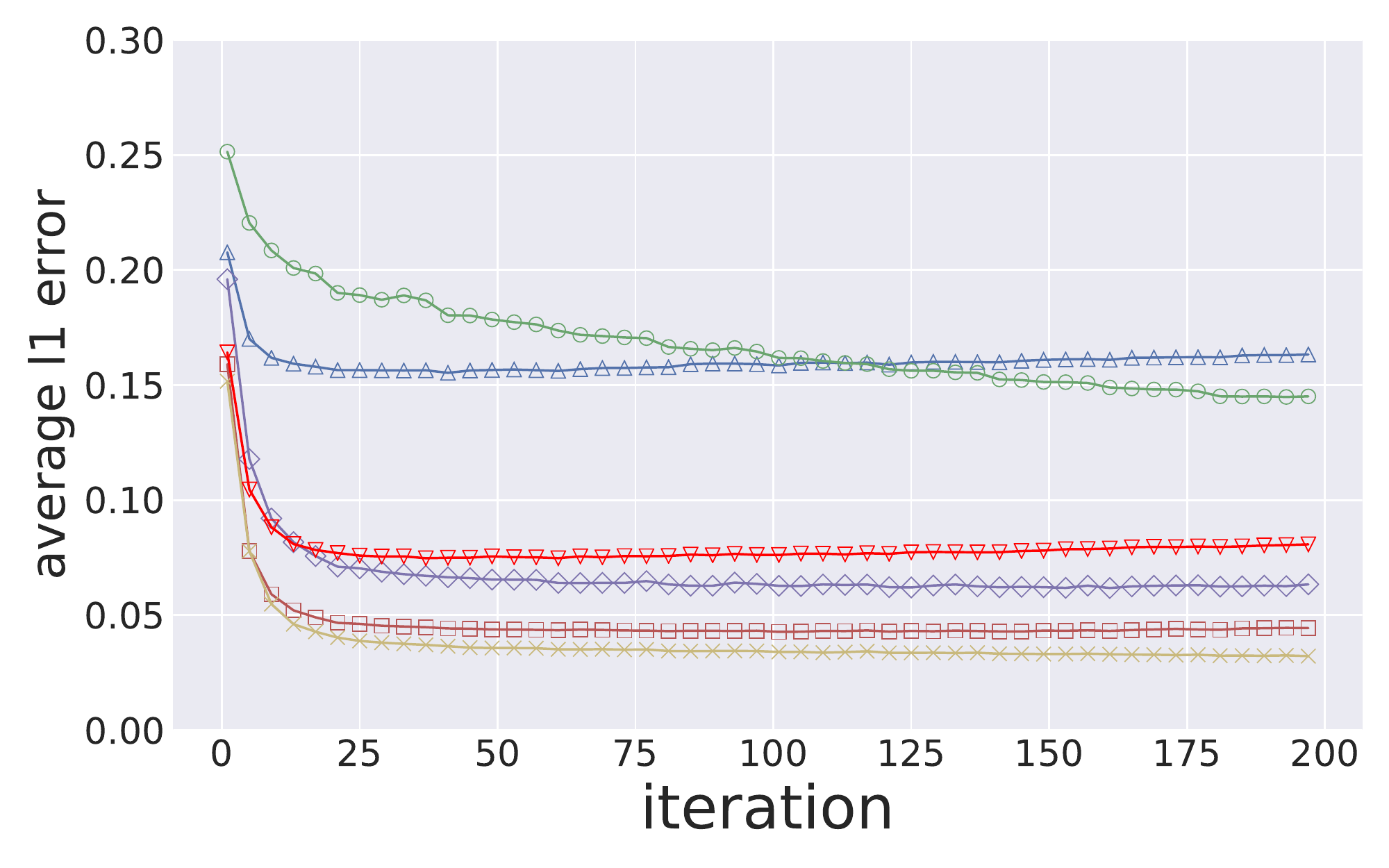}}
    \subfloat[$\epsilon=2.0$]{\includegraphics[width=0.3\textwidth]{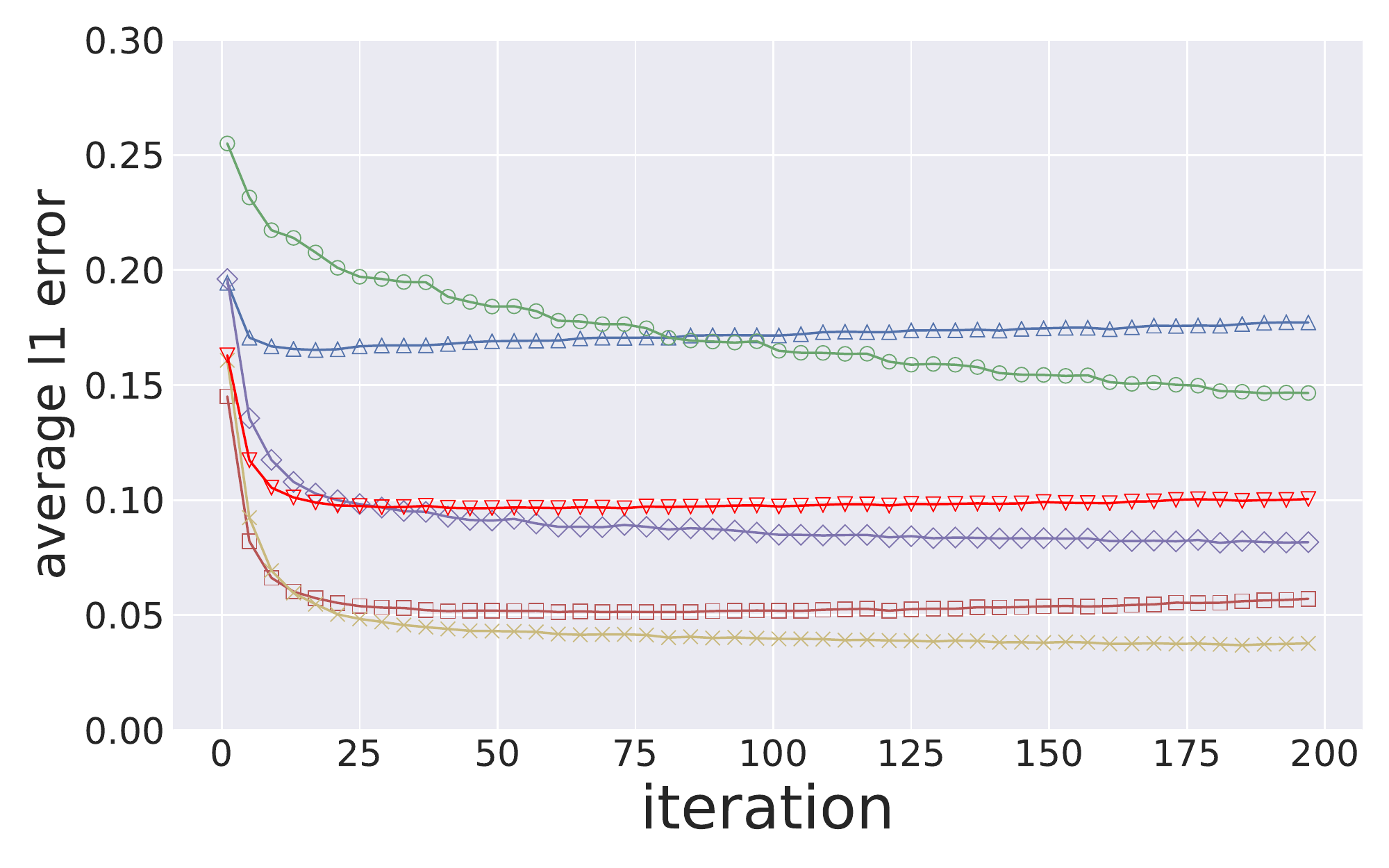}} \\ [1ex]
    
    \centering{Loan} \\ [-2ex]
    
    \subfloat[$\epsilon=0.2$]{\includegraphics[width=0.3\textwidth]{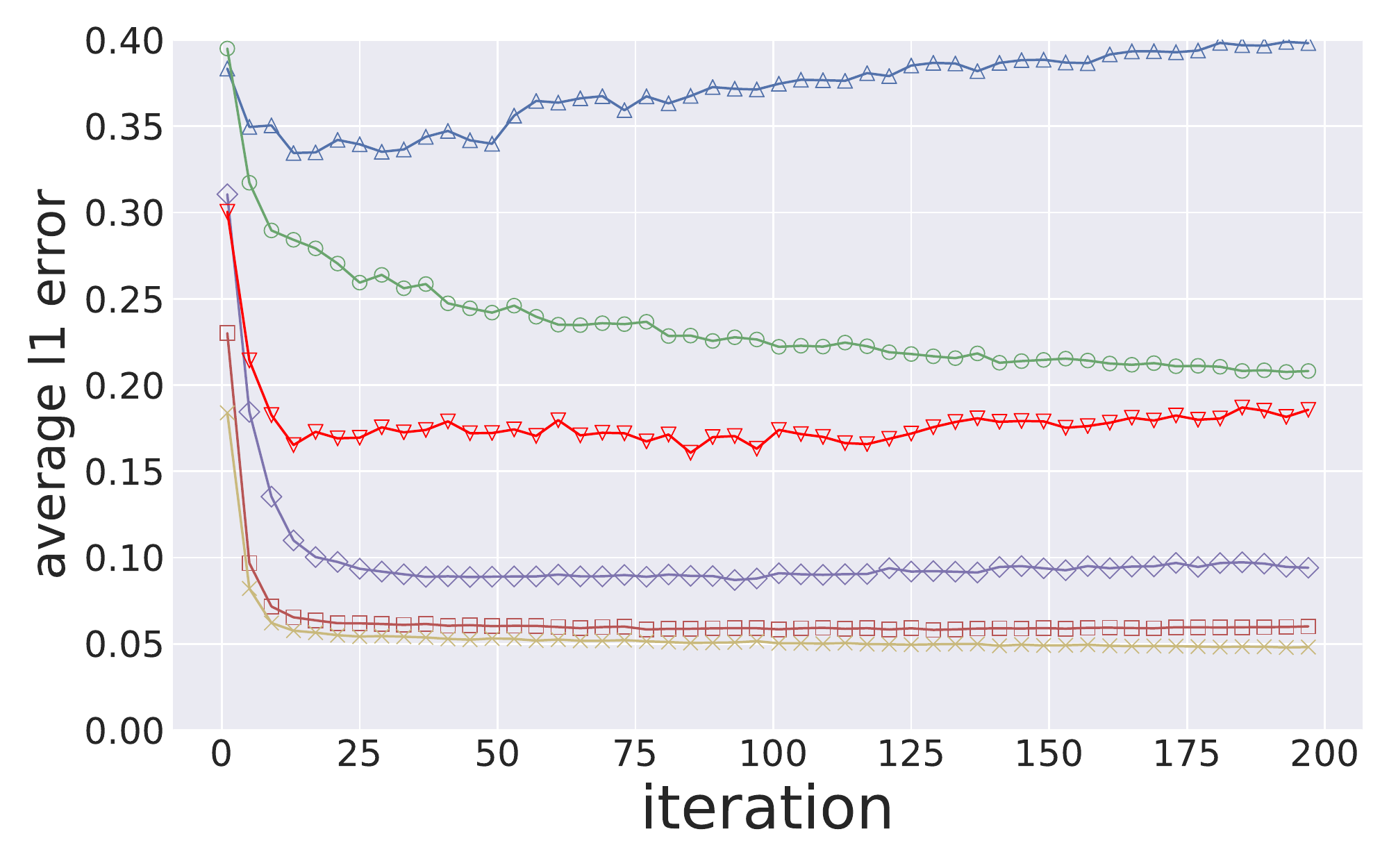}}
    \subfloat[$\epsilon=1.0$]{\includegraphics[width=0.3\textwidth]{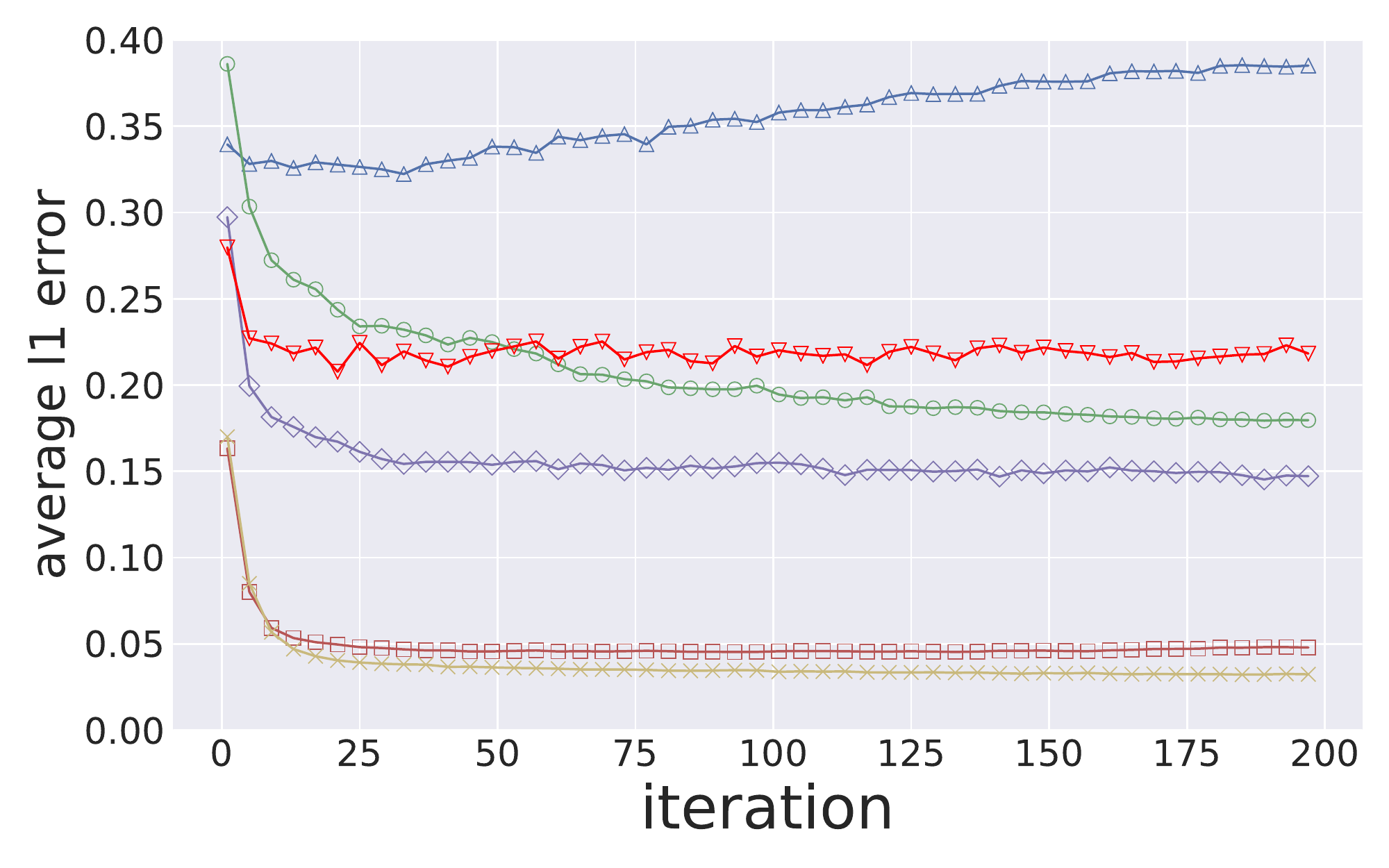}}
    \subfloat[$\epsilon=2.0$]{\includegraphics[width=0.3\textwidth]{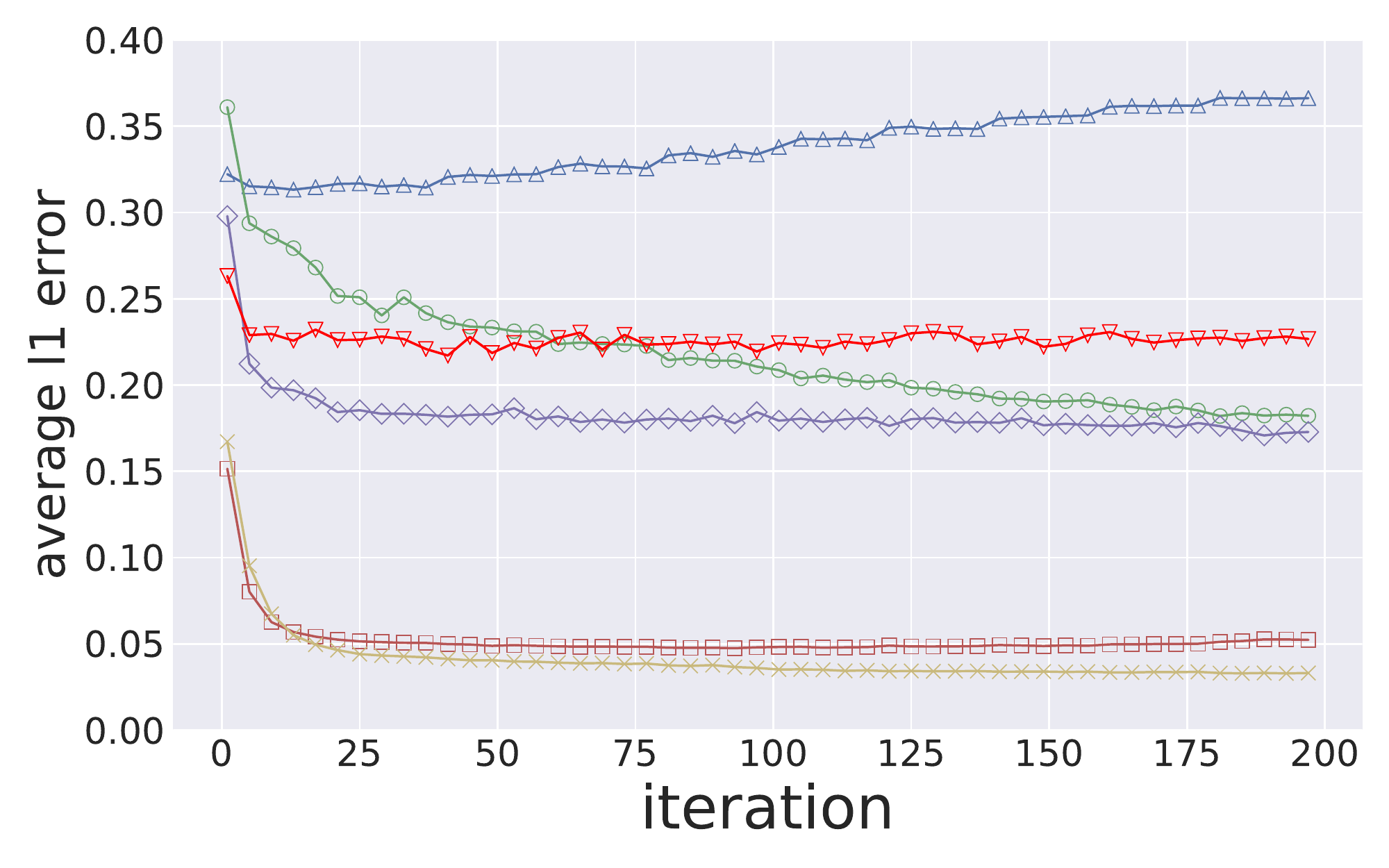}} \\ [1ex]
    
    \centering{Colorado} \\

    \subfloat{\includegraphics[width=0.7\textwidth]{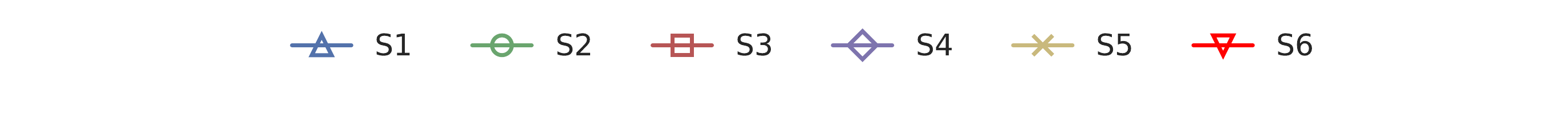}}  \\ [-3ex]

    \caption{Comparison of different records updating algorithms.
    The y-axis stands for the average $\ell_1$ error for all marginals in each iteration.
    S1 stands for the All Replace strategy, S2 stands for the All Duplicate strategy, S3 stands for the All Half-half strategy, S4 stands for the Replace Plus Duplicate strategy, S5 stands for the Half-half Plus Duplicate strategy, S6 stands for the Half-half Plus Replace strategy.
    }
    \label{fig:converge_update}
\end{figure*}

\begin{figure*}[!h]
    \centering

    \subfloat[$\epsilon=0.2$]{\includegraphics[width=0.3\textwidth]{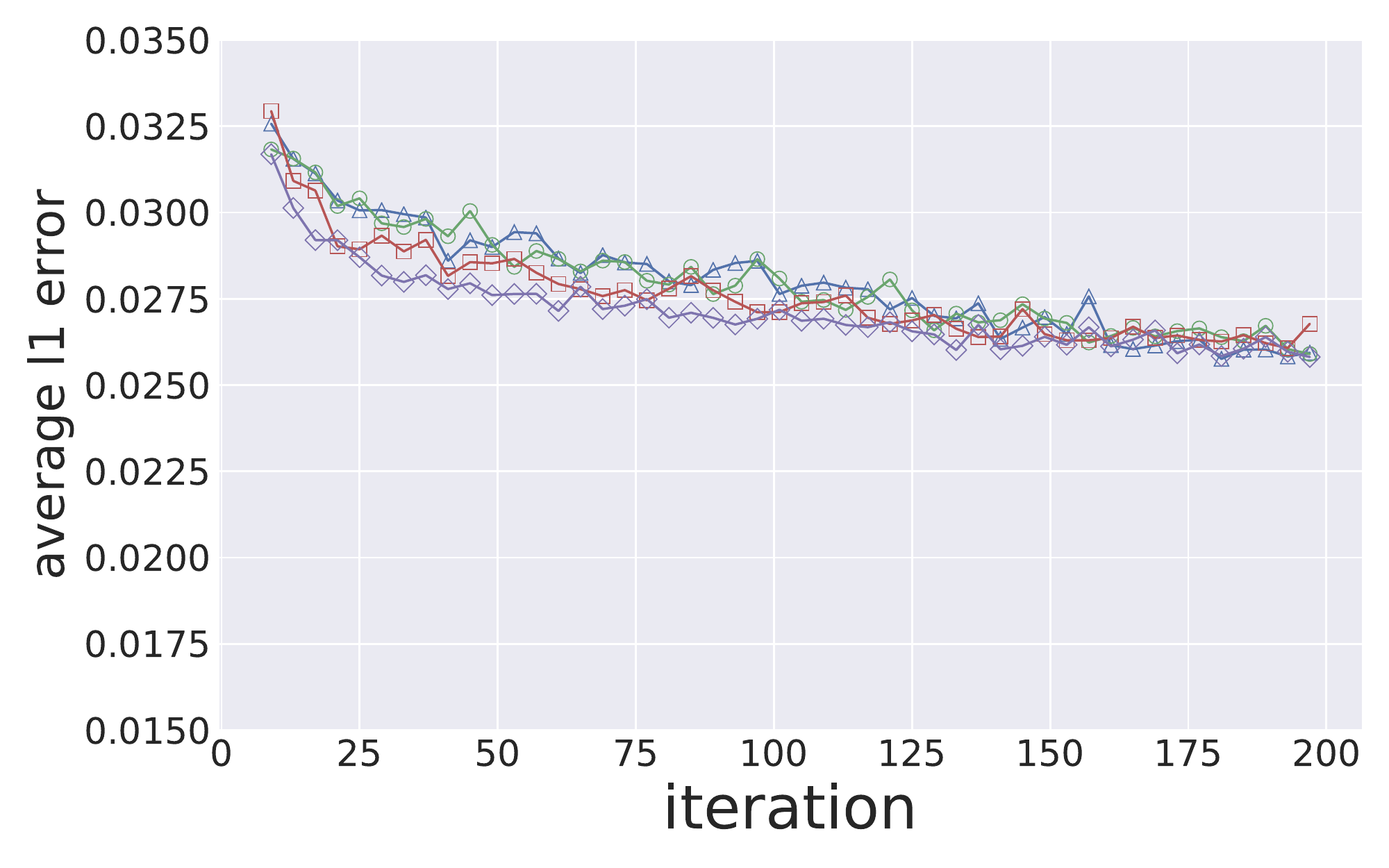}}
    \subfloat[$\epsilon=1.0$]{\includegraphics[width=0.3\textwidth]{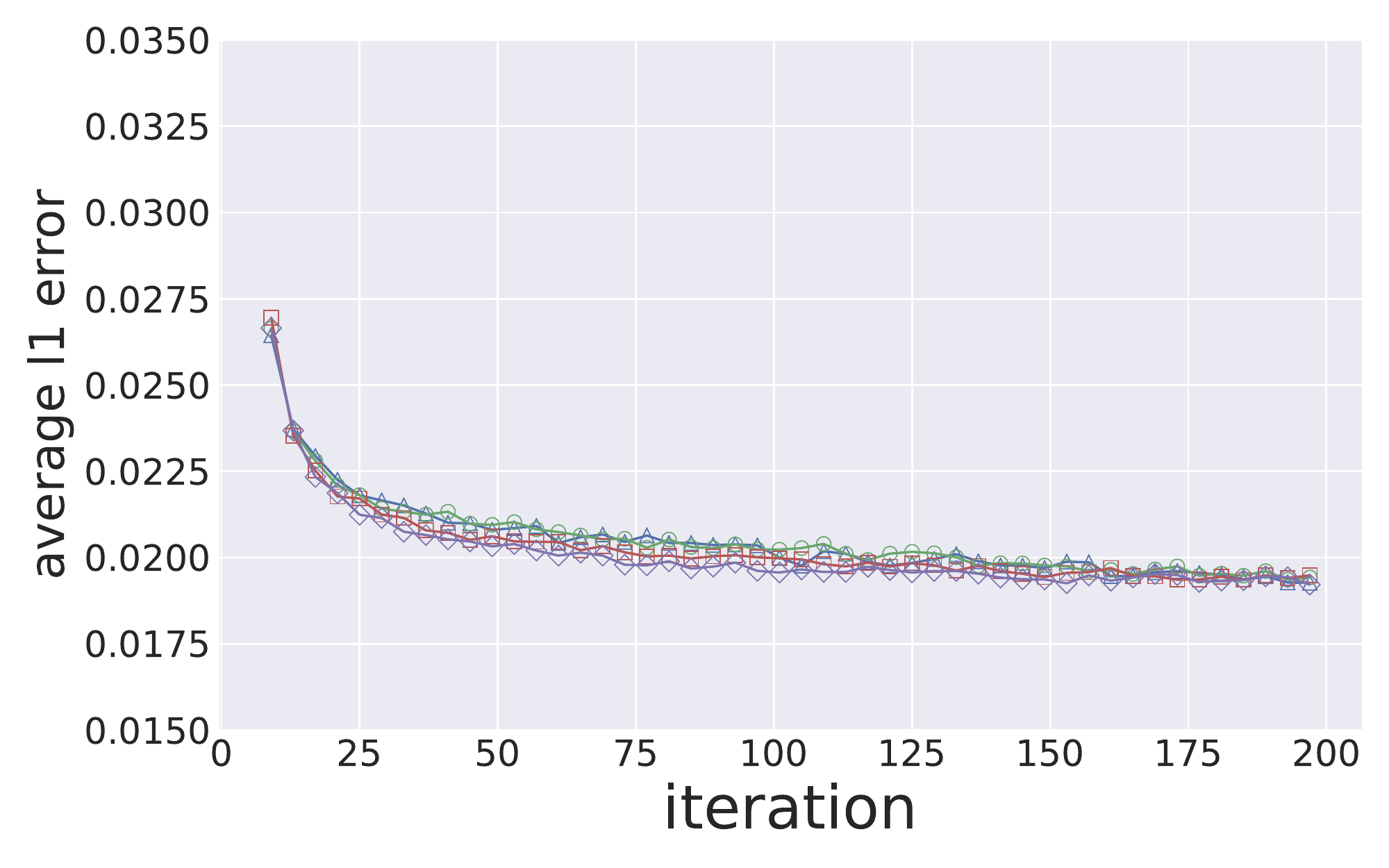}}
    \subfloat[$\epsilon=2.0$]{\includegraphics[width=0.3\textwidth]{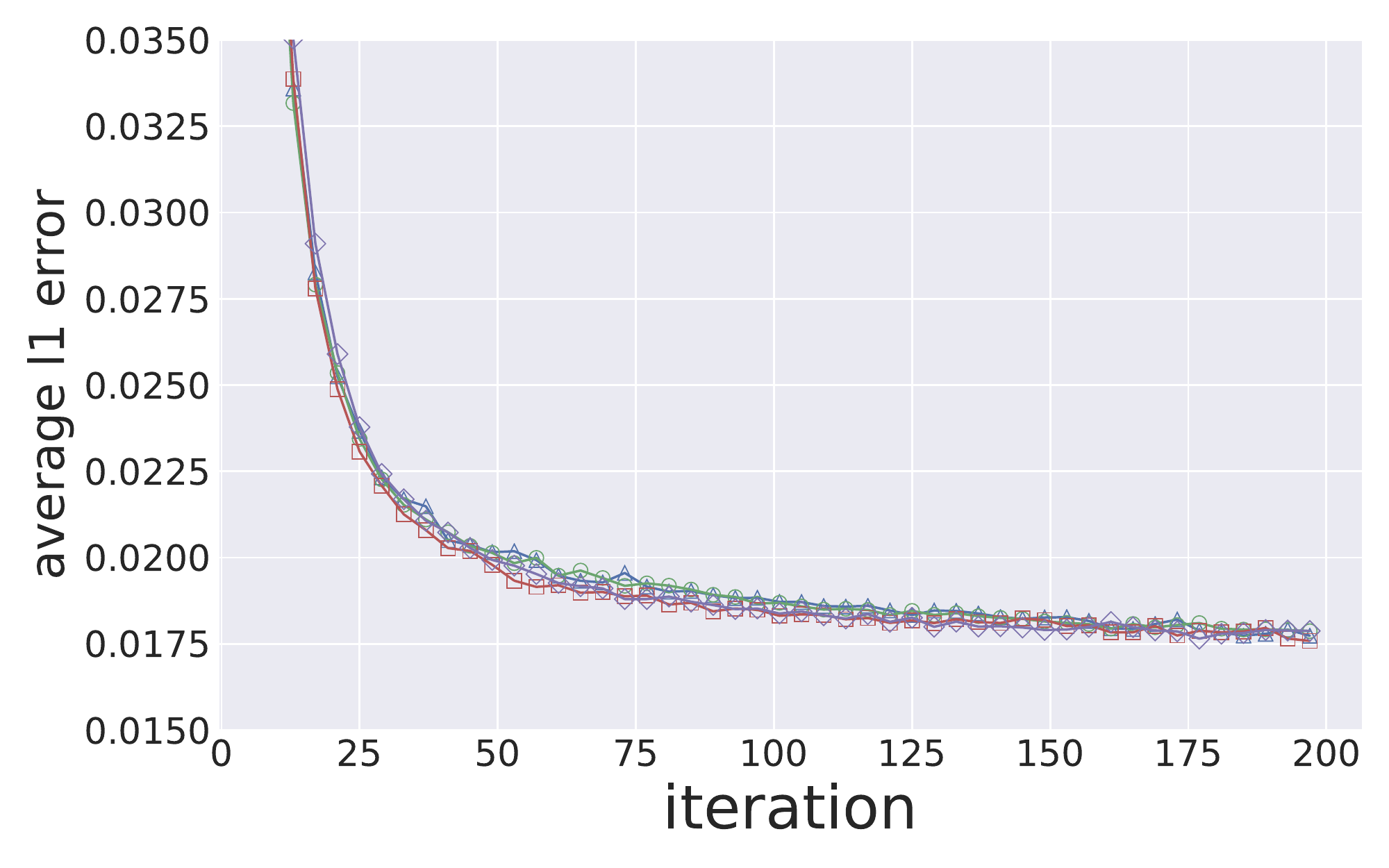}} \\ [1ex]
    
    \centering{US Accident} \\ [-2ex]
    
    \subfloat[$\epsilon=0.2$]{\includegraphics[width=0.3\textwidth]{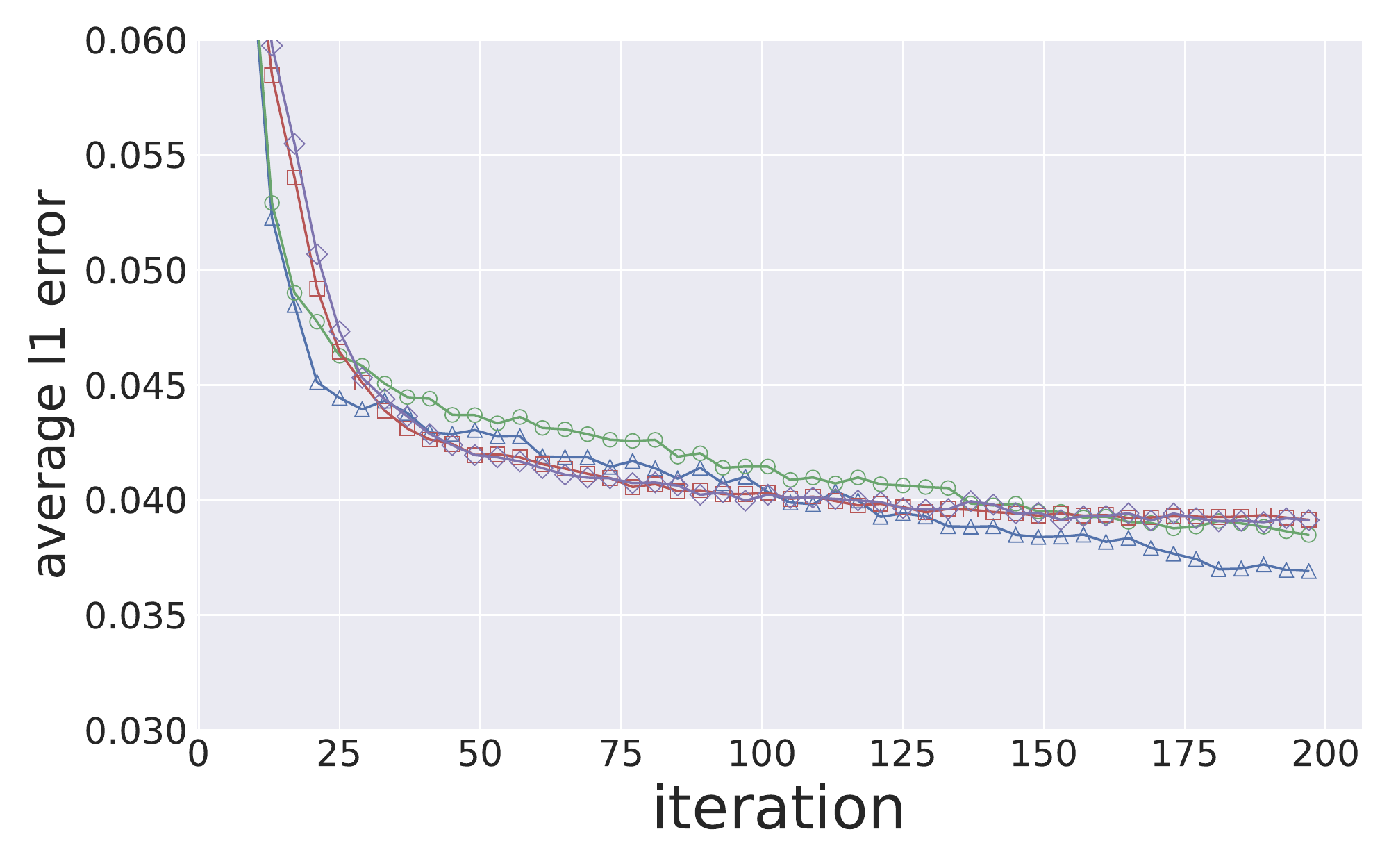}}
    \subfloat[$\epsilon=1.0$]{\includegraphics[width=0.3\textwidth]{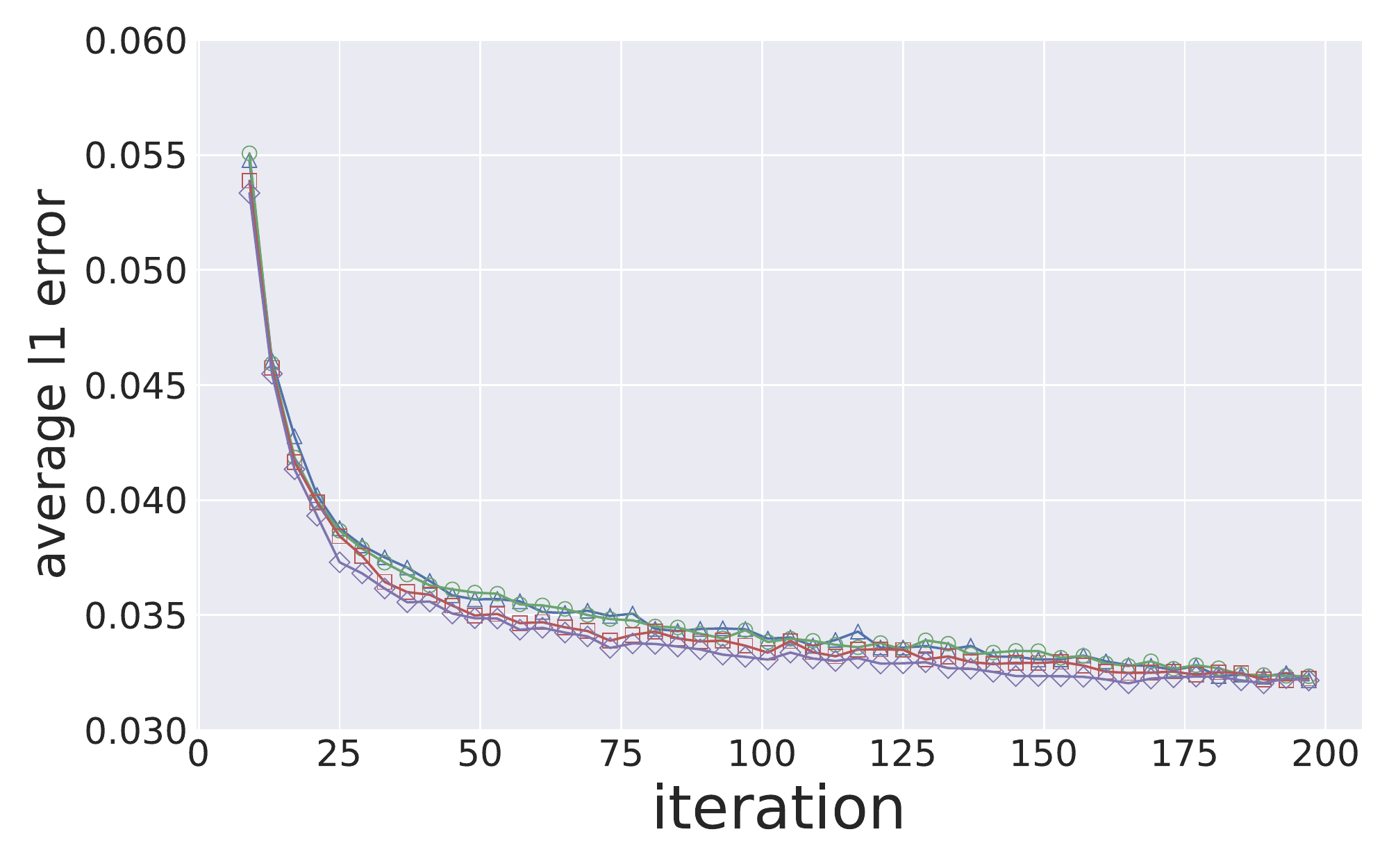}}
    \subfloat[$\epsilon=2.0$]{\includegraphics[width=0.3\textwidth]{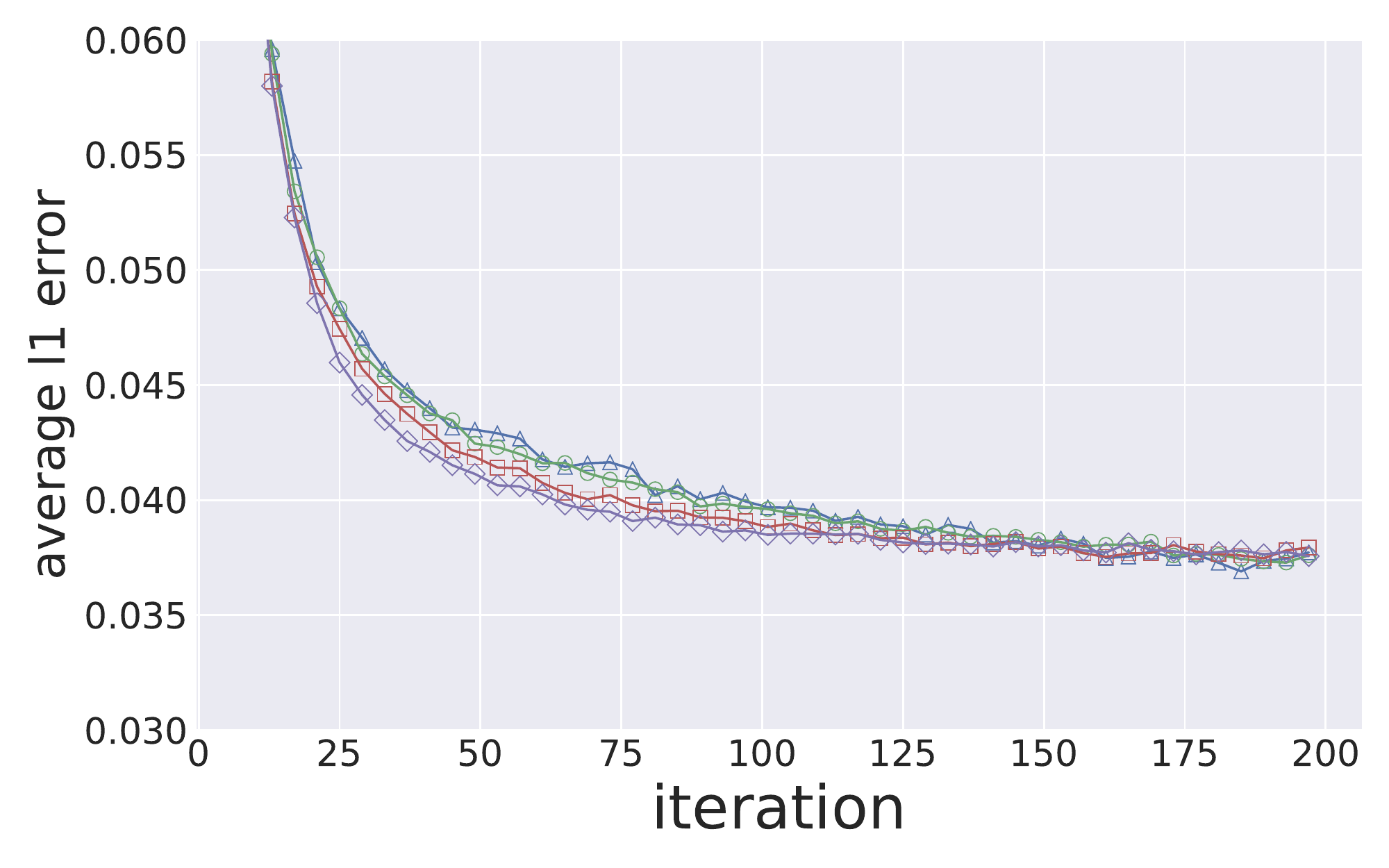}} \\ [1ex]
    
    \centering{Loan} \\ [-2ex]
    
    \subfloat[$\epsilon=0.2$]{\includegraphics[width=0.3\textwidth]{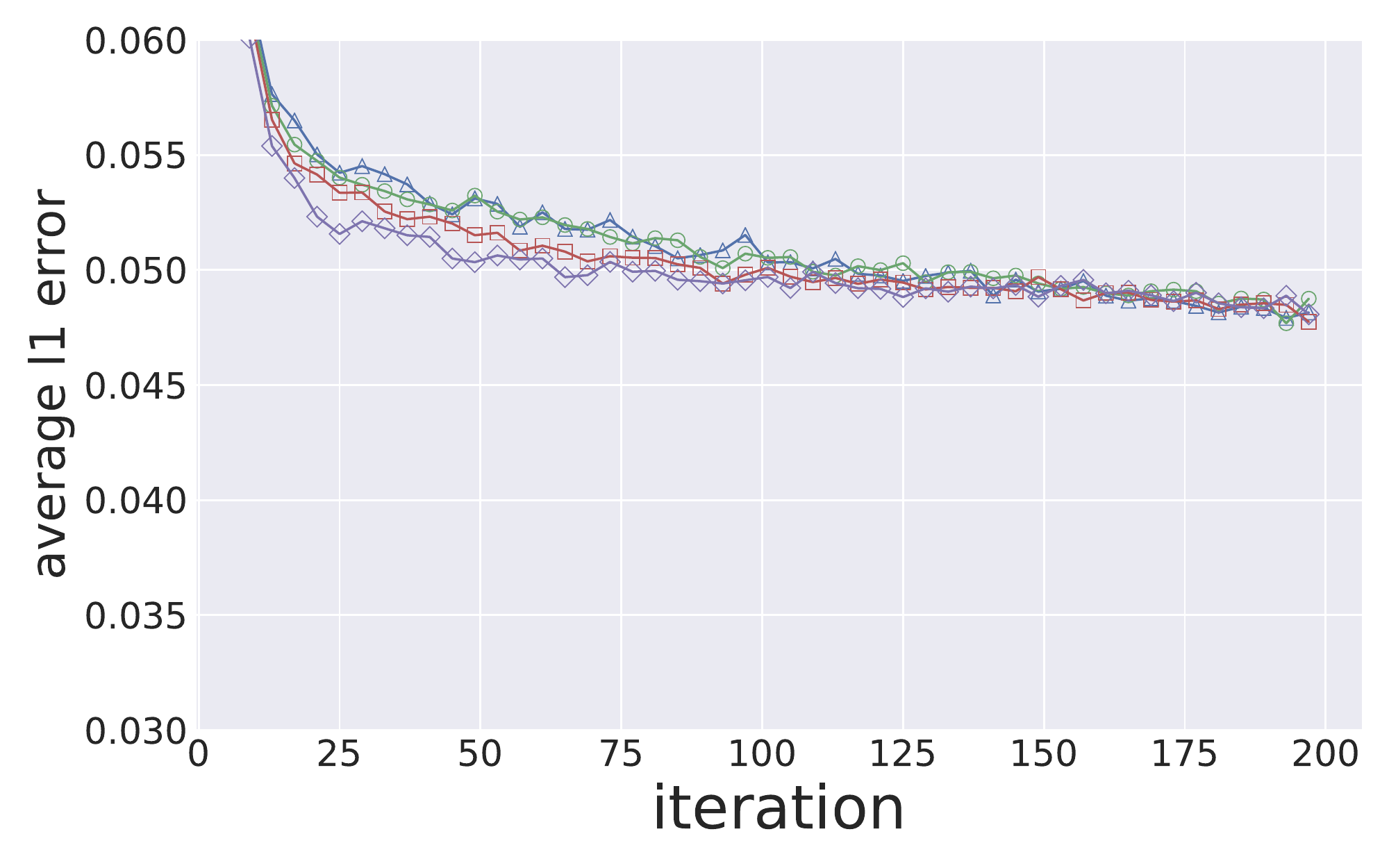}}
    \subfloat[$\epsilon=1.0$]{\includegraphics[width=0.3\textwidth]{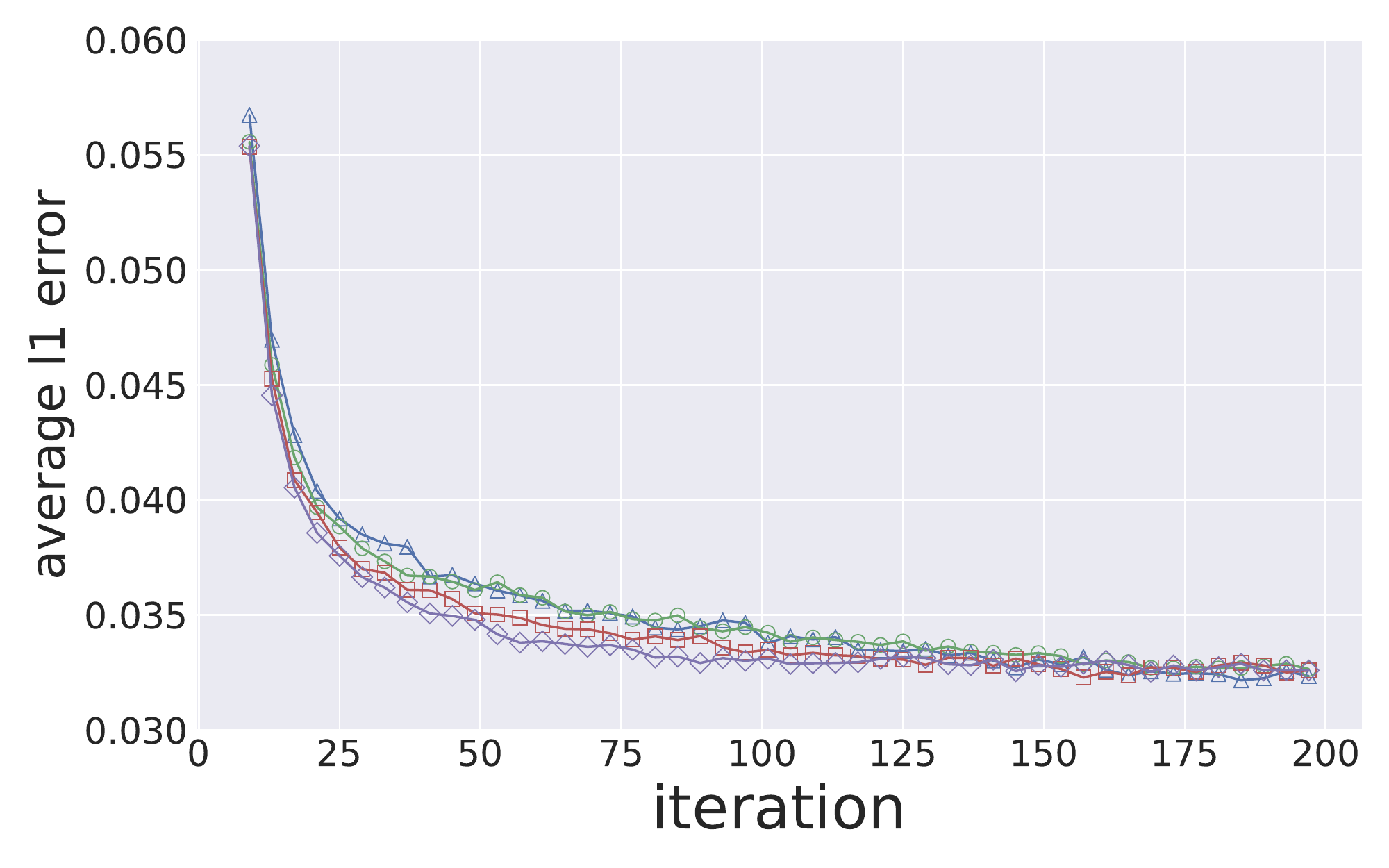}}
    \subfloat[$\epsilon=2.0$]{\includegraphics[width=0.3\textwidth]{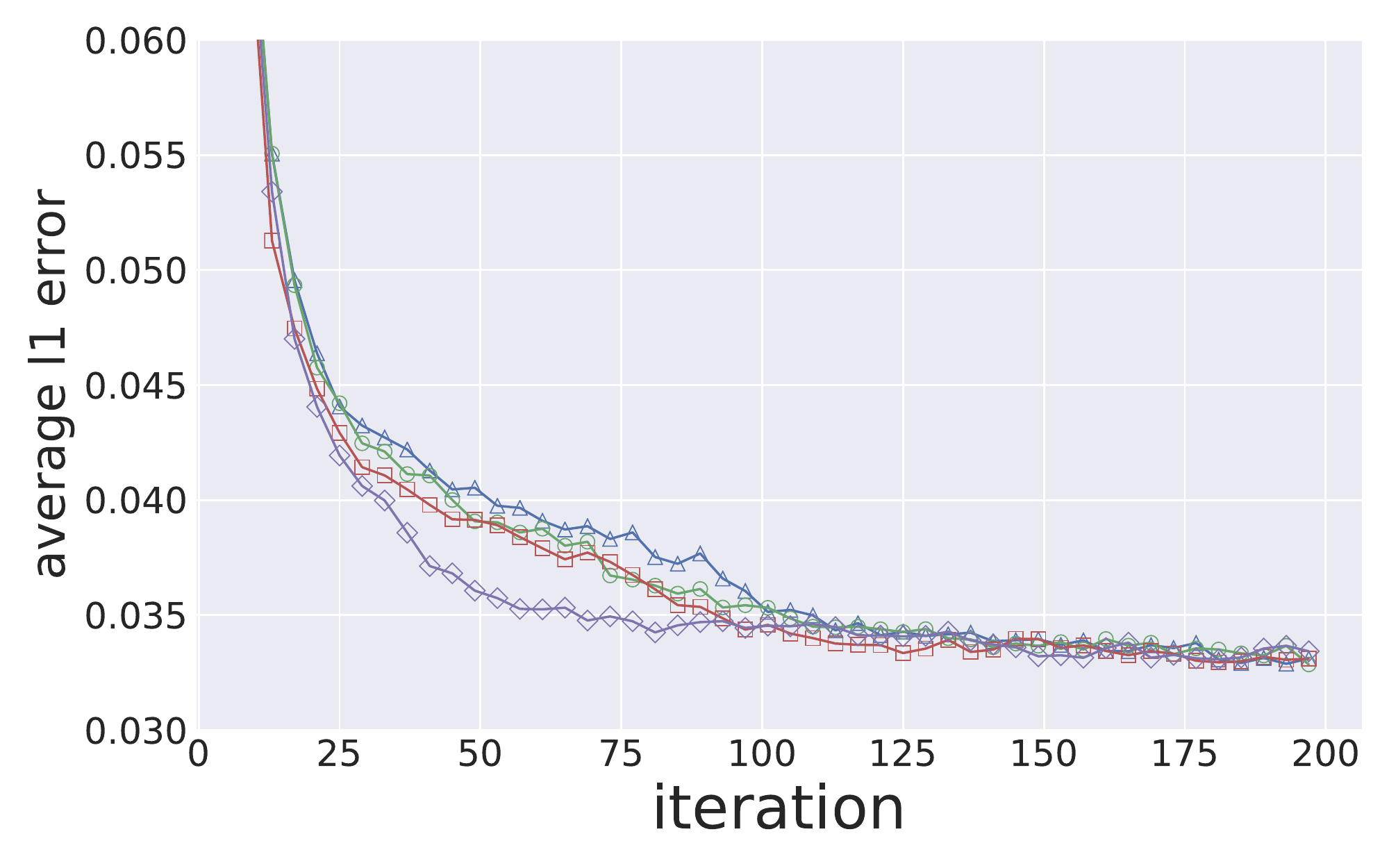}} \\ [1ex]
    
    \centering{Colorado} \\ 

    \subfloat{\includegraphics[width=0.7\textwidth]{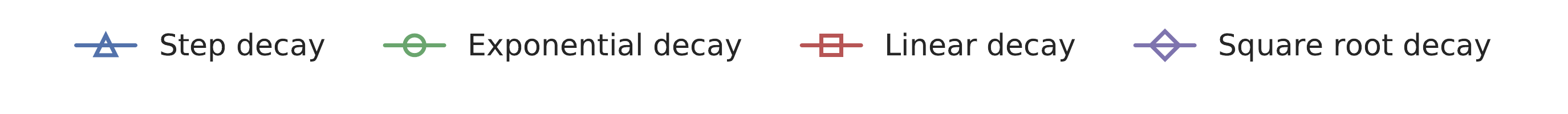}}  \\

    \caption{Comparison of different decay algorithms for $\alpha$.
    The y-axis stands for the average $\ell_1$ error for all marginals in each iteration.
    }
    \label{fig:converge_decay}
\end{figure*}

\section{Comparison of Records Updating \\ Strategies}
\label{app:comparison_update}

In this section, we compare different records updating strategies for \gum.

\mypara{Competitors}
We have three basic strategies: 
(1) only replace the attributes in the marginals;
(2) duplicate the whole records;
(3) half replace and half duplicate.
The example in Figure~\ref{fig:update_example} corresponds to half replace and half duplicate strategy, where $v_4$ is replaced and $v_3$ is duplicated by $v_5$, respectively.
In the empirical study, we observe that applying different basic strategies in different iterations could benefit the convergence performance.
Thus, we also consider three hybrid strategies that combine different basic strategies.
In what follows, we summarize several strategies in our experiments.

\begin{itemize}
    \item \mypara{S1: All Replace}
    Using the replace strategy in all iterations.
    
    \item \mypara{S2: All Duplicate}
    Using the duplicate strategy in all iterations.
    
    \item \mypara{S3: All Half-half}
    Using the half-half strategy in all iterations.
    
    \item \mypara{S4: Replace Plus Duplicate}
    Using the replace strategy and the duplicate strategy in different iterations.
    
    \item \mypara{S5: Half-half Plus Duplicate}
    Using the half-half strategy and the duplicate strategy in different iterations.
    
    \item \mypara{S6: Half-half Plus Replace}
    Using the half-half strategy and the replace strategy in different iterations.
\end{itemize}

\mypara{Results}
Figure~\ref{fig:converge_update} illustrates the convergence performance of different records updating strategies.
In each iteration, we record the average $\ell_1$ error of all marginals on the current synthetic dataset.

The experimental results show that both using replace and duplicate strategy alone cannot achieve satisfactory performance.
In one hand, only using the replace strategy would significantly destroy the correlation information established by other marginals.
In another hand, only using the duplicate strategy will not introduce new records that can better reflect the overall joint distribution.
The half-half strategy balance the drawbacks of replace and duplicate strategy, and can achieve pretty good convergence performance.
Furthermore, if we combine the duplicate strategy and the half-half strategy in different iterations, the convergence performance can be further improved.
Thus, we use Half-half Plus Duplicate strategy in all of our experiments.

\section{Comparison of Decay Algorithms}
\label{app:comparison_decay}

Figure~\ref{fig:converge_decay} illustrates the convergence performance of different decay algorithms for $\alpha$.
In each iteration, we record the average $\ell_1$ error of all marginals on the current synthetic dataset.
We set the initial $\alpha$ as $1.0$.

The experimental result shows that the average $\ell_1$ error drops significantly in the first $100$ iterations.
After $100$ iteration, the improvement of increasing the number of iterations is negligible.
To achieve a robust convergence performance, we set the number of iterations as $100$ in all of our experiments.

When we compare different decay algorithms, we find that their convergence performance are quite similar in most cases.
We further find that the step decay algorithm consistently performs good in almost all settings while others perform relatively bad in some settings.
Thus, we use the step decay algorithm in all of our experiments.
The step decay algorithm is also widely used to update the step size in the training of deep neural networks~\cite{krizhevsky2012imagenet}.

Another observation is that, when the privacy budget increases from $0.2$ to $1.0$, the average $\ell_1$ error increases for all datasets.
The reason is that when the privacy budget is small, the marginal distribution may deviate from the original dataset.
This may lead to no dataset that can perfectly match the distribution of all marginals.
However, when the privacy budget increases from $1.0$ to $2.0$, the average $\ell_1$ error of some datasets drops.
Notice that this phenomenon do not mean that the overall performance of the synthetic dataset drops.
The reason is that we only measure the average $\ell_1$ error of the marginals selected by \marginal.
According to Algorithm~\ref{alg:marginal_selection}, we could select more marginals when the privacy budget is larger which captures more correlation information.

\begin{figure*}[!ht]
    \centering

    \subfloat[US Accident]{\includegraphics[width=0.3\textwidth]{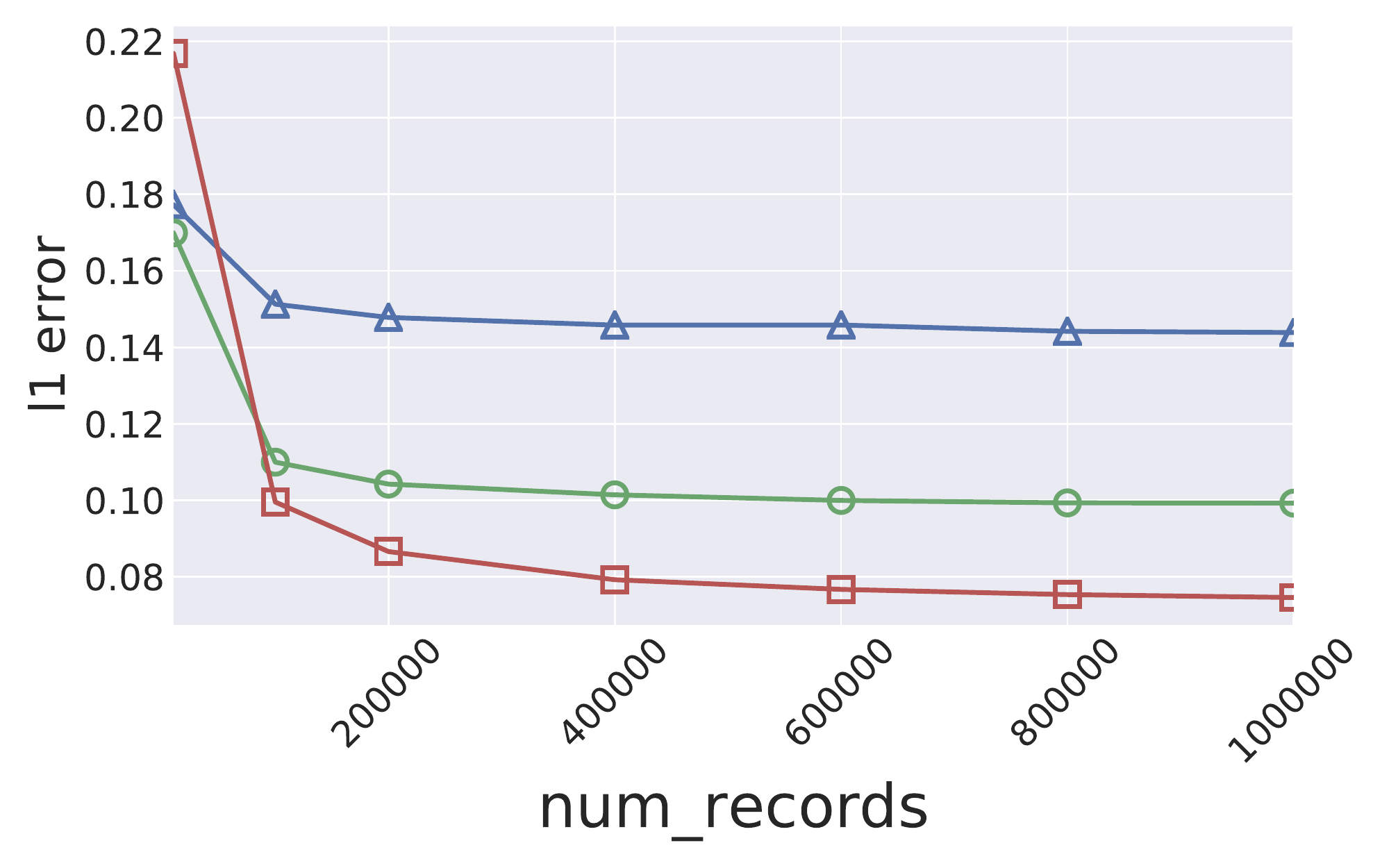}}
    \subfloat[Loan]{\includegraphics[width=0.3\textwidth]{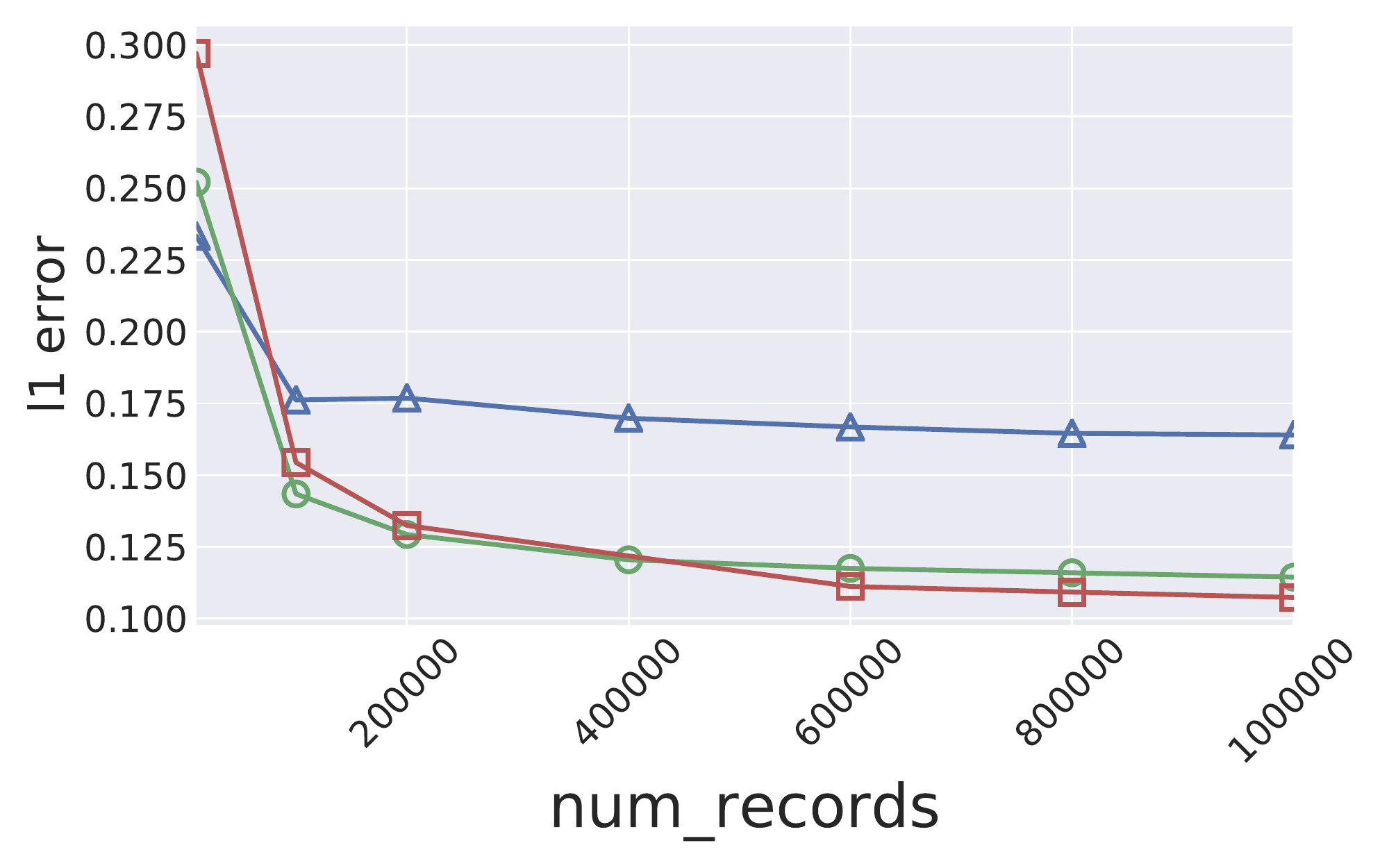}}
    \subfloat[Colorado]{\includegraphics[width=0.3\textwidth]{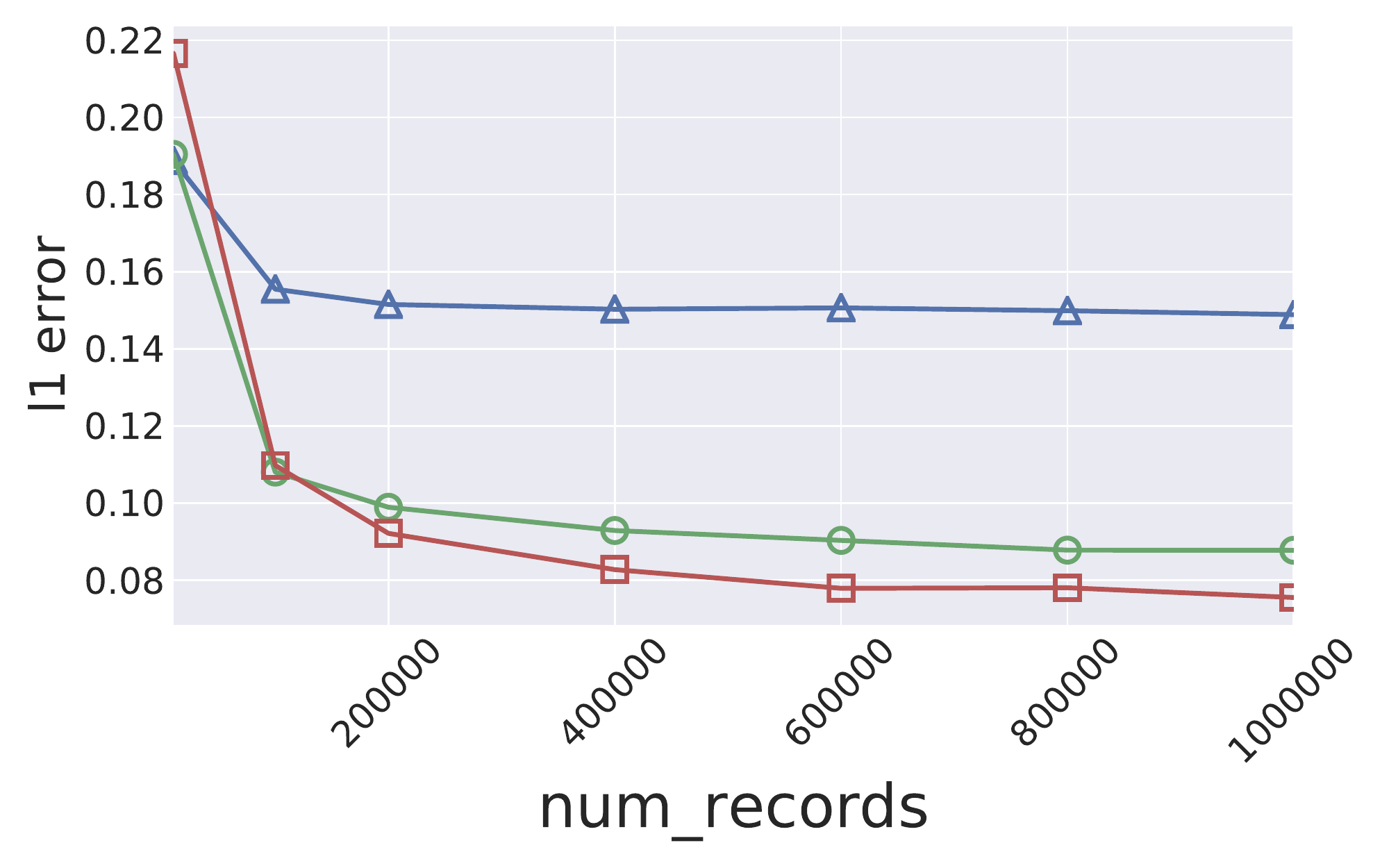}} \\ [1ex]

    \subfloat{\includegraphics[width=0.7\textwidth]{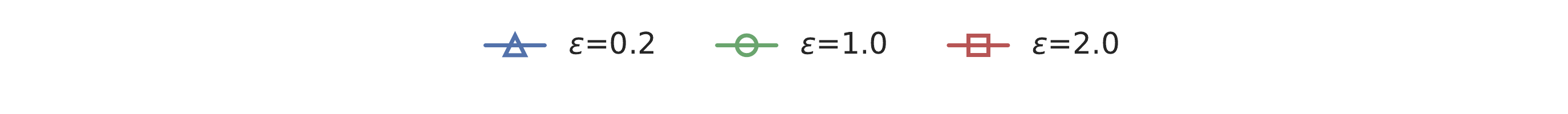}}  \\

    \caption{Impact of the number of synthetic records.
    }
    \label{fig:config_num_records}
\end{figure*}

\section{Impact of the Number of Synthetic Records}
\label{app:impact_number_records}

Figure~\ref{fig:config_num_records} shows the impact of the number of synthetic records $|\ds|$.
We only report the results of the average $\ell_1$ error of all $2$-way marginals, since
it is a good indicator for the overall performance.
In the experiment, we vary $|\ds|$ from $10000$ to $1000000$.

The experimental result shows that the performance improves when $|\ds|$ increases from $10000$ to $600000$.
When $|\ds|$ is larger than $600000$, increasing the number synthetic records has negligible impact on the performance.
We use the number of records in the original datasets, which is approximate to $600000$.

\section{Additional Results on Other Datasets}
\label{app:other_datasets}

Figure~\ref{fig:adult_results} and Figure~\ref{fig:loan_results} show the additional experimental results for Adult and Loan datasets.
The conclusion is consistent with that of the Accident and Colorado datasets.

\begin{figure*}[!htpb]
    \centering
    
    \subfloat[pair-wise marginal]{\includegraphics[width=0.3\textwidth]{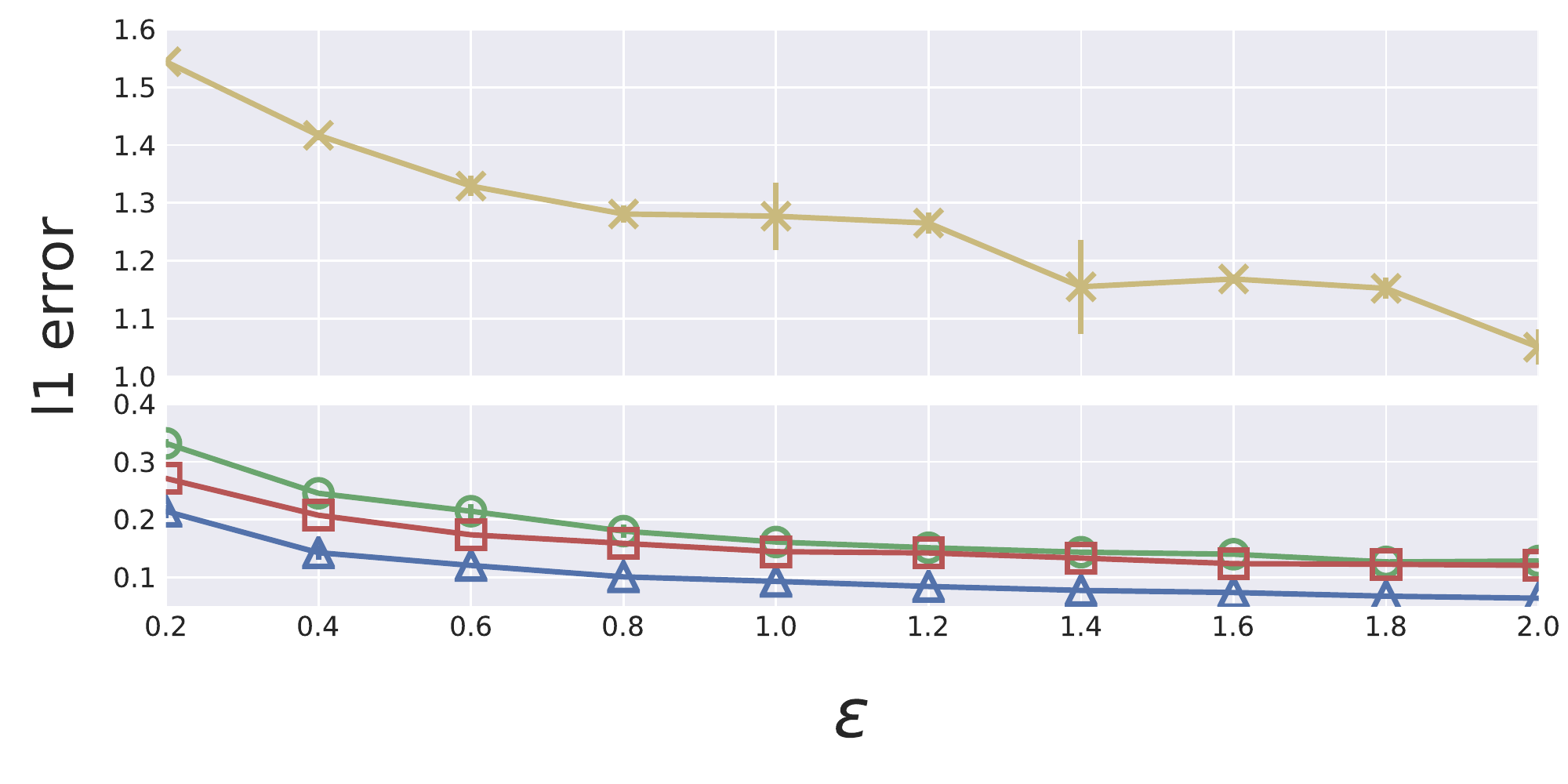}}
    \subfloat[range query]{\includegraphics[width=0.3\textwidth]{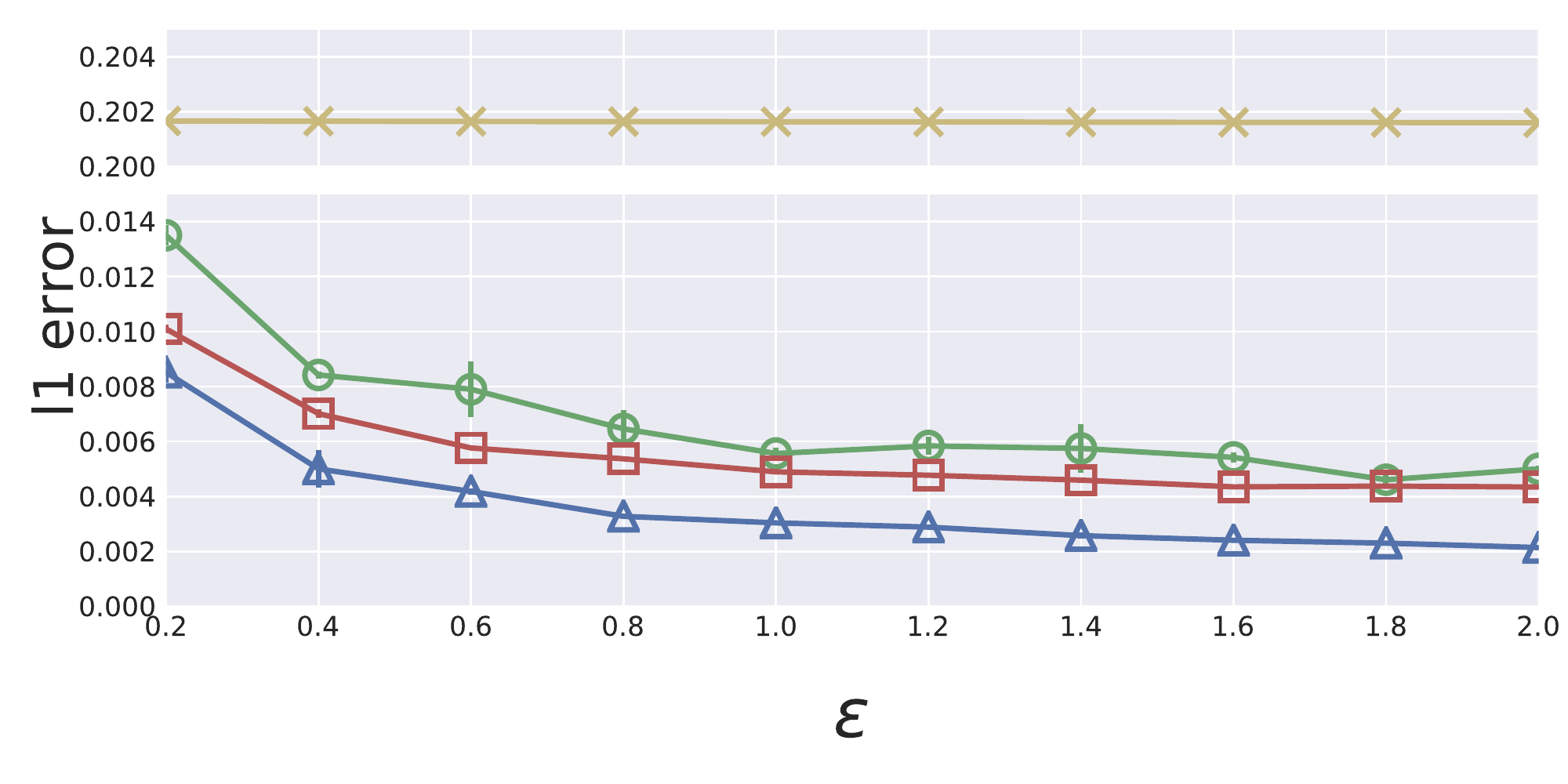}}
    \subfloat[classification]{\includegraphics[width=0.3\textwidth]{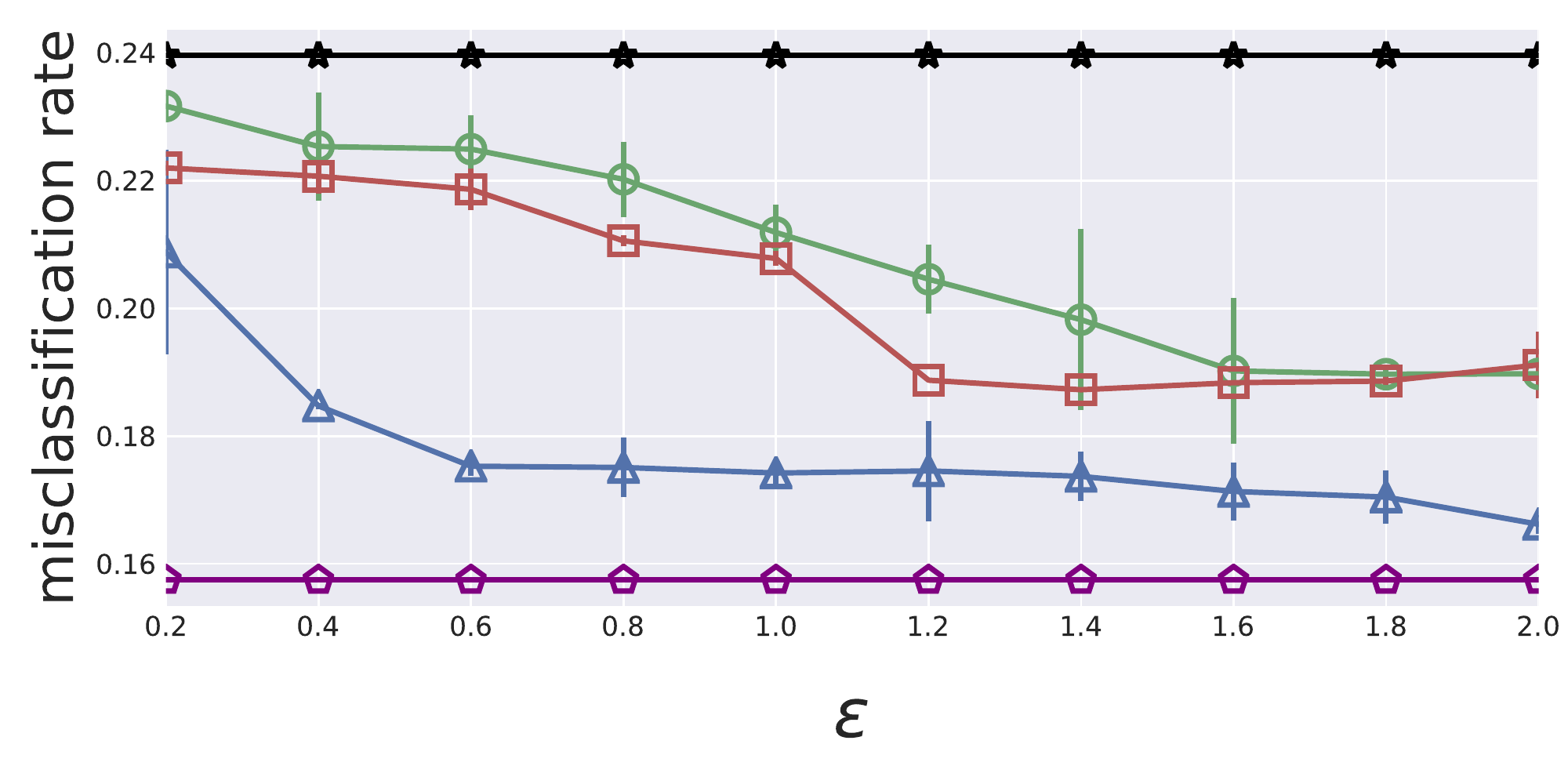}} \\ [1ex]
    
    \subfloat{\includegraphics[width=0.8\textwidth]{Figures/comparison_e2e_legend.pdf}}  \\ [-2ex]

    \subfloat[pair-wise marginal]{\includegraphics[width=0.3\textwidth]{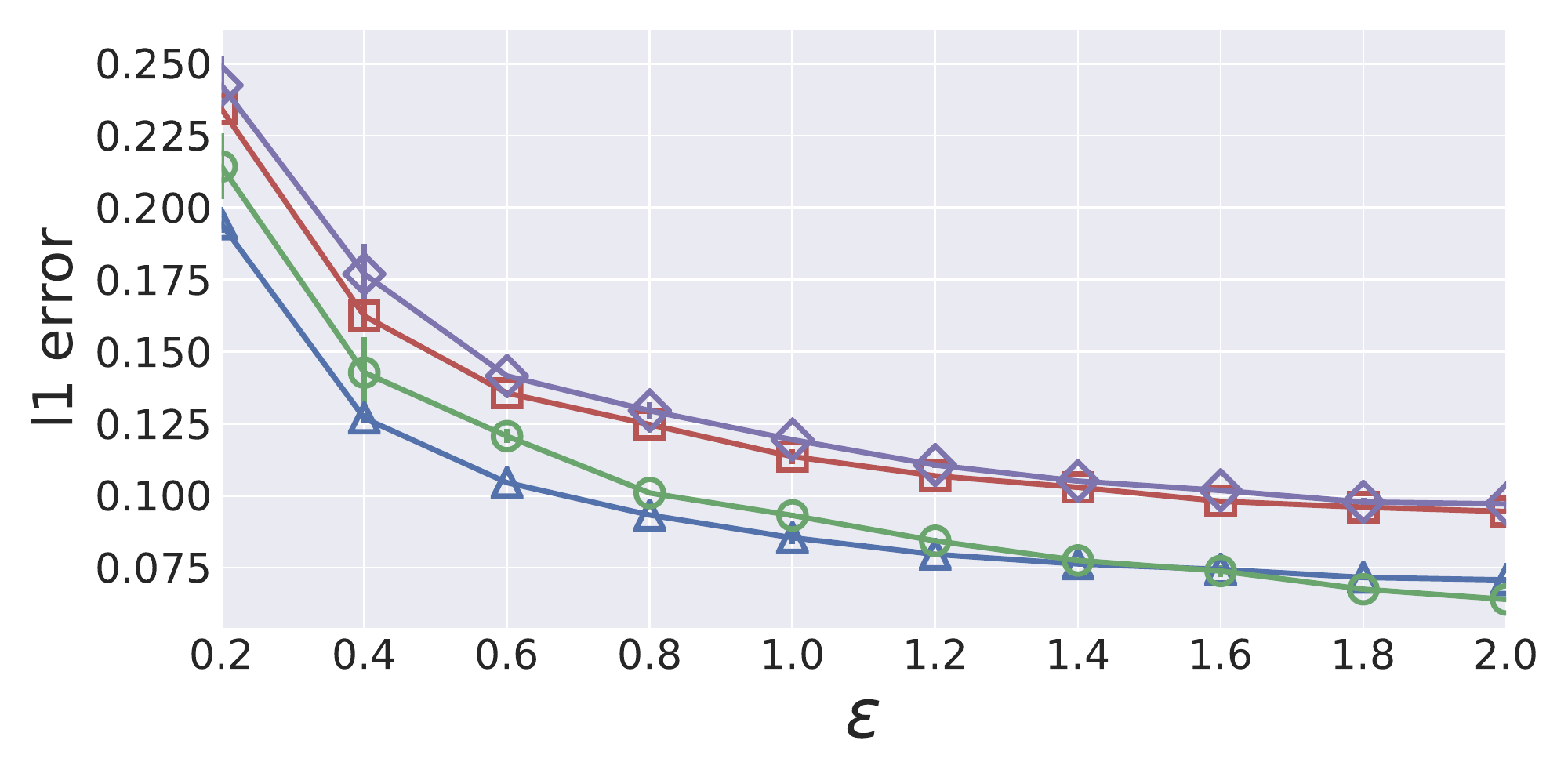}}
    \subfloat[range query]{\includegraphics[width=0.3\textwidth]{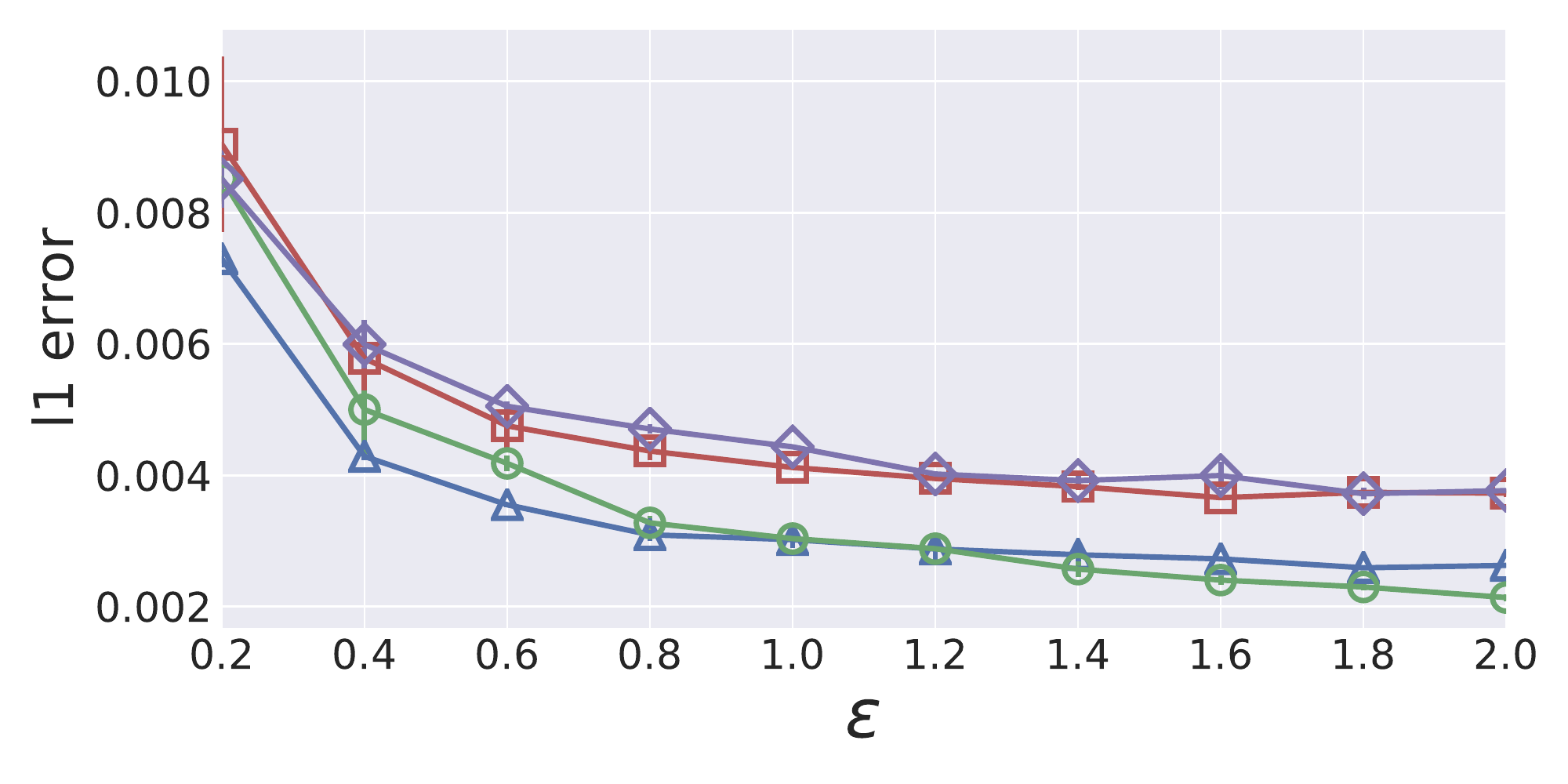}}
    \subfloat[classification]{\includegraphics[width=0.3\textwidth]{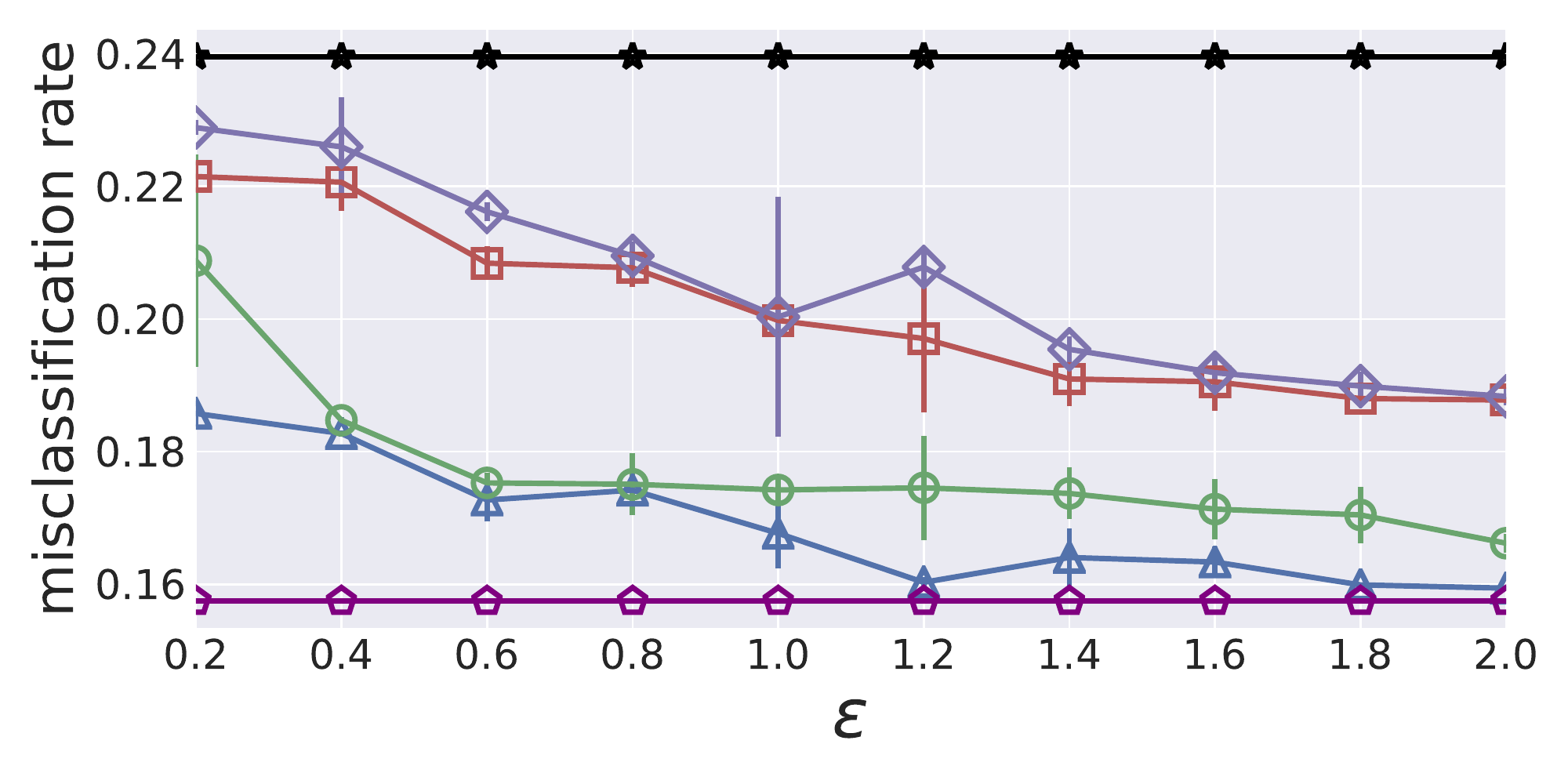}} \\ [1ex]
    
    \subfloat{\includegraphics[width=0.8\textwidth]{Figures/comparison_marginal_legend.pdf}}  \\ [-2ex]
    
    \subfloat[pair-wise marginal]{\includegraphics[width=0.3\textwidth]{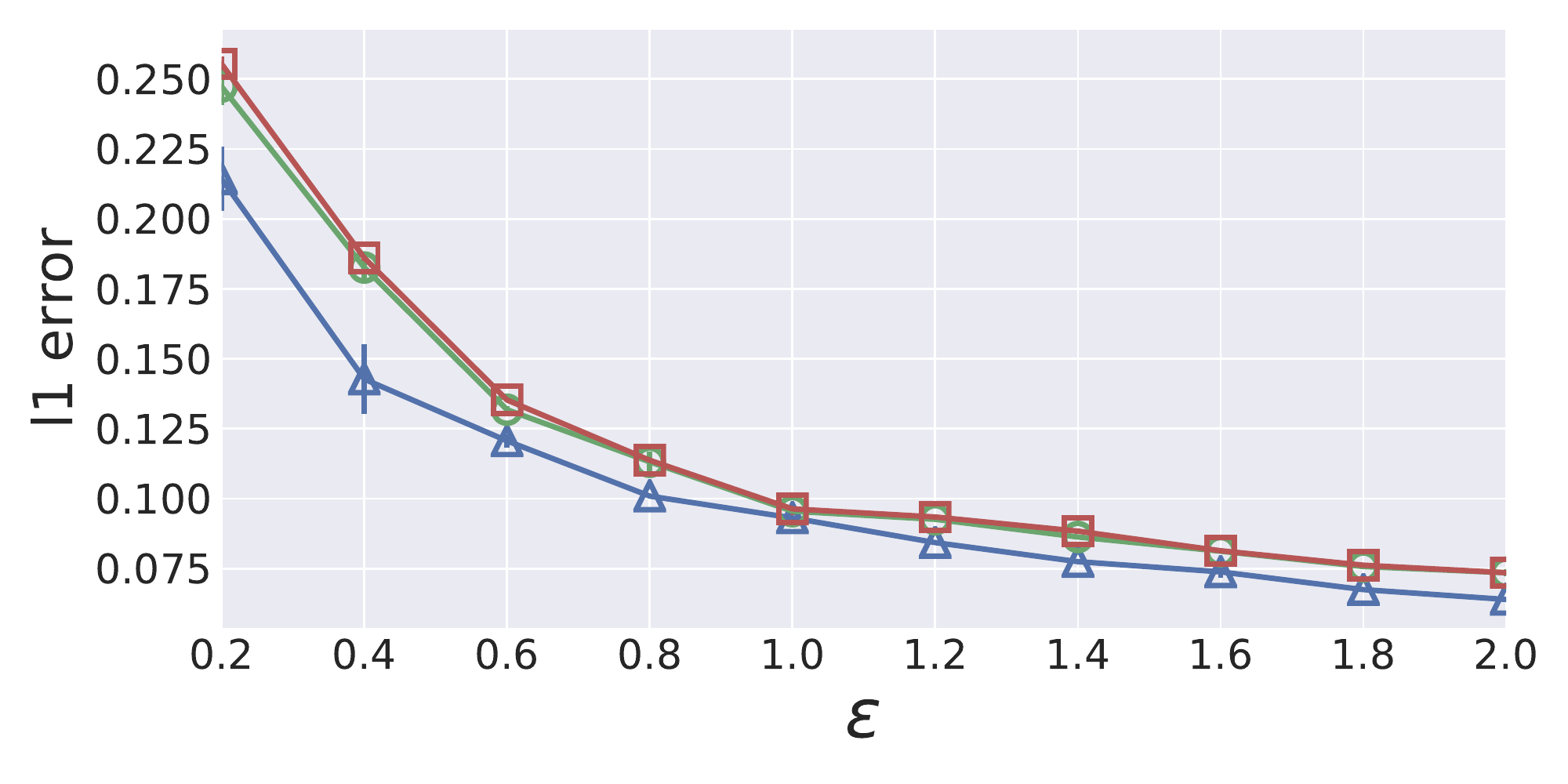}}
    \subfloat[range query]{\includegraphics[width=0.3\textwidth]{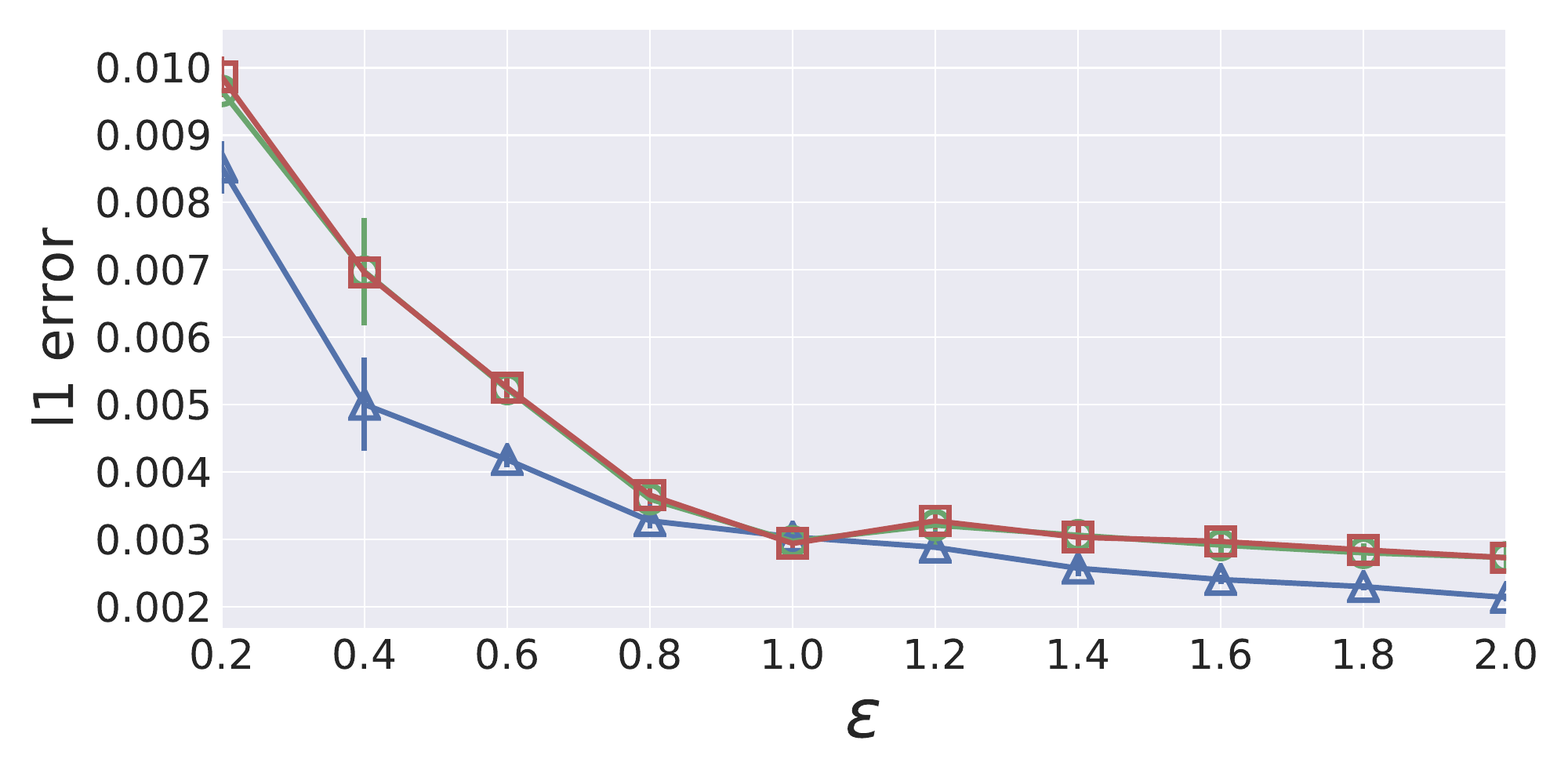}}
    \subfloat[classification]{\includegraphics[width=0.3\textwidth]{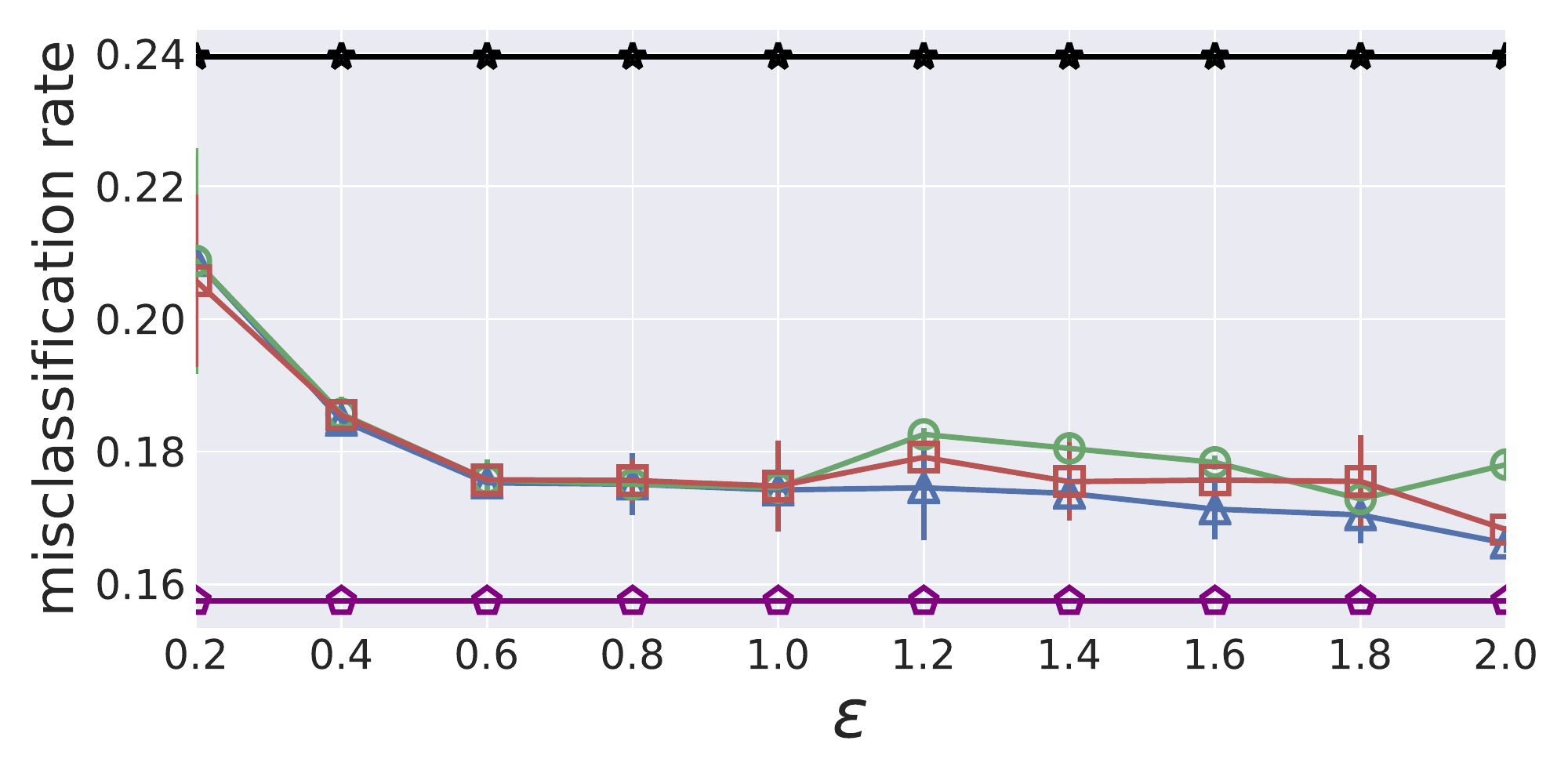}} \\ [1ex]
    
    \subfloat{\includegraphics[width=0.8\textwidth]{Figures/comparison_noise_add_legend.pdf}}  \\ [-2ex]
    
    \subfloat[pair-wise marginal]{\includegraphics[width=0.3\textwidth]{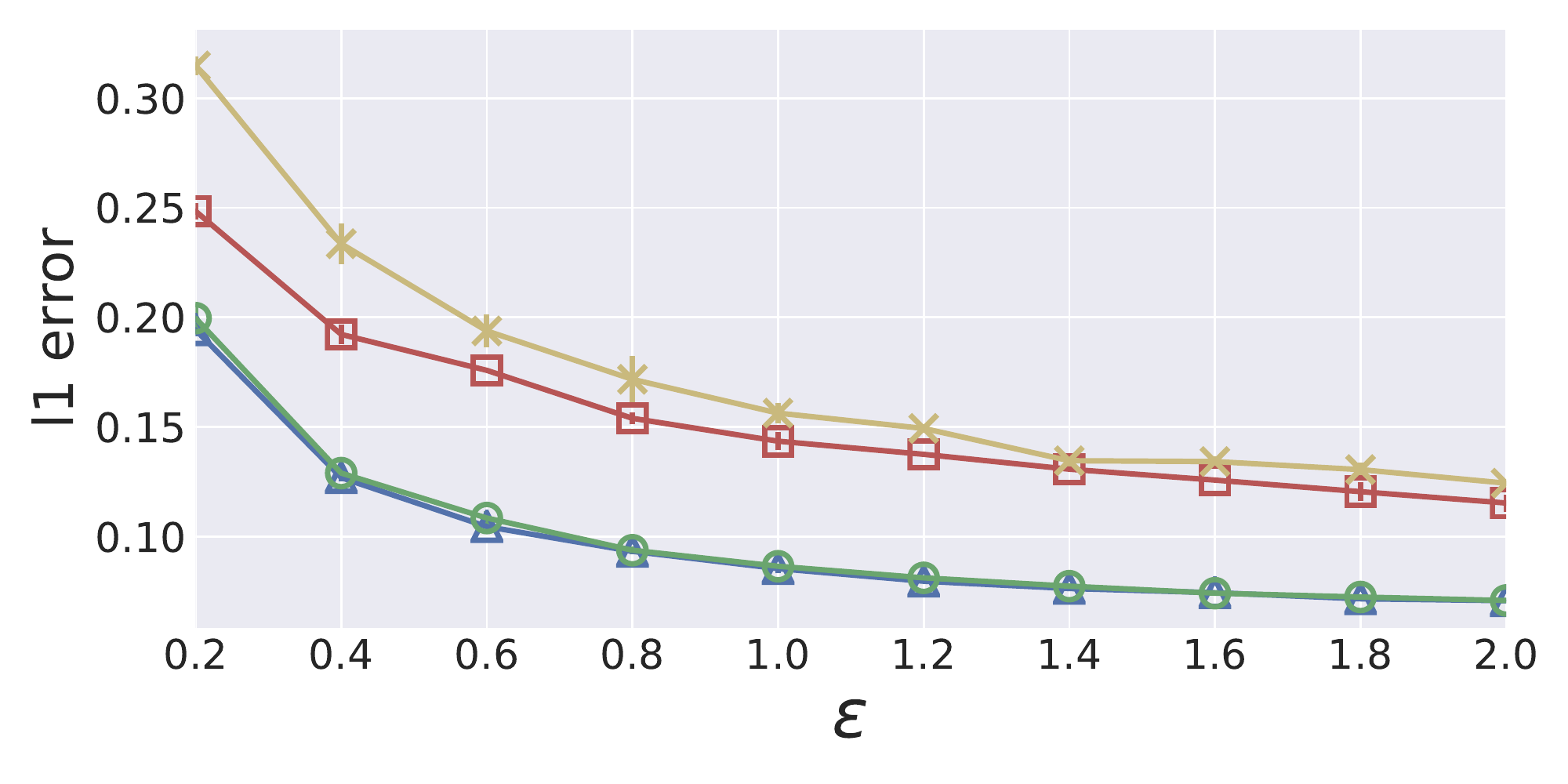}}
    \subfloat[range query]{\includegraphics[width=0.3\textwidth]{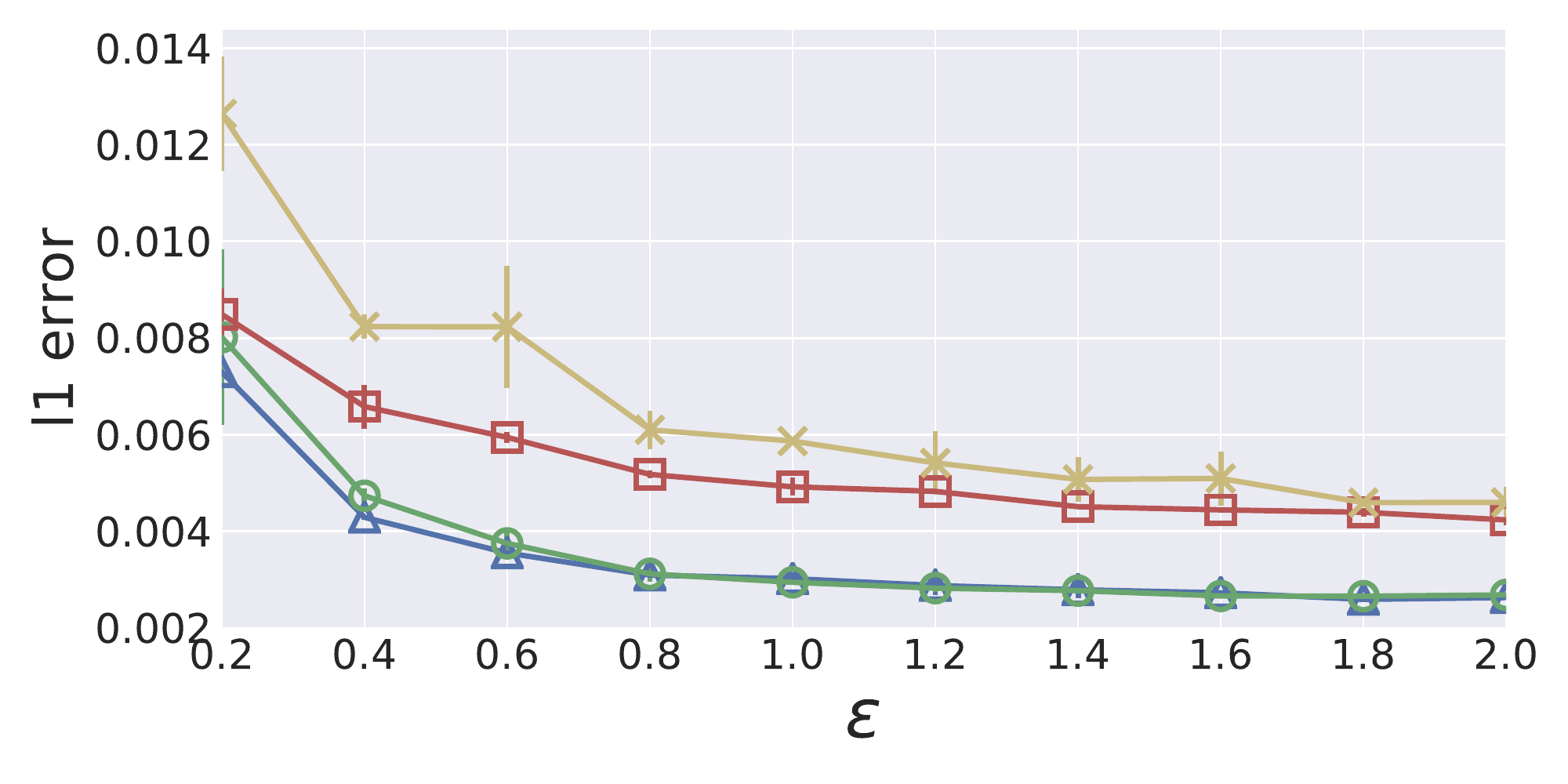}}
    \subfloat[classification]{\includegraphics[width=0.3\textwidth]{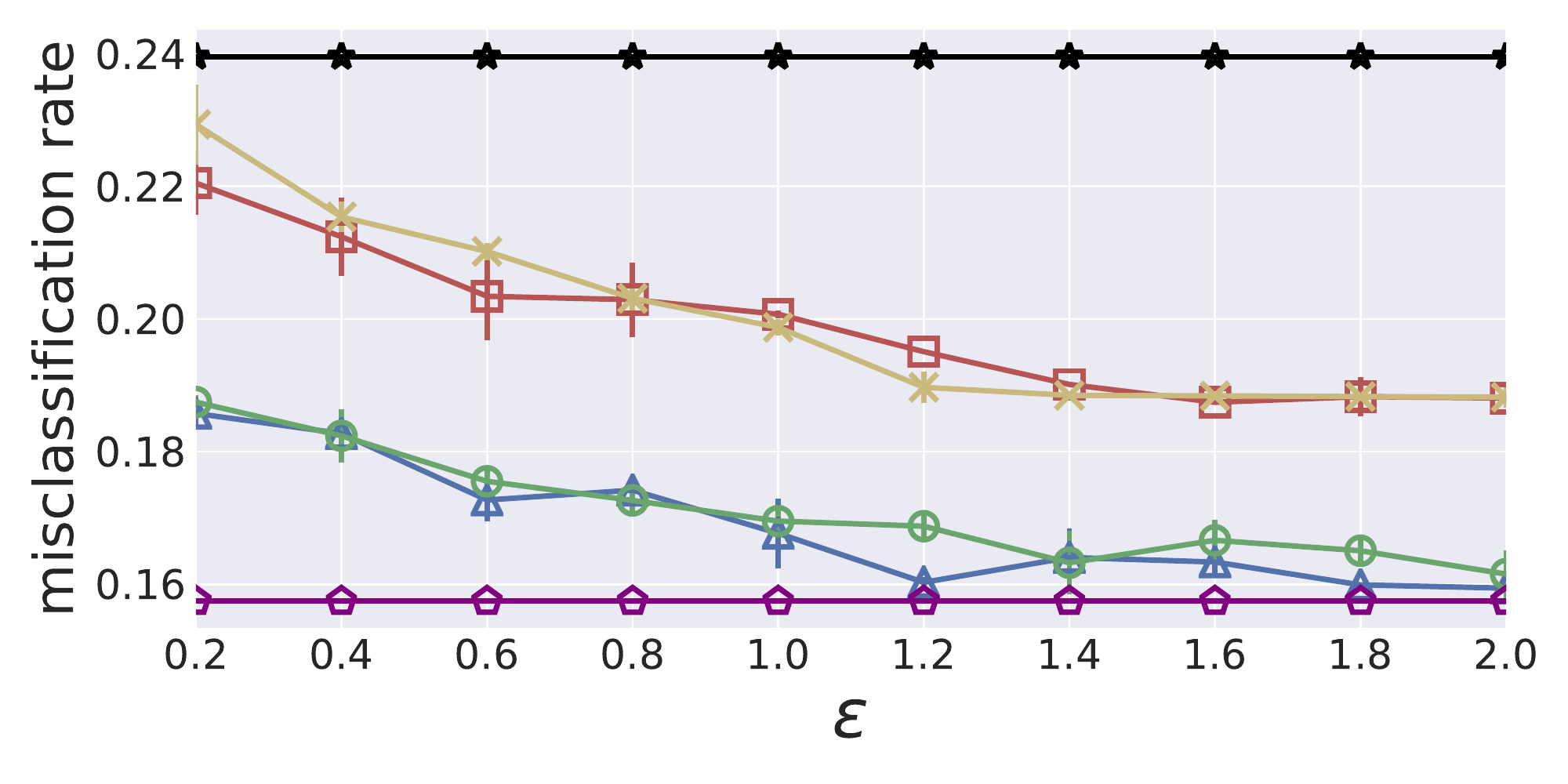}} \\ [1ex]
    
    \subfloat{\includegraphics[width=0.8\textwidth]{Figures/comparison_synthesis_legend.pdf}}  \\ 
    
    \caption{\revision{}{Experimental results for Adult dataset.
    The first row is the end-to-end comparison,
    the second row is the marginal selection methods comparison,
    the second row is the noise addition methods comparison,
    the last row is the synthesis methods comparison.}
    }
    \label{fig:adult_results}
\end{figure*}

\begin{figure*}[!htpb]
    \centering
    
    \subfloat[pair-wise marginal]{\includegraphics[width=0.3\textwidth]{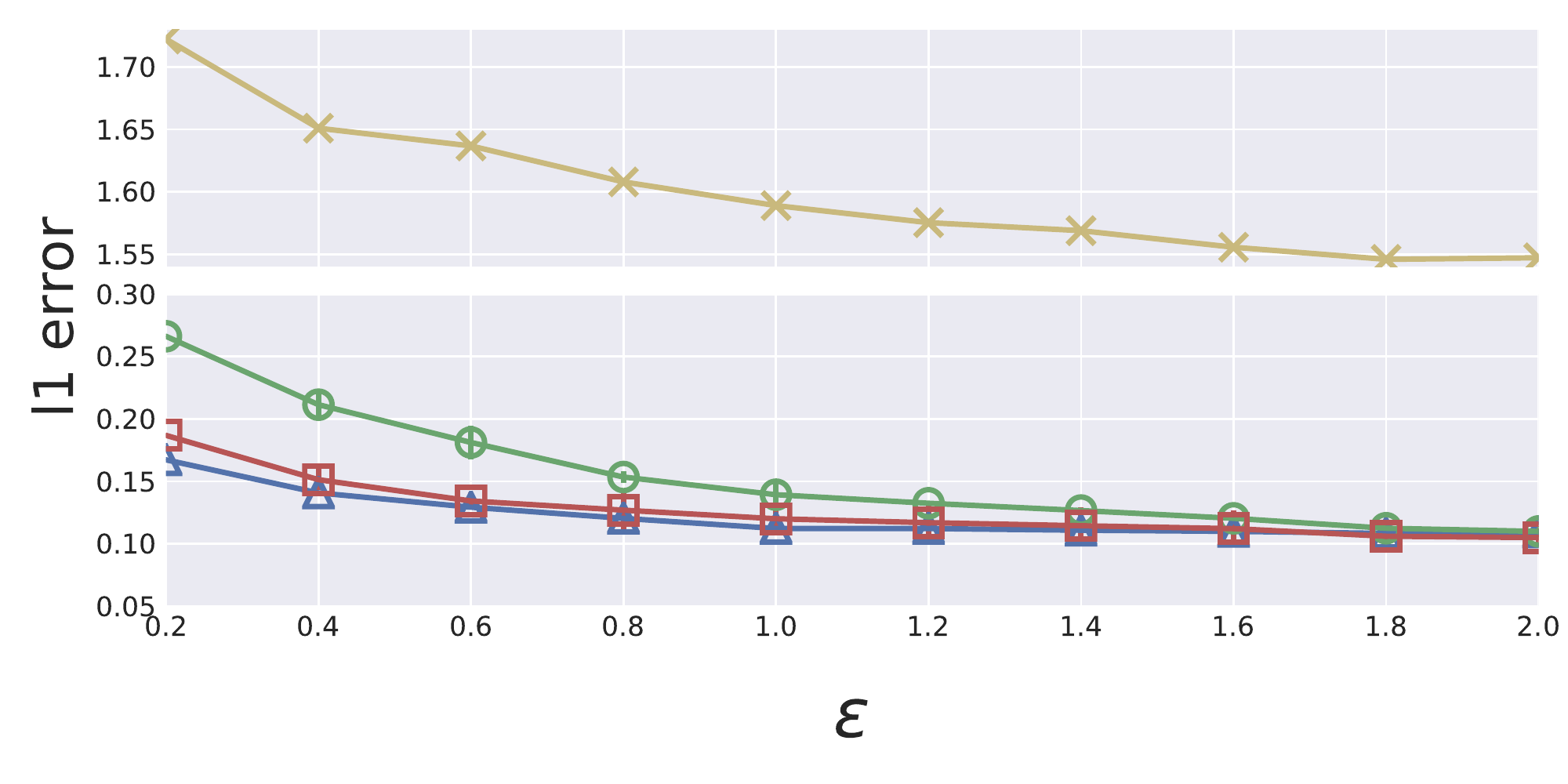}}
    \subfloat[range query]{\includegraphics[width=0.3\textwidth]{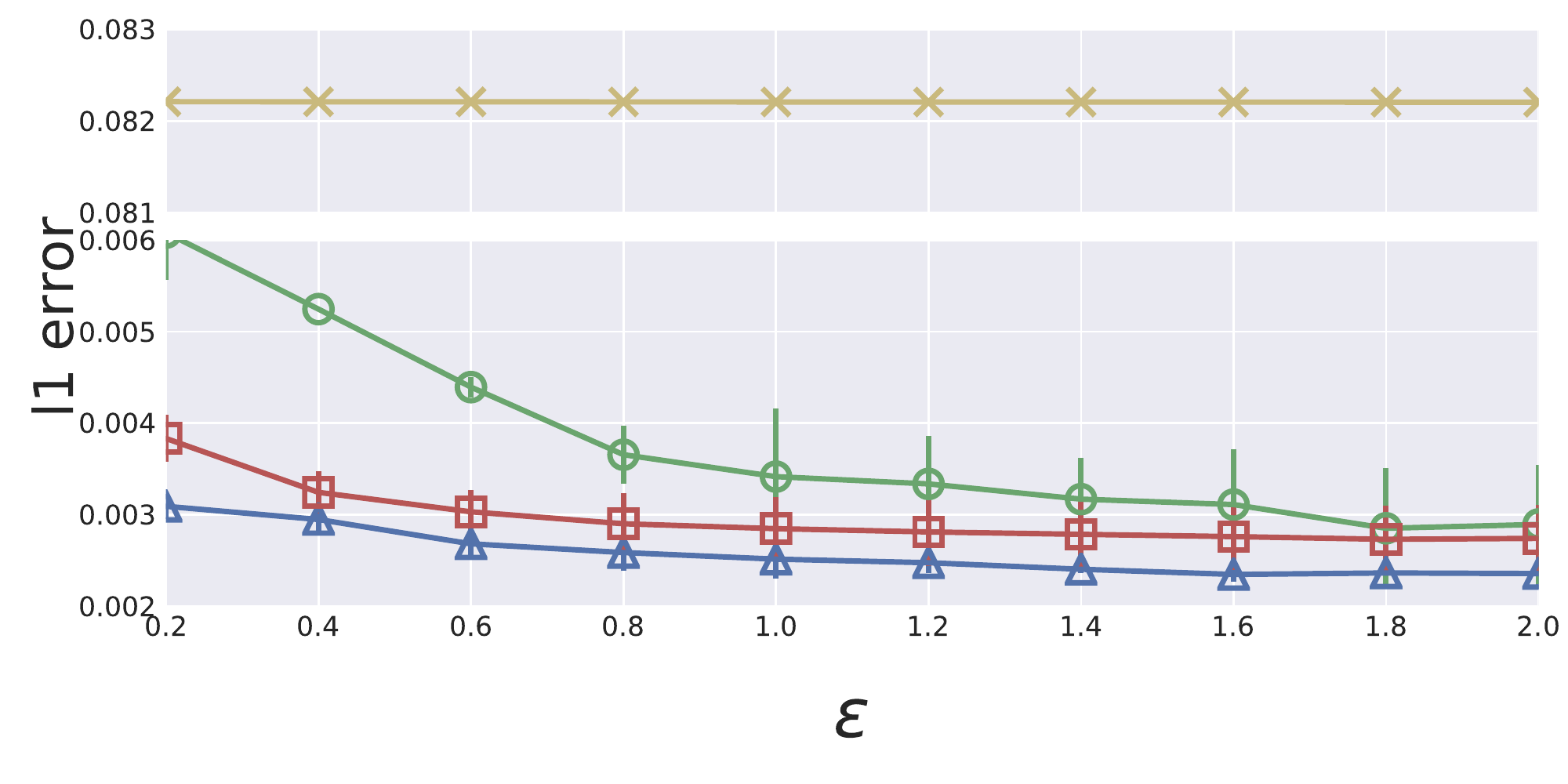}}
    \subfloat[classification]{\includegraphics[width=0.3\textwidth]{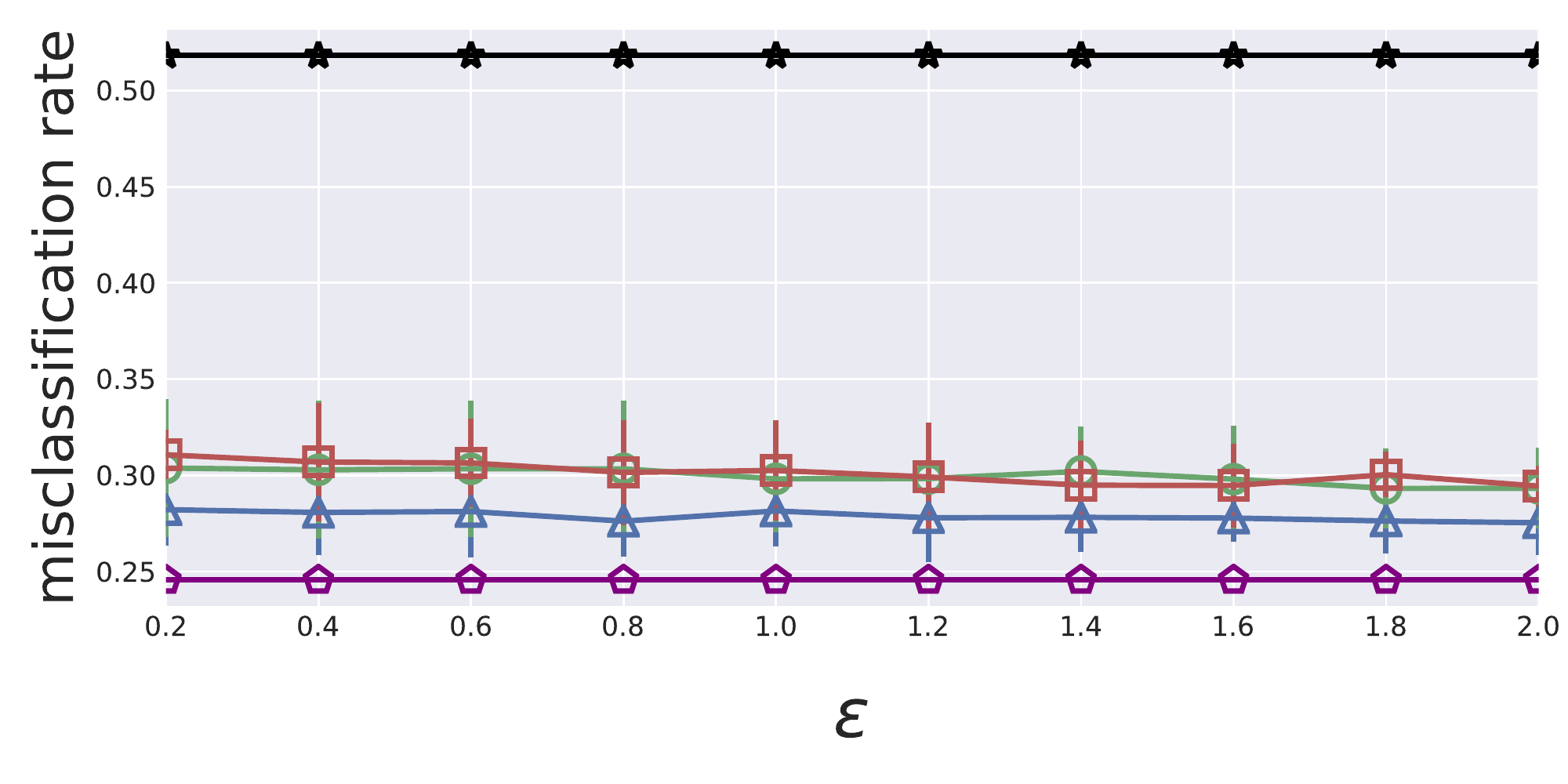}} \\ [1ex]
    
    \subfloat{\includegraphics[width=0.8\textwidth]{Figures/comparison_e2e_legend.pdf}}  \\ [-2ex]

    \subfloat[pair-wise marginal]{\includegraphics[width=0.3\textwidth]{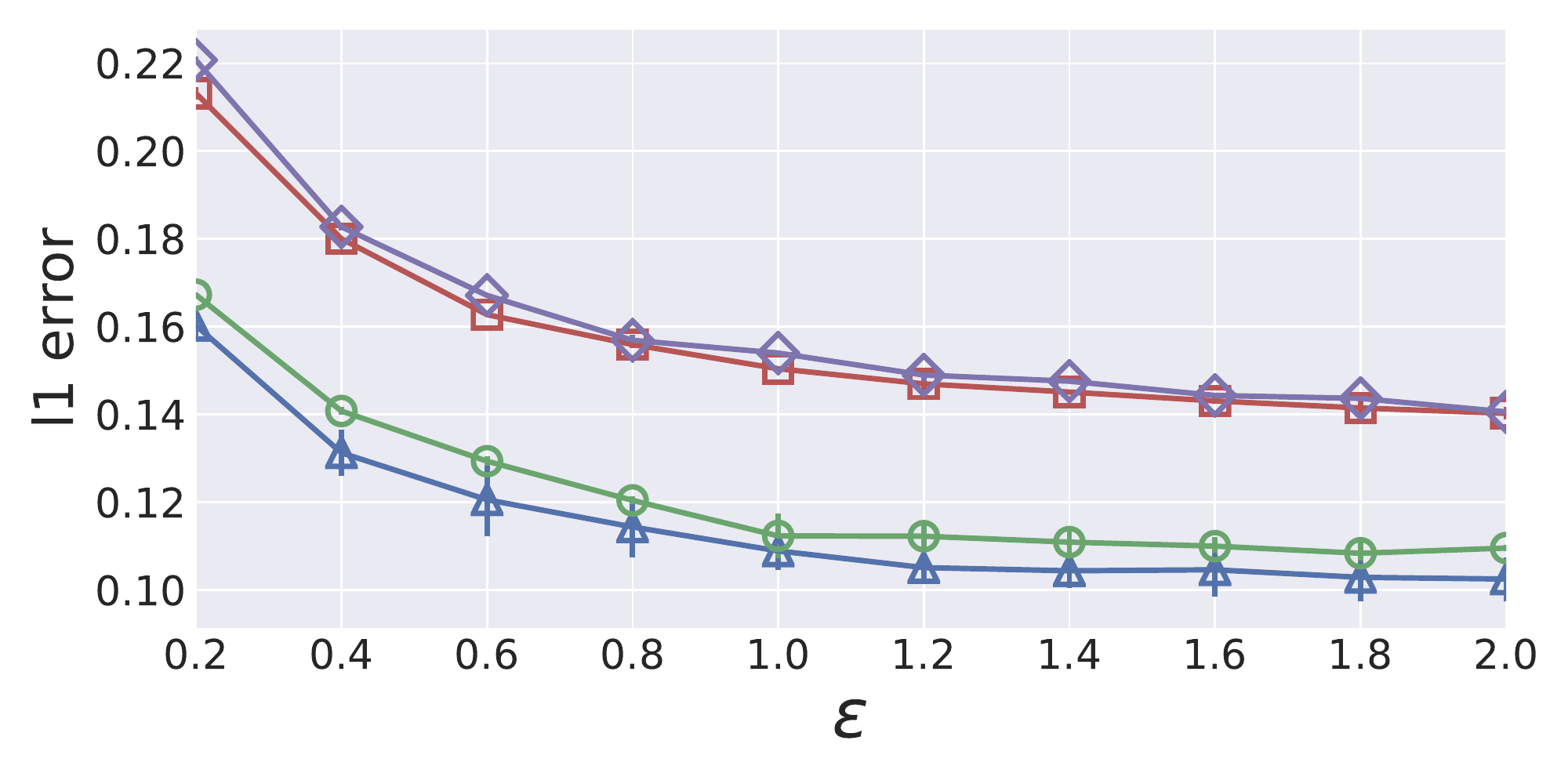}}
    \subfloat[range query]{\includegraphics[width=0.3\textwidth]{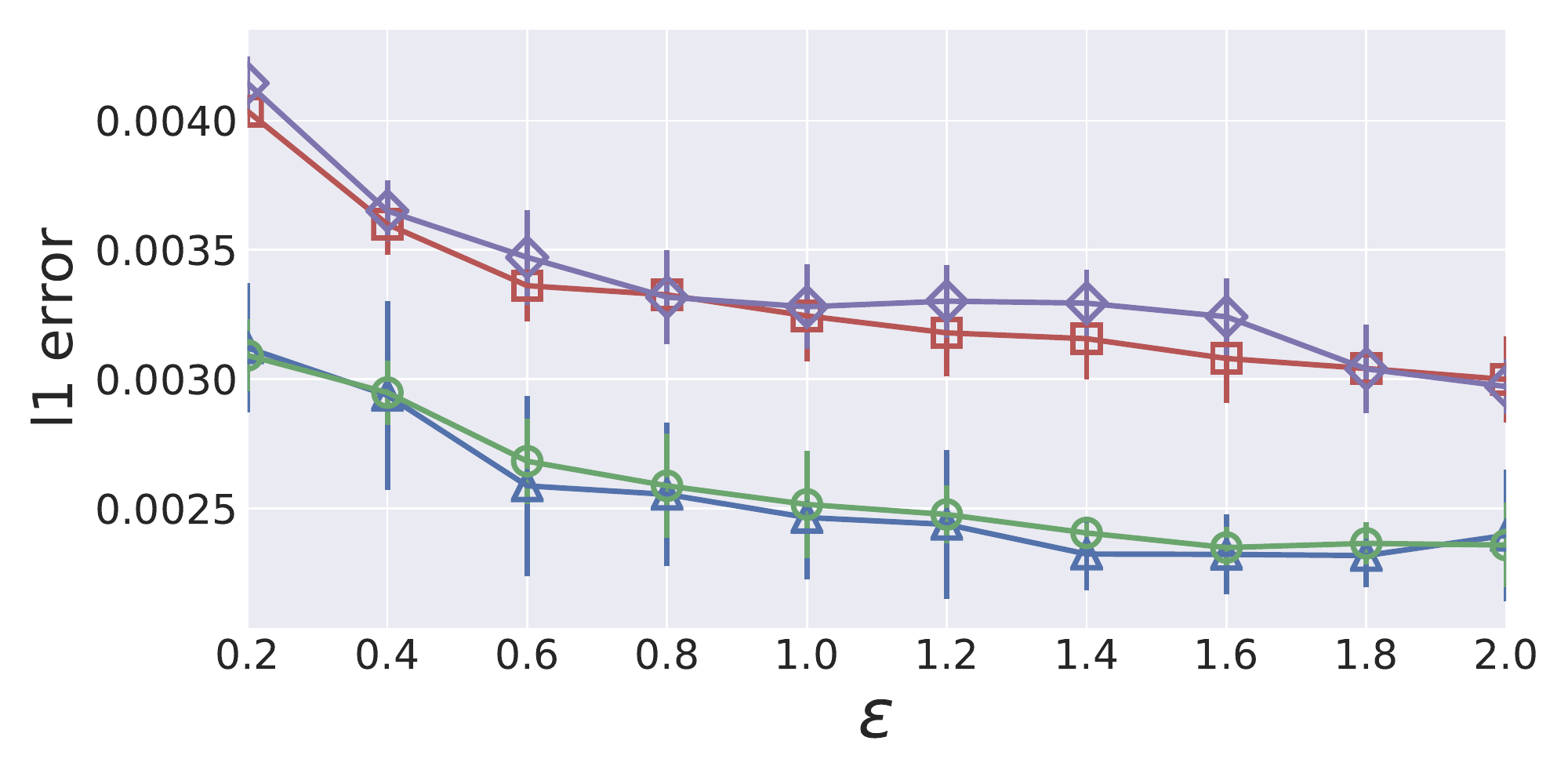}}
    \subfloat[classification]{\includegraphics[width=0.3\textwidth]{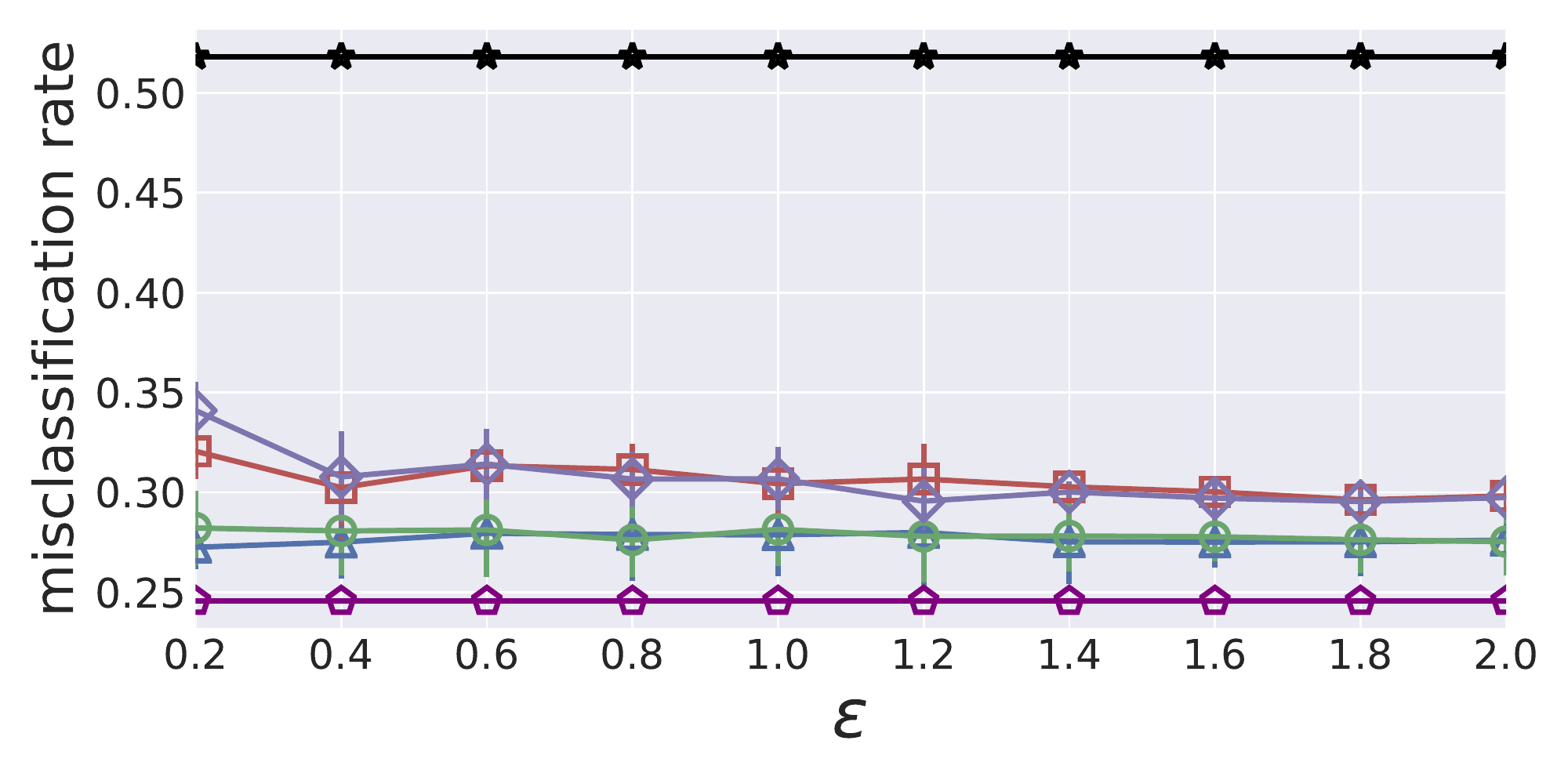}} \\ [1ex]
    
    \subfloat{\includegraphics[width=0.8\textwidth]{Figures/comparison_marginal_legend.pdf}}  \\ [-2ex]
    
    \subfloat[pair-wise marginal]{\includegraphics[width=0.3\textwidth]{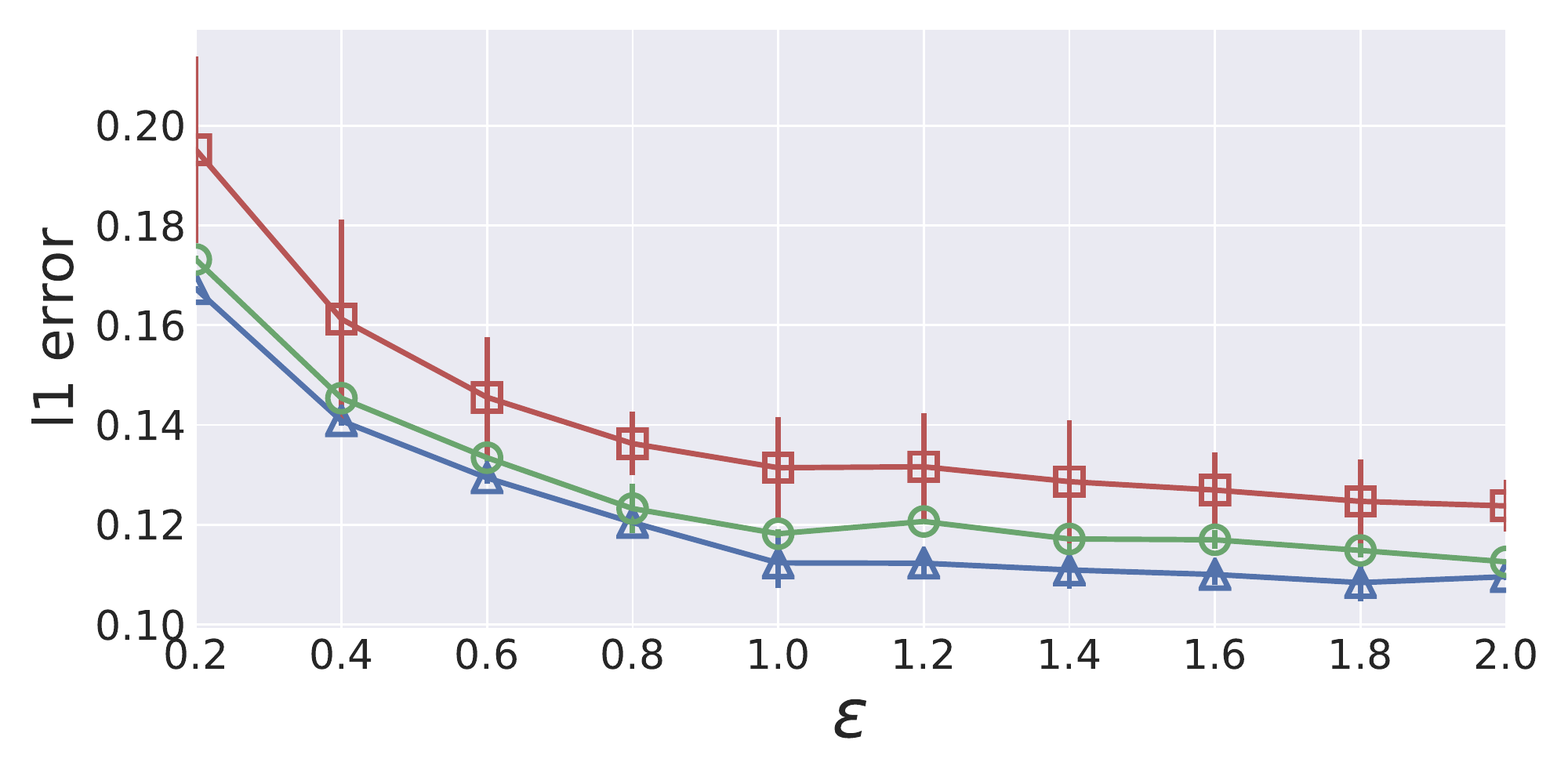}}
    \subfloat[range query]{\includegraphics[width=0.3\textwidth]{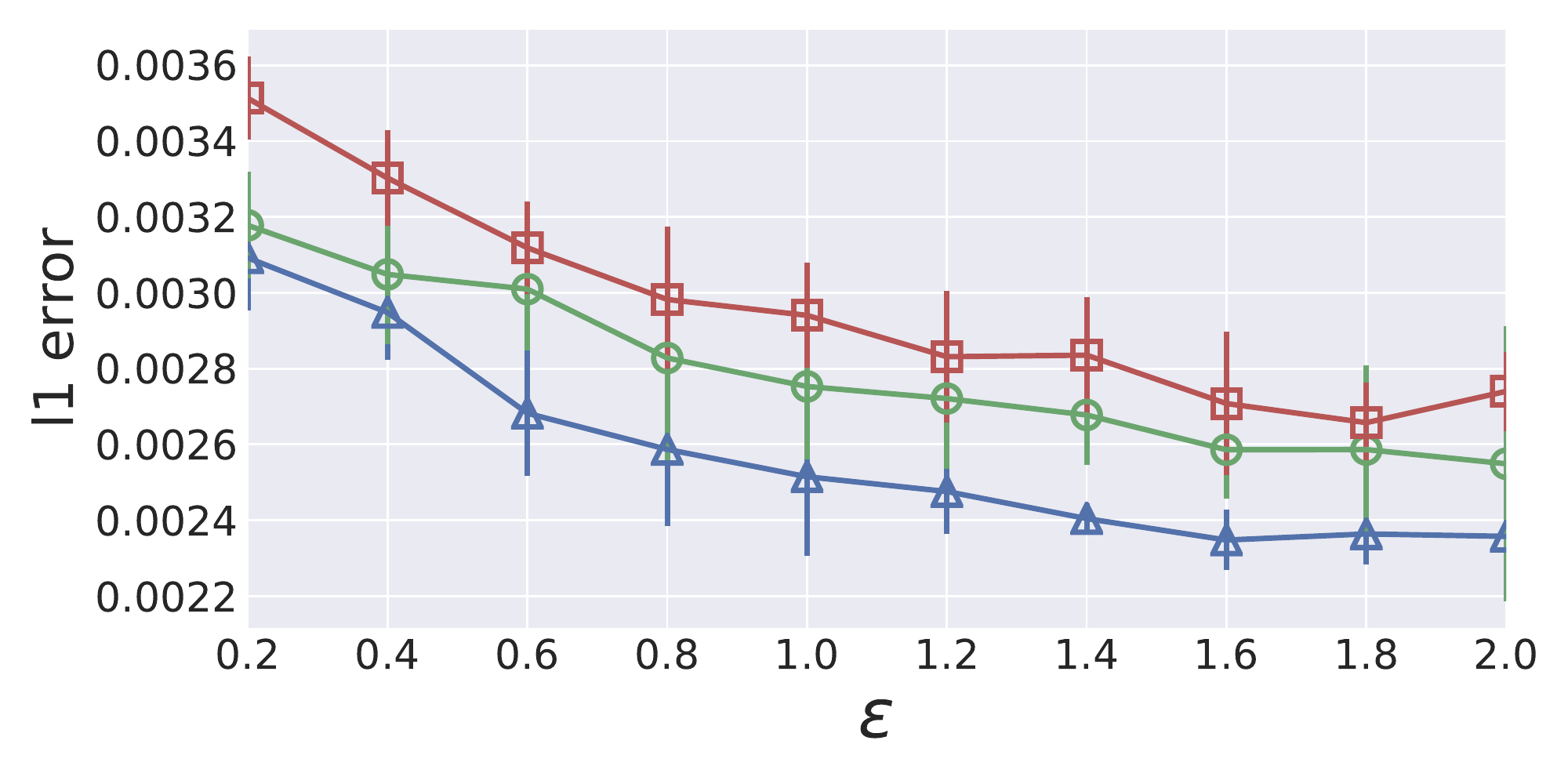}}
    \subfloat[classification]{\includegraphics[width=0.3\textwidth]{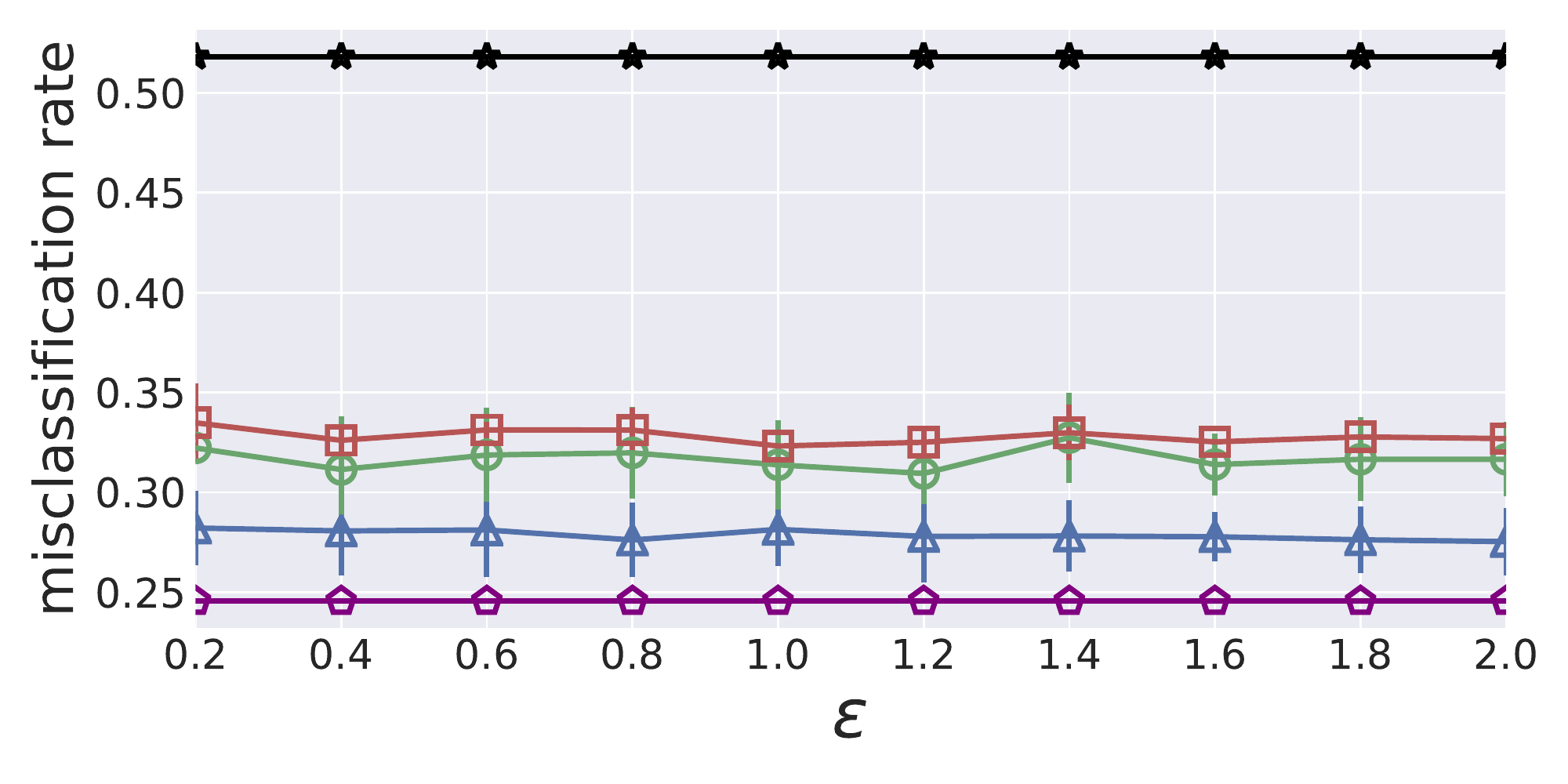}} \\ [1ex]
    
    \subfloat{\includegraphics[width=0.8\textwidth]{Figures/comparison_noise_add_legend.pdf}}  \\ [-2ex]
    
    \subfloat[pair-wise marginal]{\includegraphics[width=0.3\textwidth]{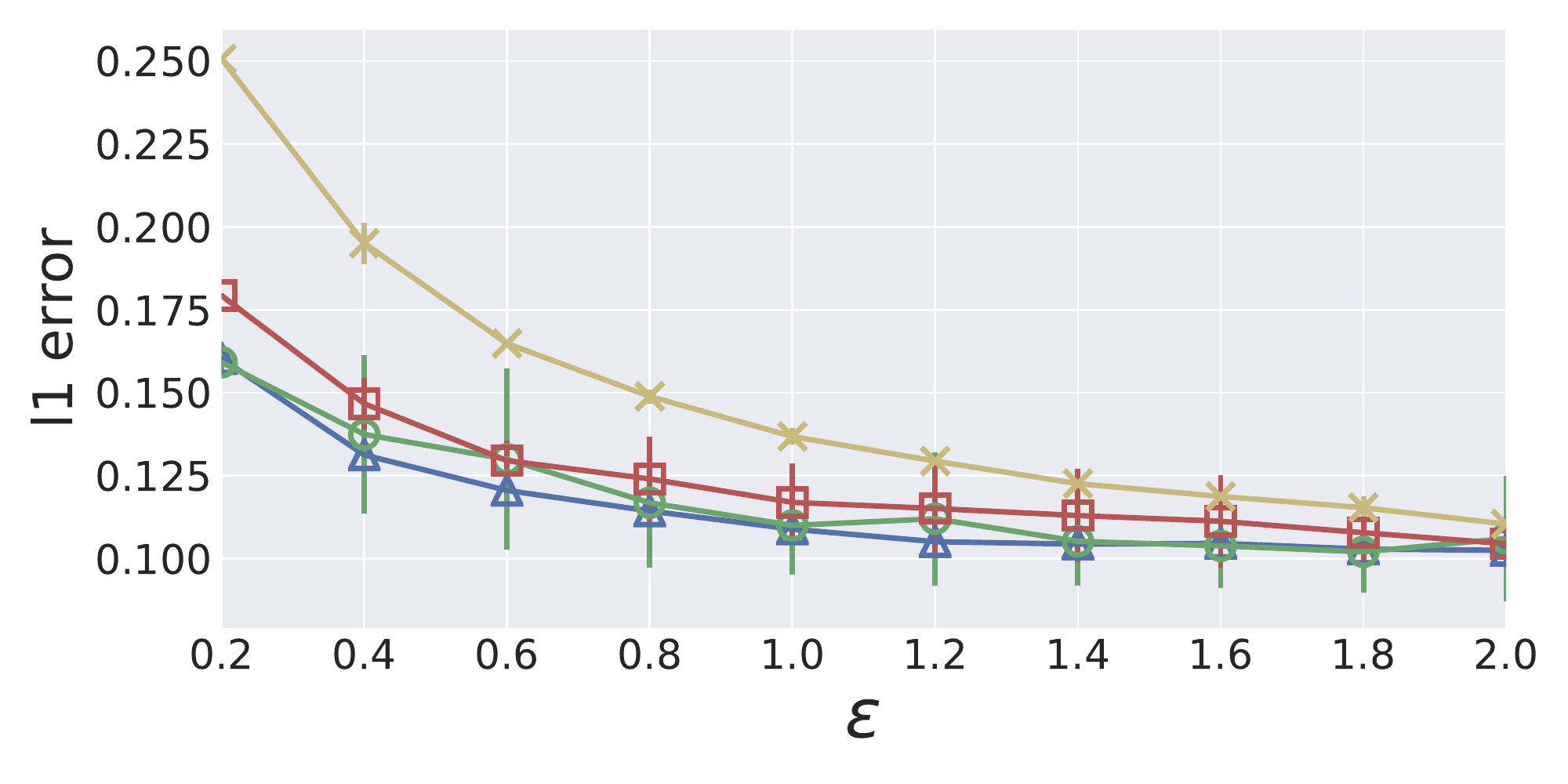}}
    \subfloat[range query]{\includegraphics[width=0.3\textwidth]{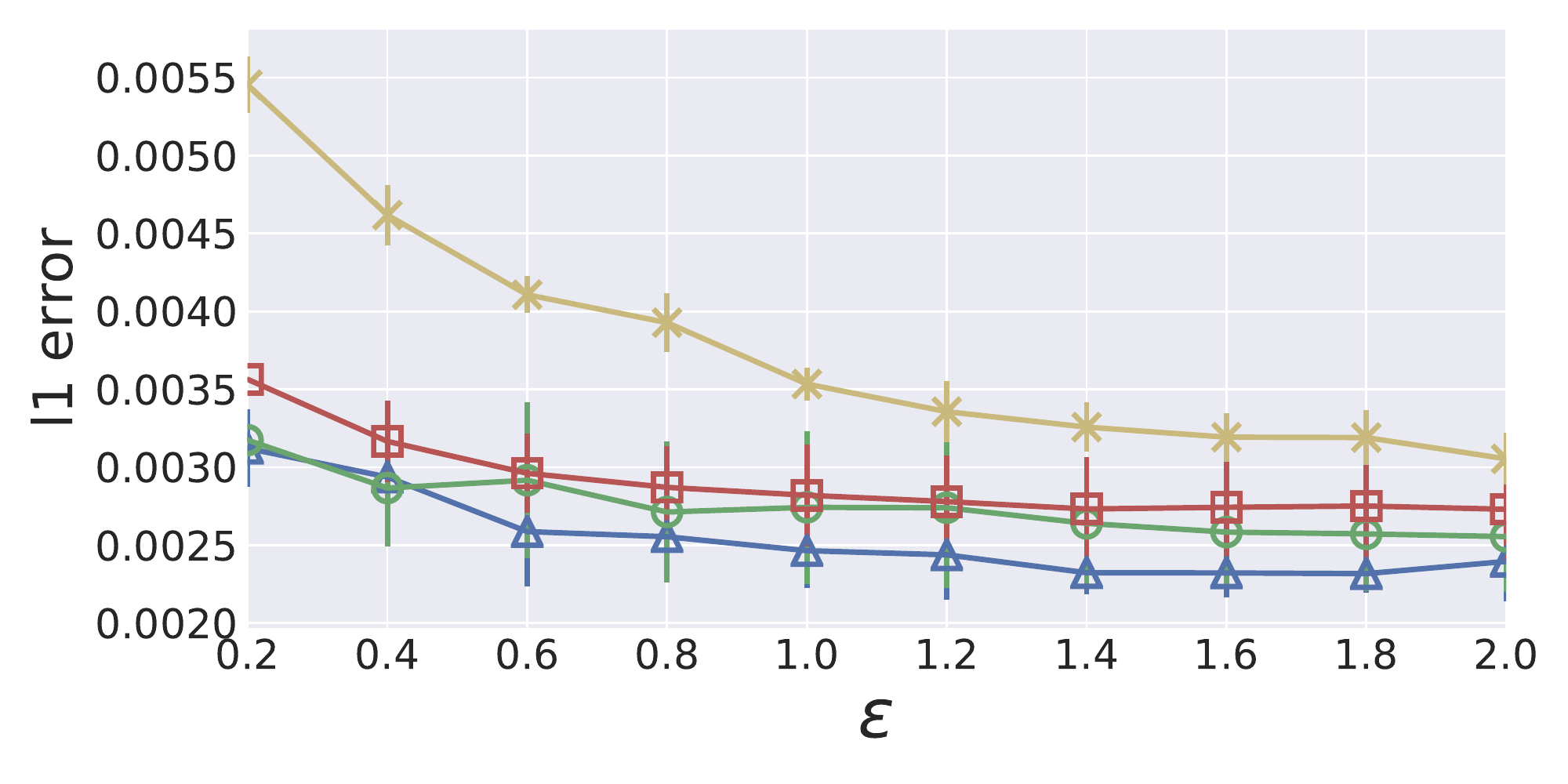}}
    \subfloat[classification]{\includegraphics[width=0.3\textwidth]{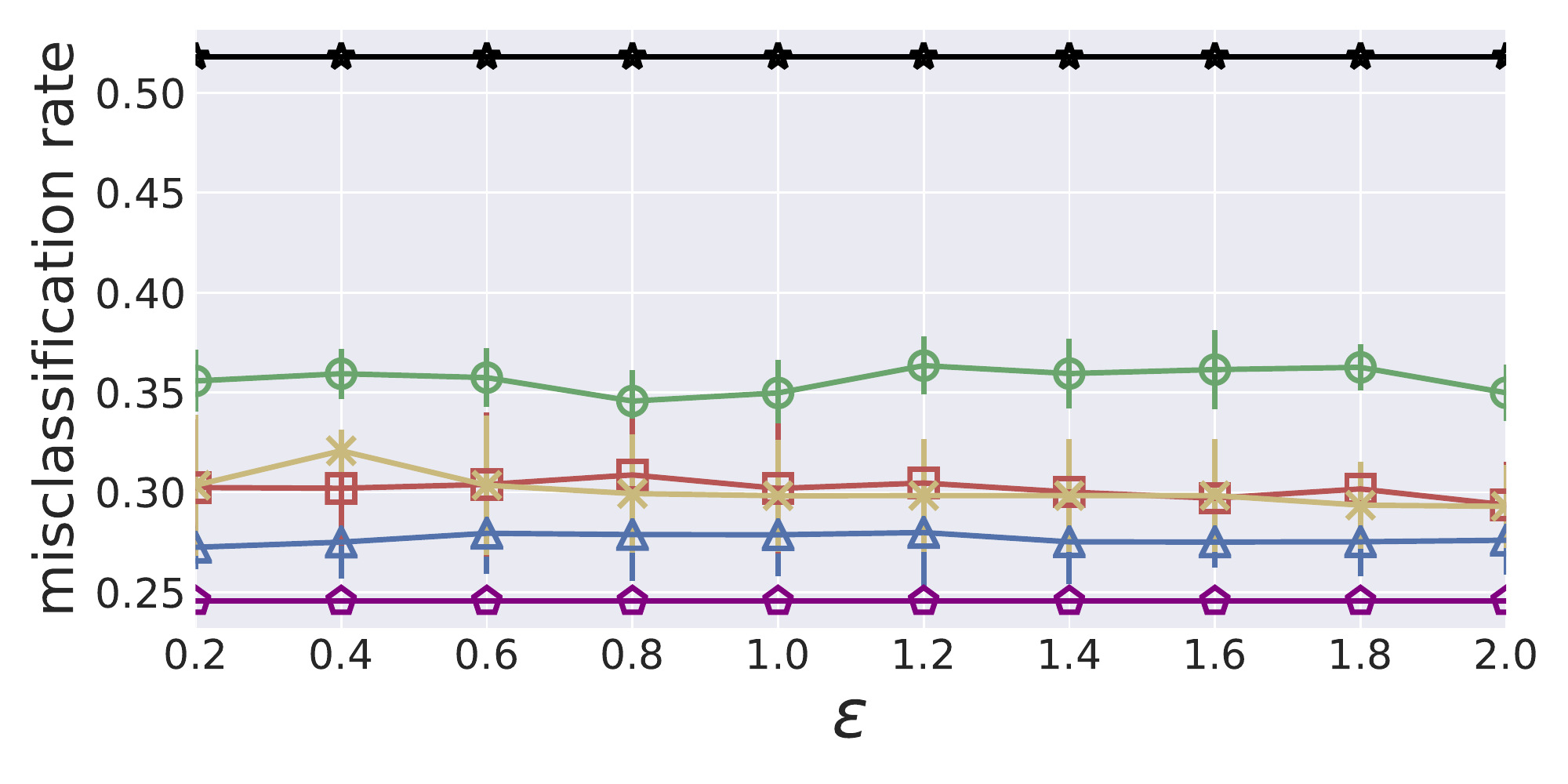}} \\ [1ex]
    
    \subfloat{\includegraphics[width=0.8\textwidth]{Figures/comparison_synthesis_legend.pdf}}  \\ 
    
    \caption{Experimental results for Loan dataset.
    The first row is the end-to-end comparison,
    the second row is the marginal selection methods comparison,
    the second row is the noise addition methods comparison,
    the last row is the synthesis methods comparison.
    }
    \label{fig:loan_results}
\end{figure*}

\end{document}